\documentclass[preprint,11pt]{article}

\usepackage{amssymb,amsfonts,amsmath,amsthm,amscd,dsfont,mathrsfs}
\usepackage{graphicx,float,psfrag,epsfig}
\usepackage{wrapfig}
\usepackage{relsize}
\usepackage{hyperref}
\usepackage{color}
\usepackage{pict2e}
\usepackage[tight]{subfigure}

\DeclareMathAlphabet{\mathpzc}{OT1}{pzc}{m}{it}

\newtheorem{propo}{Proposition}[section]
\newtheorem{lemma}[propo]{Lemma}
\newtheorem{definition}[propo]{Definition}
\newtheorem{coro}[propo]{Corollary}
\newtheorem{thm}[propo]{Theorem}

\newtheorem{remark}[propo]{Remark}
\newtheorem{claim}[propo]{Claim}

\def\tx{{\widetilde{x}}}
\def\tr{{\widetilde{r}}}
\def\ty{{\widetilde{y}}}
\def\Embed{{\sf H}}

\def\Restr{{\sf H}'}

\def\cA{{\cal A}}

\def\cS{{\cal S}}

\def\hx{{\widehat{x}}}
\def\naturals{{\mathbb N}}
\def\integers{{\mathbb Z}}
\def\reals{{\mathbb R}}

\def\ve{{\varepsilon}}
\def\eps{{\epsilon}}
\def\prob{{\mathbb P}}
\def\E{{\mathbb E}}
\def\MSE{{\sf MSE}}
\def\MSEexp{\overline{\sf MSE}}
\def\MSEAMP{{\sf MSE}_{\rm AMP}}
\def\MSEAMPexp{\overline{\sf MSE}_{\rm AMP}}
\def\ons{{\sf b}}
\def\uRenyi{\overline{d}}
\def\lRenyi{\underline{d}}
\def\Renyi{d}
\def\MMSE{D}
\def\uMMSE{\overline{D}}
\def\lMMSE{\underline{D}}
\def\MMSEd{D}
\def\Ens{{\cal M}}
\def\gr{{\sf g}}
\def\mmse{{\sf mmse}}
\def\Info{{\sf I}}
\def\normal{{\sf N}}
\def\SEmap{{\sf T}}
\def\SEmapa{{\sf T}'}
\def\SEmapb{{\sf T}''}

\def\Map{{\sf F}}
\def\Mapa{{\sf F}'}
\def\Mapb{{\sf F}''}

\def\CMapa{{\mathcal F}'}
\def\CMapb{{\mathcal F}''}

\def\realsc{\overline{\mathbb R}_+}
\def\tQ{\widetilde{Q}}
\def\Shape{{\cal W}}

\def\tW{{\widetilde{W}}}
\def\Energy{{\sf E}}
\def\tEnergy{{\tilde{\sf E}}}

\def\Lr{{L_r}}
\def\Lc{{L_c}}
\def\L0{{L_0}}

\def\de{{\rm d}}
\def\<{\langle}
\def\>{\rangle}

\def\rob{{\rm rob}}

\def\Rows{{\sf R}}
\def\Cols{{\sf C}}

\def\limsup{\rm{lim \, sup}}
\def\liminf{\rm{lim \, inf}}

\def \px{p_{X}}
\def\tX{{\widetilde{X}}}
\def \ptx{p_{\widetilde{X}}}
\def \teta{{\tilde{\eta}}}

\def\Dks{D_{\rm KS}}

\definecolor{White}{rgb}{1,1,1}
\def\white{\color{White}}

\begin{document}

\title{Information-Theoretically Optimal Compressed Sensing via\\
Spatial Coupling and Approximate Message Passing}

\author{David L. Donoho${}^{\dagger}$, Adel Javanmard${}^{*}$ and Andrea~Montanari 
            \footnote{Department of Electrical Engineering, Stanford University}
            \footnote{Department of Statistics, Stanford
              University\newline
This work was presented in part at the 2012 IEEE International
Symposium on Information Theory}}

\maketitle
%
%
\begin{abstract}
We study the compressed sensing reconstruction problem for a broad
class of random, band-diagonal sensing matrices. This construction is
inspired by the idea of spatial coupling in coding theory. As demonstrated
heuristically and numerically by Krzakala et al. \cite{KrzakalaEtAl}, message passing
algorithms can effectively solve the reconstruction problem for
spatially coupled measurements with undersampling rates close to 
the fraction of non-zero coordinates. 

We use an approximate message passing (AMP) algorithm  and analyze it 
through the state evolution method. We give a rigorous proof that this
approach is successful as soon as the undersampling rate $\delta$ exceeds the 
(upper) R\'enyi information dimension of the signal, $\uRenyi(p_X)$. More
precisely, for a sequence of signals of diverging dimension $n$ whose empirical distribution
converges to $p_X$, reconstruction is with high probability successful
from $\uRenyi(p_X)\, n+o(n)$ measurements taken according to a band
diagonal matrix.

For sparse signals, i.e., sequences of dimension $n$ and $k(n)$
non-zero entries, this implies reconstruction from $k(n)+o(n)$
measurements. For `discrete' signals, i.e., signals whose coordinates
take a fixed finite set of values, this implies reconstruction from
$o(n)$ measurements. The result is robust with respect to noise, does
not apply uniquely to random signals, but requires the knowledge of
the empirical distribution of the signal $p_X$.
\end{abstract}

%
%
\section{Introduction and main results}
\label{sec:Introduction}

\subsection{Background and contributions}

Assume that $m$ linear measurements are taken of an unknown $n$-dimensional
signal $x\in\reals^n$, according to the model
\begin{eqnarray}
y = Ax\, .\label{eq:LinearSensing}
\end{eqnarray}
 The reconstruction problem requires to
reconstruct $x$ from the measured vector $y\in\reals^m$, and the
measurement matrix $A\in\reals^{m\times n}$.

It is an elementary fact of linear algebra that the reconstruction problem will not
have a unique solution unless $m\ge n$. This observation is however
challenged within compressed sensing. A large corpus of research shows
that, under the assumption that $x$ is sparse, a dramatically smaller
number of measurements is sufficient
\cite{Donoho1,CandesFourier,Do}. Namely,
if only $k$ entries of $x$ are non-vanishing, then roughly $m\gtrsim
2k\log(n/k)$ measurements are sufficient for $A$ random, and
reconstruction can be solved efficiently by convex programming. Deterministic
sensing matrices achieve similar performances, provided they satisfy
a suitable restricted isometry condition \cite{CandesTao}.
On top of this, reconstruction is robust with respect to the addition of noise
\cite{CandesStable,DMM-NSPT-11}, i.e., under the model
\begin{eqnarray}
y = Ax+w\, ,\label{eq:NoisyModel}
\end{eqnarray}
with, say, $w\in\reals^{m}$ a random vector with i.i.d. components
$w_i\sim\normal(0,\sigma^2)$ (unless stated otherwise, $\sigma=0$ is a
valid choice). In this context, the notions of `robustness' or `stability'
refers to  the existence of universal constants $C$ such that the per-coordinate mean
square error in reconstructing $x$ from noisy observation $y$ is upper
bounded by $C\,\sigma^2$.

From an information-theoretic point of view it remains however unclear
why we cannot achieve the same goal with far fewer
than $2\,k\log(n/k)$ measurements. Indeed, we can interpret
Eq.~(\ref{eq:LinearSensing}) as describing an analog data compression
process, with $y$ a compressed version of $x$.  From this point of
view, we can encode all the information about $x$
in a single real number $y\in\reals$ (i.e., use $m=1$), because the cardinality of
$\reals$ is the same as the one of $\reals^n$. Motivated by this
puzzling remark, Wu and Verd\'u \cite{WuVerdu} introduced a Shannon-theoretic
analogue of compressed sensing, whereby the vector $x$ has
i.i.d. components $x_i\sim p_X$. Crucially, the distribution $p_X$ is
available to, and may be used by the reconstruction algorithm. Under the mild assumptions that
sensing is linear (as per Eq.~(\ref{eq:LinearSensing})), and that the
reconstruction mapping is Lipschitz continuous, they proved that 
compression is asymptotically lossless if and only if
\begin{eqnarray}
m\ge n\, \uRenyi(p_X)+o(n)\, .
\end{eqnarray}
Here $\uRenyi(p_X)$ is the (upper) R\'enyi information dimension of the
distribution $p_X$. We refer to Section \ref{sec:Formal} for a precise
definition of this quantity. Suffices to say that, if $p_X$ is
$\ve$-sparse (i.e., if it puts mass at most $\ve$ on nonzeros) then
$\uRenyi(p_X)\le \ve$. Also, if $p_X$ is the convex combination of a
discrete part (sum of Dirac's delta) and an absolutely continuous part
(with a density), then $\uRenyi(p_X)$ is equal to the weight of the
absolutely continuous part.

This result is quite striking. For instance, it implies that, for
random $k$-sparse vectors, $m\ge k+o(n)$ measurements are
sufficient. Also, if the entries of $x$ are random and take values
in, say, $\{-10,-9,\dots,-9,+10\}$, then a sublinear number of 
measurements $m=o(n)$, is sufficient!
At the same time, the result of Wu and Verd\'u presents two important limitations.
First, it does  not provide robustness guarantees\footnote{While this
  paper was about to be posted, we became aware of a paper by Wu and
  Verd\'u \cite{WuVerduNoisy} proving a robustness guarantee for $\delta
  > \uMMSE(p_X)$ for the case of probability distributions that do not
  contain singular continuous component.
The reconstruction method is again not practical.} of the type
described above. 
Second and most importantly, it does not provide any computationally practical
algorithm for reconstructing $x$ from measurements $y$.

In an independent line of work, Krzakala et al. \cite{KrzakalaEtAl} developed an
approach that leverages on the idea of \emph{spatial coupling}. This
idea was introduced for the compressed sensing literature by Kudekar
and Pfister \cite{KudekarPfister} (see
\cite{KudekarBEC} and Section \ref{sec:Related} for a discussion
of earlier work on this topic).
Spatially coupled matrices are, roughly speaking, random sensing
matrices with a band-diagonal structure. 
The analogy is, this time, with channel coding.\footnote{Unlike
  \cite{KrzakalaEtAl}, we follow here the terminology developed within coding
  theory.} In this context, spatial coupling, in
conjunction with message-passing decoding,  allows to achieve Shannon capacity
on memoryless communication channels. It is therefore natural to ask
whether  an approach based on spatial coupling can enable to 
sense random vectors $x$ at an undersampling rate $m/n$ close to the 
R\'enyi information dimension of the coordinates of $x$,
$\uRenyi(p_X)$. 
Indeed, the authors of \cite{KrzakalaEtAl} evaluate such a scheme numerically on a few classes of random vectors and demonstrate that it indeed achieves rates close to
the fraction of non-zero entries. They also support this claim by
insightful statistical physics arguments.

In this paper, we fill the gap between the above works, and
present  the following contributions:
\begin{description}
\item[Construction.] We describe a construction for spatially coupled
  sensing matrices $A$ that is somewhat broader than the one of
  \cite{KrzakalaEtAl} and give precise prescriptions for the
  asymptotic values of various parameters. We also use a somewhat different reconstruction
  algorithm from the one in \cite{KrzakalaEtAl}, by building on the
  approximate message passing (AMP) approach of
  \cite{DMM09,DMM_ITW_I}.
  AMP algorithms have the advantage of smaller memory complexity
  with respect to standard message passing, and of smaller
  computational complexity whenever fast multiplication procedures are
  available for $A$.
\item[Rigorous proof of convergence.] Our main contribution is a
  rigorous proof that the above approach indeed achieves the
  information-theoretic limits set out by Wu and Verd\'u
  \cite{WuVerdu}. Indeed, we prove that, for 
  sequences of spatially coupled sensing matrices
  $\{A(n)\}_{n\in\naturals}$, $A(n)\in\reals^{m(n)\times n}$ with asymptotic undersampling rate
 $\delta=\lim_{n\to\infty} m(n)/n$, AMP reconstruction is with high
 probability successful in recovering the signal $x$, provided
 $\delta>\uRenyi(p_X)$.
\item[Robustness to noise.] We prove that the present approach is
  robust\footnote{This robustness bound holds for all $\delta>\uMMSE(p_X)$, where
$\uMMSE(p_X)$ is the upper MMSE dimension of $p_X$. (see
Definition~\ref{def:MMSE-dim}). It is worth noting that
$\uMMSE(p_X)=\uRenyi(p_X)$  for  a broad class of distributions $p_X$
including distributions without \emph{singular continuous component}.}
to noise in the following sense. For any signal distribution
  $p_X$ and undersampling rate $\delta$, there exists a constant $C$
  such that the output $\hx(y)$ of the reconstruction algorithm
  achieves a mean square error per coordinate 
$n^{-1}\E\{\|\hx(y)-x\|_2^2\}\le C\, \sigma^2$. This result holds
under the noisy measurement model (\ref{eq:NoisyModel}) for a broad
class of noise models for $w$, including i.i.d. noise coordinates
$w_i$ with $\E\{w_i^2\}=\sigma^2<\infty$.
\item[Non-random signals.] Our proof does not apply uniquely to random
  signals $x$ with i.i.d. components, but indeed to more general
  sequences of signals $\{x(n)\}_{n\in\naturals}$, $x(n)\in\reals^n$
  indexed by their dimension $n$. The conditions required are: $(1)$ that the 
empirical distribution of the coordinates of $x(n)$ converges (weakly)
to $p_X$; and $(2)$ that $\|x(n)\|^2_2$ converges to the second moment of
the asymptotic law $p_X$.

There is a fundamental reason why this more general framework turns out
to be equivalent to the random signal model. This can be traced back to the
fact that, within our construction, the columns of the matrix $A$
are probabilistically exchangeable. Hence any
vector $x(n)$ is equivalent to the one whose coordinates have been
randomly permuted. The latter is in turn very close to the
i.i.d. model. By the same token, the rows of $A$ are exchangeable and
hence
the noise vector $w$ does not need to be random either.
\end{description}
Interestingly, the present framework changes the notion of `structure'
that is relevant for reconstructing the signal $x$.
Indeed, the focus is shifted from the
\emph{sparsity} of $x$ to the \emph{information dimension} $\uRenyi(p_X)$.
In other words, the signal structure that facilitates recovery from a small number of linear measurements
 is the low-dimensional structure in an information theoretic sense, quantified
 by the information dimension of the signal.

In the rest of this section we state formally our results, and discuss
their implications and limitations, as well as relations with earlier
work.  Section \ref{sec:GeneralAlgo} provides a precise description
of the matrix construction and reconstruction algorithm. Section \ref{sec:Lemmas}
reduces the proof of our main results to two key lemmas. One of these
lemmas is a (quite straightforward) generalization of the state
evolution technique of \cite{DMM09,BM-MPCS-2011}. The second lemma
characterizes the behavior of the state evolution recursion, and is
proved in Section \ref{sec:AnalysisLemma}. The proof of a number of
intermediate technical steps is deferred to the appendices.

\subsection{Formal statement of the results}
\label{sec:Formal}

We consider the noisy model (\ref{eq:NoisyModel}).
An instance of the problem is therefore completely specified by the
triple $(x,w,A)$. We will be interested in the asymptotic
properties of sequence of instances indexed by the problem dimensions
$\cS = \{(x(n),w(n),A(n))\}_{n\in\naturals}$. We recall a definition
from \cite{BayatiMontanariLASSO}. (More precisely,
\cite{BayatiMontanariLASSO} introduces the $B=1$ case of this definition.)
\begin{definition}\label{def:Converging}
The sequence of instances $\cS = \{x(n), w(n), A(n)\}_{n\in\naturals}$
indexed by $n$ is said to be a $B$-\emph{converging sequence}
if
$x(n)\in\reals^{n}$, $w(n)\in\reals^m$, $A(n)\in\reals^{m\times n}$
with $m=m(n)$ is  such that $m/n\to\delta\in(0,\infty)$,
and in addition the following conditions hold \footnote{If $(\mu_k)_{k\in \naturals}$
is a sequence of measures and $\mu$ is another measure, all defined on $\reals$,
the weak convergence of $\mu_k$ to $\mu$ along with the convergence of their
second moments to the second moment of $\mu$ is equivalent to convergence in 
$2$-Wasserstein distance~\cite{villani2008optimal}. Therefore, conditions $(a)$-$(b)$
are equivalent to the following. The empirical distributions of the signal $x(n)$ and the empirical 
distributions of noise $w(n)$ converge
in $2$-Wasserstein distance.} :
\begin{itemize}
\item[$(a)$] The empirical distribution of the entries of $x(n)$
converges weakly to a probability measure $p_{X}$ on $\reals$
with bounded second moment. Further
$n^{-1}\sum_{i=1}^nx_{i}(n)^2\to \E\{X^{2}\}$, where the expectation
is taken with respect to $p_X$.
\item[$(b)$] The empirical distribution of the entries of $w(n)$
converges weakly to a probability measure $p_{W}$ on $\reals$
with bounded second moment. Further
$m^{-1}\sum_{i=1}^mw_{i}(n)^2\to \E\{W^{2}\}\equiv\sigma^2$,
where the expectation is taken with respect to $p_W$.
\item[$(c)$]If $\{e_i\}_{1\le i\le n}$, $e_i\in\reals^n$ denotes the canonical
basis, then $\underset{n\to \infty}{\limsup}\max_{i\in [n]}\|A(n)e_i\|_2\le B$,\\
$\underset{n\to \infty}{\liminf} \min_{i\in [n]}\|A(n)e_i\|_2\ge 1/B$.
\end{itemize}
We further say that $\{(x(n),w(n))\}_{n\ge 0}$ is a \emph{converging
  sequence of instances}, if they satisfy conditions $(a)$ and $(b)$.
We say that $\{A(n)\}_{n\ge0}$ is a $B$-\emph{converging sequence} of
sensing matrices if they satisfy condition $(c)$ above, and we call it a \emph{converging sequence}
if it is $B$-converging for some $B$. Similarly, we say $\cS$ is a \emph{converging sequence} if it is $B$-converging for some $B$.

Finally, if the sequence $\{(x(n),w(n),A(n))\}_{n\ge 0}$ is random,
the above conditions are required to hold almost surely.
\end{definition}
Notice that standard normalizations of the sensing matrix correspond
to $\|A(n)e_i\|_2^2= 1$ (and hence $B=1$) or to $\|A(n)e_i\|_2^2=
m(n)/n$. The former corresponds to normalized columns and the latter
corresponds to normalized rows. Since
throughout we assume $m(n)/n\to\delta\in(0,\infty)$, these conventions
only differ by a rescaling of the noise variance.  In order to
simplify the proofs, we allow ourselves somewhat more freedom by
taking $B$ a fixed constant.

Given a sensing matrix $A$, and a vector of measurements $y$, 
a reconstruction algorithm produces an estimate $\hx(A;y)\in\reals^n$ 
of $x$. In this paper we assume that the empirical distribution $p_{X}$, and
the noise level $\sigma^2$ are known to the estimator, and hence the
mapping $\hx:(A,y)\mapsto \hx(A;y)$ implicitly depends on $p_X$ and
$\sigma^2$. Since however $p_X,\sigma^2$ are fixed throughout, we
avoid the cumbersome notation $\hx(A,y,p_X,\sigma^2)$.

Given a converging sequence of instances
$\cS=\{x(n),w(n),A(n)\}_{n\in\naturals}$, 
and an estimator  $\hx$, we define the asymptotic
per-coordinate reconstruction mean square error as
\begin{eqnarray}
\MSE(\cS;\hx)=
\underset{n\to\infty}{\limsup}\,  \frac{1}{n}\, \big\|\hx\big(A(n);y(n)\big)-x(n)\big\|^2\, .
\end{eqnarray}
Notice that the quantity  on the right hand side
depends on the matrix $A(n)$, which will be random, and on the signal and noise
vectors  $x(n)$, $w(n)$ which can themselves be random. 
Our results hold almost surely with respect to these random
variables. In some applications it is more customary to take the
expectation with respect to the noise and signal distribution, i.e., to
consider the quantity
\begin{eqnarray}
\MSEexp(\cS;\hx)=
\underset{n\to\infty}{\limsup}\,  \frac{1}{n}\, \E\big\|\hx\big(A(n);y(n)\big)-x(n)\big\|^2\, .
\end{eqnarray}
It turns out that the almost sure bounds imply, in the present setting, bounds on the expected mean square error $\MSEexp$, as well.

In this paper we study a specific low-complexity estimator, based on
the AMP algorithm first proposed in \cite{DMM09}. AMP is an iterative
algorithm derived from the theory of belief propagation in graphical
models \cite{MontanariChapter}. At each iteration $t$, it keeps track
of an estimate $x^t\in\reals^n$ of the unknown signal $x$. This is used to
compute residuals $(y-Ax^t)\in\reals^m$. These  correspond to the part of
observations that is not explained by the current estimate
$x^t$. The residuals are then processed through a matched filter
operator (roughly speaking, this amounts to multiplying the residuals
by the transpose of $A$) and then applying a non-linear denoiser, to produce the new
estimate $x^{t+1}$.

Formally, we start with an initial guess $x^1_i =\E\{X\}$ for all $i\in[n]$ and proceed by 
\begin{eqnarray}
x^{t+1} & = & \eta_t(x^t+(Q^t\odot A)^*r^t)\, ,\label{eq:AMP1}\\
r^t & = & y-Ax^t+\ons^t\odot r^{t-1}\, .\label{eq:AMP2}
\end{eqnarray}
The second equation corresponds to the computation of new residuals
from the current estimate. The memory term (also known as `Onsager
term' in statistical physics) plays a crucial role as emphasized in \cite{DMM09,BM-MPCS-2011,BM-Universality,JM-StateEvolution}. The first equation
describes matched filter, with multiplication by $(Q_t\odot A)^*$,
followed by application of the denoiser $\eta_t$. Throughout  $\odot$ indicates Hadamard
(entrywise) product and $X^*$ denotes the transpose of matrix
$X$. 

For each $t$, the denoiser $\eta_t:\reals^n\to \reals^n$ is a differentiable non-linear function that depends
on the input distribution $p_{X}$.
Further, $\eta_t$ is separable\footnote{We refer to \cite{DonohoJohnstoneMontanari}
  for a study of non-separable denoisers in AMP algorithms.}, namely,  for a vector $v\in \reals^n$,
we have $\eta_t(v) =(\eta_{1,t}(v_1),\dots,\eta_{n,t}(v_n))$.
The matrix $Q^t\in \reals^{m\times n}$ and the vector
$\ons^t\in\reals^m$ can be efficiently computed from the
current state $x^t$ of the algorithm,
Further $Q^t$ does not depend on the problem instance and hence can
be precomputed. Both $Q^t$ and  $\ons^t$ are block-constants, i.e., they can be 
partitioned into blocks such that within each block all the entries have the same value. This
property makes their evaluation, storage and manipulation particularly 
convenient. 

We refer to the next section for explicit definitions of
these quantities. A crucial element is the specific choice of
$\eta_{i,t}$. The general guiding principle is that the argument
$y^t=x^t+(Q^t\odot A)^*r^t$  in Eq.~(\ref{eq:AMP1}) should be interpreted
as a noisy version of the unknown signal $x$, i.e., $y^t= x +{\sf
  noise}$. 
The denoiser $\eta_t$
must therefore be chosen as to minimize the mean square error at
iteration $(t+1)$.
The papers \cite{DMM09,DonohoJohnstoneMontanari} take a
minimax point of view, and hence study denoisers that achieve the
smallest mean square error over the worst case signal $x$ in a certain
class. For instance, coordinate-wise soft
thresholding is nearly minimax optimal over the class of sparse signals~\cite{DonohoJohnstoneMontanari}.
Here we instead assume that the prior $p_X$ is known, and hence the
choice of $\eta_{i,t}$ is uniquely dictated by the objective of minimizing the mean square error at
iteration $t+1$. In other words $\eta_{i,t}$ takes the form of a Bayes optimal estimator
for the prior $p_X$. In order to stress this point, we will
occasionally refer to
this as the Bayes optimal AMP algorithm.
As shown in Appendix \ref{app:Lipschitz}, $x^t$ is (almost
surely) a local Lipschitz continuous function of the observations $y$.

Finally notice that \cite{DMM_ITW_I,MontanariChapter} also derived AMP starting from a
Bayesian graphical models point of view, with the signal $x$ modeled as random
with i.i.d. entries. The algorithm in Eqs.~(\ref{eq:AMP1}),
(\ref{eq:AMP2}) differs from the one in \cite{DMM_ITW_I}
in that the matched filter operation requires scaling
$A$ by the matrix $Q^t$. This is related to the fact that we will use
a matrix $A$ with independent but not identically distributed entries
and,  as a consequence, the accuracy of each entry $x^t_i$ depends on
the index $i$ as well as on $t$.

We denote by $\MSEAMP(\cS;\sigma^2)$ the mean square error
achieved by the Bayes optimal AMP algorithm, where we made explicit
the dependence on $\sigma^2$. Since the AMP estimate depends on the
iteration number $t$, the definition of $\MSEAMP(\cS;\sigma^2)$
requires some care. The basic point is that we need to iterate the
algorithm only for a constant number of iterations, as $n$ gets
large. Formally, we let
\begin{eqnarray}
\MSEAMP(\cS;\sigma^2) \equiv
\lim_{t\to\infty}\underset{n\to\infty}{\limsup}\, \frac{1}{n}\big\|x^t\big(A(n);y(n)\big)-x(n)\big\|^2
\,. 
\end{eqnarray}
As discussed above, limits will be shown to exist almost surely, when
the instances $(x(n),w(n),A(n))$ are random, and almost sure upper
bounds on $\MSEAMP(\cS;\sigma^2)$ will be proved. (Indeed
$\MSEAMP(\cS;\sigma^2)$ turns out to be deterministic.)
On the other hand, one might be interested in the expected error
\begin{eqnarray}
\MSEAMPexp(\cS;\sigma^2) \equiv
\lim_{t\to\infty}\underset{n\to\infty}{\limsup} \, \frac{1}{n}\E\big\{\big\|x^t\big(A(n);y(n)\big)-x(n)\big\|^2\big\}
\, .
\end{eqnarray}

We will tie the success of our compressed sensing scheme to the
fundamental information-theoretic limit established in
\cite{WuVerdu}. The latter is expressed in terms of the R\'enyi
information dimension of the probability measure $p_X$. 
\begin{definition}
 Let $p_X$ be a probability measure over $\reals$, and $X\sim p_X$. 
 The \emph{upper} and \emph{lower information dimension} of $p_X$ are defined as
\begin{align}
\uRenyi(p_X) &= \underset{\ell\to\infty}{\limsup} \frac{H([X]_{\ell})}{\log
  \ell}\, .\\
 \lRenyi(p_X) &= \underset{\ell\to\infty}{\liminf} \frac{H([X]_{\ell})}{\log
  \ell}\, .
\end{align}
Here $H(\,\cdot\,)$ denotes Shannon entropy and, for $x\in \reals$, $[x]_{\ell}\equiv \lfloor \ell
x\rfloor/\ell$, and $\lfloor x\rfloor \equiv \max \{k\in \integers\;
:\;\;\; k\le x\}$. If the $\limsup$ and $\liminf$ coincide, then we
let $\Renyi(p_X) = \uRenyi(p_X) = \lRenyi(p_X)$.
\end{definition}
Whenever the limit of $H([X]_{\ell})/\log\ell$ exists and is finite\footnote{This condition can be replaced by $H(\lfloor X\rfloor) < \infty$. A sufficient condition is that $\E[\log(1+|X|)] < \infty$, which is certainly satisfied if $X$ has a finite variance~\cite{WuVerduMMSE}. }, the R\'enyi
information dimension can also be characterized as follows.  
Write the binary expansion of $X$, $X=D_0.D_1D_2D_3\dots$ with $D_i\in
\{0,1\}$ for $i\ge 1$. Then $\uRenyi(p_X)$ is the entropy rate of the
stochastic process $\{D_1,D_2,D_3,\dots\}$.
It is also convenient to recall the following result from
\cite{Renyi,WuVerdu}.
\begin{propo}[\cite{Renyi,WuVerdu}]\label{propo:Renyi}
Let $p_X$ be a probability measure over $\reals$, and $X\sim p_X$.  
Assume $H(\lfloor X\rfloor)$ to be finite. If $p_X = (1-\ve) \nu_{\rm
  d}+\ve\widetilde{\nu}$ with $\nu_{\rm d}$ a discrete distribution
(i.e., with countable support), then $\uRenyi(p_X)\le \ve$. 
Further, if $\widetilde{\nu}$ has a density with respect to Lebesgue
measure,  then $\Renyi(p_X) = \uRenyi(p_X) = \lRenyi(p_X)=\ve$.
In particular, if $\prob\{X\neq 0\}\le \ve$ then $\uRenyi(p_X)\le \ve$. 
\end{propo}
In order to present our result concerning the robust reconstruction, we need the definition
of \emph{MMSE dimension} of the probability measure $p_X$.
 
Given the signal distribution $p_X$, we let $\mmse(s)$ denote the 
minimum mean square error in estimating $X\sim p_X$ from a noisy
observation in gaussian noise, at signal-to-noise ratio $s$.
Formally
\begin{eqnarray}
\mmse(s)
\equiv\inf_{\eta:\reals\to\reals}\E\big\{\big[X-\eta(\sqrt{s}\,
X+Z)\big]^2\big\}\, ,
\end{eqnarray}
where $Z\sim\normal(0,1)$.
Since the minimum mean square error estimator is just the conditional
expectation, this is given by
\begin{eqnarray}
\mmse(s) = \E\big\{\big[X-\E[X|Y]\big]^2\big\}\, ,\;\;\;\;\;\; Y =
\sqrt{s}\, X+Z\, .
\end{eqnarray}
Notice that $\mmse(s)$ is naturally well defined for $s=\infty$, with
$\mmse(\infty) = 0$. We will therefore interpret it as a function
$\mmse:\realsc\to\realsc$ where $\realsc\equiv [0,\infty]$
is the completed non-negative real line.

We recall the inequality
\begin{eqnarray}
0 \le \mmse(s) \le \frac{1}{s},
\end{eqnarray}
obtained by the estimator $\eta(y) = y/\sqrt{s}$. A
finer characterization of the scaling of $\mmse(s)$ is provided by the
following definition.
\begin{definition}[\cite{WuVerduMMSE}]\label{def:MMSE-dim}
The \emph{upper} and~\emph{lower MMSE dimension} of the probability
measure $p_X$ over $\reals$ are defined as
\begin{align}
\uMMSE(p_X) = \underset{s \to \infty}{\limsup}\; s\cdot \mmse(s)\,, \label{eqn:ummse_def}\\
\lMMSE(p_X) = \underset{s \to \infty}{\liminf}\; s \cdot \mmse(s)\,. \label{eqn:lmmse_def}
\end{align} 
 If the $\limsup$ and $\liminf$ coincide, then we
let $\MMSE(p_X) = \uMMSE(p_X) = \lMMSE(p_X)$.
\end{definition}

It is also convenient to recall the following result from~\cite{WuVerduMMSE}.
\begin{propo}[\cite{WuVerduMMSE}]\label{propo:MMSEdim}
If $H(\lfloor X \rfloor) < \infty$, then
\begin{eqnarray}
\lMMSE(p_X) \le \lRenyi(p_X) \le \uRenyi(p_X) \le \uMMSE(p_X).
\end{eqnarray}
Hence, if $\MMSEd(p_X)$ exists, then $d(p_X)$ exists and $\MMSEd(p_X) = d(p_X)$. 
In particular, this is the case if $p_X = (1-\ve) \nu_{\rm
  d}+\ve\widetilde{\nu}$ with $\nu_{\rm d}$ a discrete distribution
(i.e., with countable support), and $\widetilde{\nu}$ has a density with respect to Lebesgue
measure.
\end{propo}
 
We are now in position to state our main results. The first one
states that for any undersampling rate above Renyi information
dimension $\delta> \uRenyi(p_X)$, we have
$\MSEAMP(\cS;\sigma^2)\to 0$ as $\sigma^2\to 0$ with, in particular, $\MSEAMP(\cS;\sigma^2=0)=0$.
\begin{thm}\label{thm:MainTheorem}
Let $p_X$ be a probability measure on the real line and assume
\begin{eqnarray}
\delta> \uRenyi(p_X).\label{eq:MainTheorem}
\end{eqnarray}
Then there exists a random converging sequence of sensing matrices 
$\{A(n)\}_{n\ge 0}$, $A(n)\in\reals^{m\times n}$, $m(n)/n\to\delta$
(with distribution depending only on $\delta$), for which the
following holds.
For any $\ve > 0$, there exists $\sigma_0=\sigma_0(\ve,\delta,p_X)$ such that for any converging sequence of instances  $\{(x(n),w(n))\}_{n\ge0}$
with parameters $(p_X,\sigma^2,\delta)$ and $\sigma \in [0,\sigma_0]$, we have, almost surely
\begin{eqnarray}
\MSEAMP(\cS;\sigma^2)\le \ve\, . \label{eq:construct}
\end{eqnarray}
Further, under the same assumptions, we have 
$\MSEAMPexp(\cS;\sigma^2) \le \ve$.
\end{thm}
The second theorem characterizes the rate at which the mean square
error goes to $0$. In particular, we show that $\MSEAMP(\cS;\sigma^2)
= O(\sigma^2)$ provided $\delta> \uMMSE(p_X)$.
\begin{thm}\label{thm:MainTheoremR}
Let $p_X$ be a probability measure on the real line and assume
\begin{eqnarray}
\delta> \uMMSE(p_X).
\end{eqnarray}
Then there exists a random converging sequence of sensing matrices 
$\{A(n)\}_{n\ge 0}$, $A(n)\in\reals^{m\times n}$, $m(n)/n\to\delta$
(with distribution depending only on $\delta$) and a finite stability
constant $C = C(p_X,\delta)$, such that the
following is true.
For any converging sequence of instances  $\{(x(n),w(n))\}_{n\ge0}$
with parameters $(p_X,\sigma^2,\delta)$, we have, almost surely
\begin{eqnarray}
\MSEAMP(\cS;\sigma^2) \le C\, \sigma^2\, . \label{eq:Robustness}
\end{eqnarray}
Further, under the same assumptions, we have 
$\MSEAMPexp(\cS;\sigma^2) \le C\, \sigma^2$.

 Finally, the sensitivity to small noise is bounded as
\begin{eqnarray}
\lim_{\sigma \to 0} \frac{1}{\sigma^2}\, \MSEAMP(\cS;\sigma^2) \le \frac{4\delta - 2\uMMSE(p_X)}{\delta - \uMMSE(p_X)}\,.
\end{eqnarray}
\end{thm}
The performance guarantees in Theorems \ref{thm:MainTheorem} and
\ref{thm:MainTheoremR} are achieved with special constructions of the
sensing matrices $A(n)$. These are matrices with independent Gaussian
entries with unequal variances (heteroscedastic entries), with a band
diagonal structure. The motivation for this construction, and
connection with coding theory is further discussed in Section
\ref{sec:Intuition},
while formal definitions are given in Section \ref{sec:MatrixEnsemble} and
\ref{sec:Choices}.

Notice that, by Proposition \ref{propo:MMSEdim},
$\uMMSE(p_X)\ge \uRenyi(p_X)$, and $\uMMSE(p_X)= \uRenyi(p_X)$
for a broad class of probability measures $p_X$, including all
measures that do not have a singular continuous component
(i.e., decomposes into a pure point mass component and an absolutely continuous component).

The noiseless model (\ref{eq:LinearSensing}) is covered as a special
case of  Theorem \ref{thm:MainTheorem} by taking $\sigma^2 \downarrow 0$. For the reader's
convenience, 
we state the result explicitly as a corollary.
\begin{coro}\label{coro:Noiseless}
Let $p_X$ be a probability measure on the real line. Then, for any
$\delta>\uRenyi(p_X)$ there exists a random converging sequence of sensing matrices 
$\{A(n)\}_{n\ge 0}$, $A(n)\in\reals^{m\times n}$, $m(n)/n\to\delta$
(with distribution depending only on $\delta$)
such that, for any sequence of vectors $\{x(n)\}_{n\ge 0}$ whose
empirical distribution converges to $p_X$, the Bayes optimal AMP asymptotically almost surely recovers $x(n)$ from
$m(n)$ measurements $y = A(n)x(n)\in\reals^{m(n)}$.
(By `asymptotically almost surely' we mean $\MSEAMP(\cS;0) =0$ almost surely, and $\MSEAMPexp(\cS;0) =0$.)
\end{coro}
Note that it would be interesting to prove a stronger guarantee in the
noiseless case, namely $\lim_{t\to\infty}x^t(A(n);y(n)) = x(n)$ with
probability converging to $1$ as $n\to\infty$. The present paper does
not lead to a proof of this statement.
%
%
\subsection{Discussion}
\label{sec:Discussion}

Theorem \ref{thm:MainTheorem} and Corollary \ref{coro:Noiseless} are,
in many ways, puzzling. It is instructive to spell out in detail a few
specific examples, and discuss interesting features.

\vspace{0.2cm}

{\bf Example 1 (Bernoulli-Gaussian signal).} Consider a Bernoulli-Gaussian distribution
\begin{eqnarray}
p_X = (1-\ve) \, \delta_0 + \ve \, \gamma_{\mu,\sigma}
\end{eqnarray}
where $\gamma_{\mu,\sigma}(\de x) =
(2\pi\sigma^2)^{-1/2}\exp\{-(x-\mu)^2/(2\sigma^2)\}\de x$ is the Gaussian
measure with  mean $\mu$ and variance $\sigma^2$. This model
has been studied numerically in a number of papers, including
\cite{BaronBayesian,KrzakalaEtAl}.  By Proposition \ref{propo:Renyi},
we have $\uRenyi(p_X) = \ve$, and by Proposition \ref{propo:MMSEdim},
$\uMMSE(p_X) = \lMMSE(p_X) = \ve$ as well.

Construct random signals $x(n)\in\reals^n$ by sampling
i.i.d. coordinates $x(n)_i\sim p_X$. Glivenko-Cantelli's theorem
implies that the empirical distribution of the coordinates of $x(n)$
converges almost surely to $p_X$, hence we can apply  Corollary
\ref{coro:Noiseless} to recover $x(n)$ from $m(n)=n\ve+o(n)$
spatially coupled measurements $y(n)\in\reals^{m(n)}$. Notice that the number of
non-zero entries in $x(n)$ is, almost surely, $k(n)=n\ve+o(n)$.
Hence, we can restate the implication of  Corollary
\ref{coro:Noiseless}  as follows.
A sequence of vectors $x(n)$ with Bernoulli-Gaussian distribution and
$k(n)$ nonzero entries can almost surely recovered by $m(n) =
k(n)+o(n)$ spatially coupled measurements.

\vspace{0.2cm}

{\bf Example 2 (Mixture signal with a point mass).}  The above remarks generalize immediately to arbitrary
mixture distributions of the form
\begin{eqnarray}
p_X = (1-\ve) \, \delta_0 + \ve \, q\, ,\label{eq:Mixture}
\end{eqnarray}
where $q$ is a measure that is absolutely continuous with respect to
Lebesgue measure, i.e., $q(\de x) = f(x)\, \de x$ for some measurable
function $f$. Then, by Proposition \ref{propo:Renyi},
we have $\uRenyi(p_X) = \ve$, and by Proposition \ref{propo:MMSEdim},
$\uMMSE(p_X) = \lMMSE(p_X) = \ve$ as well. Arguing as above we have the following.
\begin{coro}
Let $\{x(n)\}_{n\ge 0}$ be a sequence of vectors with
i.i.d. components $x(n)_i\sim p_X$ where $p_X$ is a mixture
distribution as per Eq.~(\ref{eq:Mixture}). Denote by $k(n)$ the
number of nonzero entries in $x(n)$. Then,  almost
surely as $n\to\infty$, Bayes optimal AMP recovers the signal $x(n)$ from $m(n)  =
k(n)+o(n)$ spatially coupled measurements. 
\end{coro}
Under the regularity hypotheses of \cite{WuVerdu}, no scheme can do
substantially better, i.e., reconstruct $x(n)$ from
$m(n)$ measurements if $\underset{n\to\infty}{\limsup}\, m(n)/k(n)<1$.

One way to think about this result is the following. If an oracle gave
us the support of $x(n)$, we would still need $m(n)\ge k(n)-o(n)$
measurements to reconstruct the signal. Indeed, the entries in the
support have distribution $q$, and $\uRenyi(q)=1$.   Corollary
\ref{coro:Noiseless} implies that the measurements overhead for
estimating the support of $x(n)$ is  sublinear, $o(n)$, even when the
support is of order $n$.

It is sometimes informally argued that compressed sensing requires at
least $\Theta(k\log (n/k))$ for `information-theoretic reasons',
namely that specifying the support requires about $nH(k/n)\approx
k\log(n/k)$ bits. This argument is of course incomplete because it
assumes that each measurement $y_i$ is described by a bounded number
of bits.
Since it is folklore to say that sparse signal recovery requires $m\ge
C\, k\log (n/k)$ measurement, it is instructive to survey the results
of this type and explain why they do not apply to the present setting.
This elucidates further the implications of our results.

Specifically, \cite{Wainwright2009,Saligrama} prove
information-theoretic lower bounds on the required number of
measurements, under specific constructions for
the random sensing matrix $A$. Further, these papers focus on the
specific problem of exact support recovery. 
The paper \cite{WainwrightEllP} proves minimax bounds for
reconstructing vectors belonging to $\ell_p$-balls. 
Notice that these bounds are usually proved by exhibiting a least
favorable prior, which is close to a signal with i.i.d. coordinates.
However, as the
noise variance tends to zero, these bounds depend on the sensing matrix in a way that is
difficult to quantify. In particular, they provide no explicit lower
bound on the number of measurements required for exact recovery in the
noiseless limit. Similar bounds were obtained for arbitrary 
measurement matrices in \cite{CandesDavenport}. Again, these lower
bounds vanish as noise tends to zero as soon as $m(n)\ge k(n)$.

A different line of work derives lower bounds from Gelfand' width
arguments \cite{Donoho1,KashinTemlyakov}.
These lower bounds are only proved to be a necessary
condition for a stronger reconstruction guarantee. Namely, these
works require the vector of measurements $y= Ax$ to enable recovery
for \emph{all} $k$-sparse vectors $x\in\reals^n$. This corresponds to the
`strong' phase transition of \cite{DoTa05,Do}, and is also referred to
as the `for all' guarantee in the computer science literature
\cite{Indyk}.

The lower bound that comes closest to the present setting is the `randomized' lower
bound \cite{IndykLower}. In this work the authors consider a fixed
signal $x$ and a random sensing matrix as in our setting. In other
words they do not assume a standard minimax setting.
However they require an $\ell_1-\ell_1$ error guarantee which is a stronger
stability condition than what is achieved in Theorem
\ref{thm:MainTheoremR}, allowing for a more powerful noise process.
Indeed the same paper also proves that recovery is possible from $m(n)
= O(k(n))$ measurements under stronger conditions.

\vspace{0.2cm}

{\bf Example 3 (Discrete signal).} Let $K$ be a fixed integer,
$a_1,\dots,a_K\in\reals$, and $(p_1,p_2,\dots,p_K)$ be a collection of
non-negative numbers that add up to one. Consider the probability
distribution that puts mass $p_i$ on each $a_i$
\begin{eqnarray}
p_X = \sum_{i=1}^Kp_i\,\delta_{a_i}\, ,\label{eq:Discrete}
\end{eqnarray}
and let $x(n)$ be a signal with i.i.d. coordinates $x(n)_i\sim
p_X$. By Proposition \ref{propo:Renyi}, we have $\uRenyi(p_X)=0$.
As above, the empirical distribution of the coordinates of the vectors
$x(n)$ converges to $p_X$. By applying Corollary \ref{coro:Noiseless}
we obtain the following
\begin{coro}
Let $\{x(n)\}_{n\ge 0}$ be a sequence of vectors with
i.i.d. components $x(n)_i\sim p_X$ where $p_X$ is a discrete
distribution as per Eq.~(\ref{eq:Discrete}). Then, almost
surely as $n\to\infty$, Bayes optimal AMP recovers the signal $x(n)$
from $m(n)  =o(n)$ 
spatially coupled measurements. 
\end{coro}
It is important to further discuss the last statement because the reader might
be misled into too optimistic a conclusion. Consider any 
signal $x\in\reals^n$. For practical purposes, this will be represented
with finite precision, say as a vector of $\ell$-bit numbers. Hence, in
practice, the distribution $p_X$ is always discrete, with
$K=2^{\ell}$ a fixed number dictated by the precision requirements. 
A sublinear
number of measurements $m(n)=o(n)$ will then be sufficient to achieve
this precision.

On the other hand,  Theorem
\ref{thm:MainTheorem} and Corollary \ref{coro:Noiseless} are
asymptotic statements, and the convergence rate is not claimed to be
uniform in $p_X$. In particular, the values of $n$ at which it
becomes accurate will likely increase with $K$.

\vspace{0.2cm}

{\bf Example 4 (A discrete-continuous mixture).} Consider the
probability distribution 
\begin{eqnarray}
p_X = \ve_+\,\delta_{+1}+\ve_-\delta_{-1}+\ve\, q\, ,
\end{eqnarray}
where $\ve_++\ve_-+\ve = 1$ and the probability measure
$q$ has a density with respect to Lebesgue measure. Again, 
let $x(n)$ be a vector with i.i.d. components $x(n)_i\sim p_X$. 
We can apply  Corollary \ref{coro:Noiseless} to conclude that 
$m(n) = n\ve+o(n)$ spatially coupled measurements are sufficient. 
This should be contrasted with the case of sensing matrices
with i.i.d. entries studied in \cite{DoTa10b} under convex
reconstruction methods (namely solving the feasibility problem $y=Ax$
under the constraint $\|x\|_{\infty}\le 1$). In this case $m(n) = n(1+\ve)/2+o(n)$
measurements are necessary.

\vspace{0.2cm}

In the next section we describe the basic intuition
behind the surprising phenomenon in Theorems \ref{thm:MainTheorem}
and~\ref{thm:MainTheoremR}, and why spatially coupled sensing matrices are so useful.
We conclude by stressing once more the limitations of these results: 
\begin{itemize}
\item The Bayes optimal AMP algorithm requires knowledge of the 
signal distribution $p_X$. Notice however that only a good
approximation of $p_X$ (call it $p_{\tX}$, and denote by $\tX$ the
corresponding random variable) is sufficient. Assume indeed that
$p_{X}$ and $p_{\tX}$ can be coupled in such a way that
$\E\{(X-\tX)^2\}\le \tilde{\sigma}^2$. Then
\begin{eqnarray}
x = \tx +u
\end{eqnarray}
where $\|u\|_2^2\lesssim n \tilde{\sigma}^2$. This is roughly equivalent to
adding to the noise vector $z$ further `noise' $\widetilde{z}$ with
variance $\tilde{\sigma}^2/\delta$. By this argument the
guarantee in Theorem \ref{thm:MainTheoremR} degrades gracefully as
$p_{\tX}$ gets different from $p_{X}$.
Another argument that leads to the same conclusion consists in
studying the evolution of the algorithm (\ref{eq:AMP1}),
(\ref{eq:AMP2}) when $\eta_t$ is matched to the incorrect prior, see
Appendix \ref{app:Prior}. 

Finally, it was demonstrated numerically in \cite{SchniterEM,KrzakalaEtAl}
that, in some cases, a good `proxy' for $p_X$ can be learned through an 
Expectation-Maximization-style iteration.  A rigorous study of this
approach goes beyond the scope of present paper.
\item In particular, the present approach does not provide uniform
  guarantees over the class of, say, sparse signals characterized by
  $p_X(\{0\})\ge 1-\ve$. In particular, both the phase transition
  location, cf. Eq.~(\ref{eq:MainTheorem}), and the robustness
  constant, cf. Eq.~(\ref{eq:Robustness}), depend on the distribution
  $p_X$. This should be contrasted with the minimax approach of
  \cite{DMM09,DMM-NSPT-11,DonohoJohnstoneMontanari} which provides
  uniform guarantees that are uniform over sparse signals.
\item As mentioned above, the guarantees in Theorems~\ref{thm:MainTheorem} and~\ref{thm:MainTheoremR} are only
  asymptotic. It would be important to develop  analogous
  non-asymptotic results.
\item The stability bound (\ref{eq:Robustness}) is non-uniform, in that
the proportionality constant $C$ depends on the signal distribution. 
It would be important to establish analogous bounds that are  uniform
over suitable classes of distributions. (We do not expect Eq.~(\ref{eq:Robustness})
to hold uniformly over \emph{all} distributions.)
\end{itemize}

%
%
\subsection{How does spatial coupling work?}
\label{sec:Intuition}

Spatial coupling was developed in coding theory to construct capacity
achieving LDPC codes \cite{Felstrom,Costello,KudekarBMS,HassaniEtAl,kudekar2012spatially}. 
The standard construction starts from the parity check matrix of an LDPC
code that is sparse but unstructured apart from the degree sequence. A
spatially coupled ensemble is then obtained by enforcing a
band-diagonal structure, while keeping the degree sequence unchanged. 
Usually this is done by graph liftings, but the underlying principle
is more general \cite{HassaniEtAl}.

Following the above intuition, spatially coupled sensing matrices $A$ are, roughly speaking, random
band-diagonal matrices. The construction given below (as the one of
\cite{KrzakalaEtAl}) uses matrices with independent zero-mean Gaussian
entries, with non-identical variances (heteroscedastic
entries). However, the simulations of \cite{JavanmardMon12} suggest
that a much broader set of matrices display similar performances. As
discussed in Section \ref{sec:MatrixEnsemble}, the construction is
analogous to graph liftings. We start by a matrix of variances $W =
(W_{r,c})$ and obtain the sensing matrix $A$ by replacing each entry
$W_{r,c}$ by a block with i.i.d. Gaussian entries with variance
proportional to 
$W_{r,c}$. 

It is convenient to think of the graph structure
that they induce on the reconstruction problem. Associate one node
(a \emph{variable node} in the language of factor graphs) to each
coordinate $i$ in the unknown signal $x$. Order these nodes on the
real line $\reals$,
putting the $i$-th node  at   location $i\in\reals$. Analogously, associate a node
(a \emph{factor node})
to each coordinate $a$ in the measurement vector $y$, and place the
node $a$ at position $a/\delta$ on the same line. Connect this node to
all the variable nodes $i$ such that $A_{ai}\neq 0$. If $A$ is band
diagonal, only  nodes that are placed close enough will be connected by an edge.
See Figure~\ref{fig:Spatial_graph} for an illustration.

In a spatially coupled matrix, additional measurements are associated
to the first few coordinates of $x$, say coordinates
$x_1,\dots,x_{n_0}$ with $n_0$ much smaller than $n$. 
This has a negligible impact on the overall undersampling ratio as $n/n_0\to\infty$.
Although the overall undersampling remains $\delta<1$, the coordinates
$x_1,\dots,x_{n_0}$ are oversampled.
This ensures that these first 
coordinates are recovered correctly (up to a mean square error of
order $\sigma^2$). As the algorithm is iterated, the
contribution of these first few coordinates is correctly subtracted from all
the measurements, and hence we can effectively eliminate those nodes
from the graph. In the resulting graph, the first few variables are
effectively oversampled and hence the algorithm will reconstruct
their values, up to a mean square error of order $\sigma^2$. As the process is iterated, 
variables are progressively reconstructed, proceeding from left to
right along the node layout.

While the above explains the basic dynamics of AMP reconstruction algorithms
under spatial coupling, a careful consideration reveals that this
picture leaves open several challenging questions. In particular,
why does the overall undersampling factor $\delta$ have to exceed $\uRenyi(p_X)$
for reconstruction to be successful?
Our proof is based on a potential function argument. We will prove
that there exists a potential function for the AMP algorithm, such
that, when  $\delta>\uRenyi(p_X)$, this function has its global
minimum close to exact
 reconstruction. Further, we will prove that, unless this minimum is
 essentially achieved, AMP can always decrease the function. 
This technique is different from the one followed in \cite{KudekarBEC}
for the LDPC codes over the binary erasure channel, and we think it is
of independent interest.

\begin{figure}[!t]
\centering
\includegraphics*[viewport = 60 240 740 480, width = 5in]{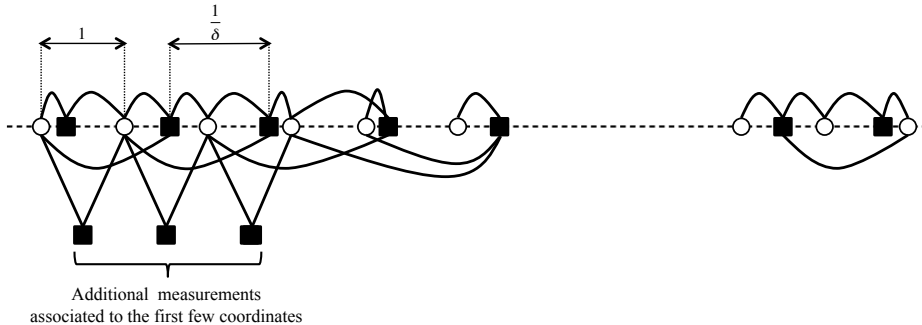}
\caption{\small {Graph structure of a spatially coupled matrix. Variable nodes are shown as circle and check nodes are represented by square.}}  \label{fig:Spatial_graph}
\end{figure}
%
%
%
\vspace{.5cm}
\subsection{Further related work}
\label{sec:Related}

The most closely related earlier work was already discussed above.

More broadly, message passing algorithms for compressed sensing where 
the object of a number of studies studies, starting with
\cite{BaronBayesian}.
As mentioned, we will focus on approximate message passing (AMP) as
introduced in \cite{DMM09,DMM_ITW_I}. As shown in \cite{DonohoJohnstoneMontanari}
these algorithms can be used in conjunction with a rich class of
denoisers $\eta(\,\cdot\,)$. A subset of these denoisers arise as
posterior mean associated to a prior $p_X$. Several interesting
examples were studied by Schniter and collaborators 
\cite{SchniterTurbo,SchniterOFDM,SchniterTree}, and by Rangan and
collaborators  \cite{RanganGAMP,RanganQuantized}.

Spatial coupling has been the object of growing interest within coding
theory over the last few years. The first instance of spatially
coupled code ensembles were the convolutional LDPC codes of Felstr\"om
and Zigangirov \cite{Felstrom}. While the excellent performances of
such codes had been known for quite some time \cite{Costello}, the
fundamental reason was not elucidated until recently \cite{KudekarBEC}
(see also \cite{Lentmaier}). In particular \cite{KudekarBEC}
proved, for communication over the binary erasure channel (BEC), 
that the thresholds of spatially coupled ensembles under message
passing decoding coincide with the thresholds of the base LDPC code
under MAP decoding. In particular, this implies that spatially coupled
ensembles achieve capacity over the BEC. The analogous statement for
general memoryless symmetric channels was first elucidated in
\cite{KudekarBMS} and
finally proved in \cite{kudekar2012spatially}. The paper
\cite{HassaniEtAl}
discusses similar ideas in a number of graphical models.

The first application of spatial coupling ideas to compressed sensing
is due to Kudekar and Pfister \cite{KudekarPfister}. 
They consider a class of sparse spatially coupled sensing matrices, very
similar to parity check matrices for spatially coupled LDPC codes. 
On the other hand, their proposed message passing algorithms do not make use of
the signal distribution $p_X$, and do not fully exploit the
potential of spatially coupled matrices. The message passing algorithm
used here belongs to the general class introduced in \cite{DMM09}. The
specific use of the minimum-mean square error denoiser was suggested
in \cite{DMM_ITW_I}. The same choice is made in \cite{KrzakalaEtAl},
which also considers Gaussian matrices with heteroscedastic entries
although the variance structure is somewhat less general. 

Finally, let us mention that robust sparse recovery of $k$-sparse
vectors from $m = O(k\log\log (n/k))$ measurement is possible,
using suitable `adaptive' sensing schemes \cite{IndykAdaptive}.
%
%
\section{Matrix and algorithm construction}

In this section, we define an ensemble of random matrices, and the
corresponding choices of $Q^t$, $\ons^t$, $\eta_t$ that achieve the
reconstruction guarantees in Theorems \ref{thm:MainTheorem} and \ref{thm:MainTheoremR}.
We proceed by first introducing a general ensemble of random
matrices. Correspondingly, we define a deterministic recursion named
state evolution, that plays a crucial role in the algorithm analysis. 
In Section \ref{sec:GeneralAlgo}, we define the algorithm parameters and
construct specific choices of $Q^t$, $\ons^t$, $\eta_t$. The last section also contains a restatement of
Theorems \ref{thm:MainTheorem} and \ref{thm:MainTheoremR}, in which this construction is made explicit.

\subsection{General matrix ensemble}
\label{sec:MatrixEnsemble}

The sensing matrix $A$ will be constructed randomly, from an 
ensemble denoted by $\Ens(W,M,N)$.  
The ensemble depends on two
integers $M,N\in\naturals$, and on a matrix with non-negative
entries $W\in \reals_+^{\Rows\times \Cols}$,
whose rows and columns are indexed by the finite sets $\Rows$, $\Cols$
(respectively `rows' and `columns').
The band-diagonal structure that is characteristic of spatial coupling
is imposed by a suitable choice of the matrix $W$. In this section we
define the ensemble for a general choice of $W$. In
Section \ref{sec:Choices} we discuss a class of choices for $W$
that corresponds to spatial coupling, and that yields Theorems
\ref{thm:MainTheorem} and \ref{thm:MainTheoremR}.

In a nutshell, the sensing matrix $A$ is obtained from $W$ through a
suitable `lifting' procedure. Each entry $W_{r,c}$ is replaces my an
$M\times N$ block with i.i.d. entries
$A_{ij}\sim\normal(0,W_{r,c}/M)$. Rows and columns of $A$ are then
re-ordered uniformly at random to ensure exchangeability.
For the reader
familiar with the application of spatial coupling to coding theory, it
might be useful to notice the differences and analogies with
graph liftings. In that case, the `lifted' matrix is obtained by replacing each
edge in the base graph with a random permutation matrix.

Passing to the formal definition, we will assume that the matrix $W$ is \emph{roughly row-stochastic}, i.e., 
\begin{eqnarray}\label{eqn:almost_row_stochastic}
\frac{1}{2}\le\sum_{c\in\Cols}W_{r,c} \le 2\, ,\;\;\;\;\;\;\; \mbox{for all
}r\in\Rows\, .
\end{eqnarray}
(This is a convenient simplification for ensuring correct
normalization of $A$.)
We will let $|\Rows|\equiv L_r$ and $|\Cols| = L_c$ denote the 
matrix dimensions.
The ensemble parameters are related to the
sensing matrix dimensions by $n= NL_c$ and $m=ML_r$.

In order to describe a random matrix $A\sim\Ens(W,M,N)$ from this
ensemble, partition the columns and row indices in, respectively,
$L_c$ and $L_r$ groups of equal
size. Explicitly
\begin{align*}
[n] &= \cup_{s\in\Cols}C(s)\, ,\;\;\;\;  |C(s)|= N\, ,\\
[m] &= \cup_{r\in\Rows}R(r)\, ,\;\;\;\;  |R(r)|=M\, .
\end{align*}
Here and below we use $[k]$ to denote the set of first $k$ integers
$[k]\equiv\{1,2,\dots,k\}$. Further, if $i\in R(r)$ or $j\in C(s)$ we
will write, respectively,   $r= \gr(i)$ or $s= \gr(j)$. In other
words $\gr(\,\cdot\, )$ is the operator determining the group index of
a given row or column.

With this notation we have the following concise definition of the
ensemble.
\begin{definition}
A random sensing matrix  $A$ is distributed according to the ensemble
$\Ens(W,M,N)$ (and we write $A\sim \Ens(W,M,N)$) if the partition of
rows and columns  ($[m]= \cup_{r\in\Rows}R(r)$ and $[n] =
\cup_{s\in\Cols}C(s)$) are uniformly random, and given this partitioning, the entries
$\{A_{ij}, \;\; i\in [m], j\in [n]\}$ are independent  Gaussian random
variables with \footnote{As in many papers on compressed sensing, the matrix here has independent zero-mean
Gaussian entries; however, unlike standard practice, here the entries are of widely different variances.}
\begin{eqnarray}\label{eqn:A_W}
A_{ij}\sim \normal\Big(0,\frac{1}{M}\, W_{\gr(i),\gr(j)}\Big)\, .
\end{eqnarray}
\end{definition}
We refer to Fig.~\ref{fig:Amatrix} for an illustration. Note that the
randomness of the partitioning of row and column indices  is only used
in the proof of Lemma \ref{lemma:SE}
(cf. \cite{JM-StateEvolution}), and hence this and other illustrations
assume that the partitions are contiguous. 

For proving Theorem \ref{thm:MainTheorem} and Theorem
\ref{thm:MainTheoremR} we will consider suitable sequences of
 ensembles  $\Ens(W,M,N)$ with undersampling ratio converging to
 $\delta$. While a complete description is given below, let us stress
that we take the limit $M,N\to\infty$ (with $M=N\delta$)
 \emph{before} the limit $L_r,L_c \to \infty$ . Hence, the resulting matrix $A$ is 
essentially dense: the fraction of non-zero entries per row vanishes
only \emph{after} the number of groups goes to $\infty$. 
\begin{figure}[!t]
\centering
\includegraphics*[viewport = 60 60 700 520, width = 5in]{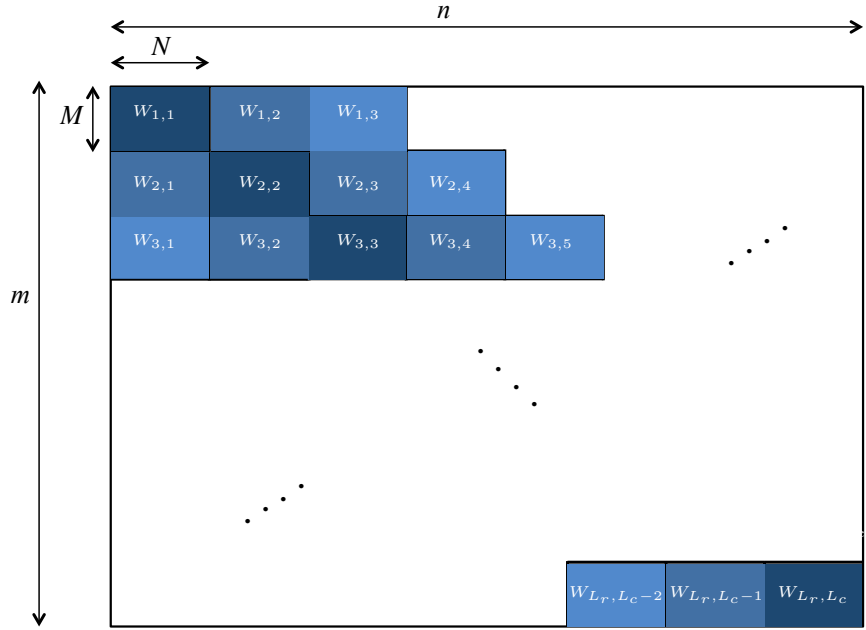}
\put(-292,217){\tiny{\white $W_{1,1}$}}
\put(-252,217){\tiny{\white $W_{1,2}$}}
\put(-215,217){\tiny{\white $W_{1,3}$}}
\put(-292,190){\tiny{\white $W_{2,1}$}}
\put(-252,190){\tiny{\white $W_{2,2}$}}
\put(-215,190){\tiny{\white $W_{2,3}$}}
\put(-180,190){\tiny{\white $W_{2,4}$}}
\put(-292,168){\tiny{\white $W_{3,1}$}}
\put(-252,168){\tiny{\white $W_{3,2}$}}
\put(-215,168){\tiny{\white $W_{3,3}$}}
\put(-180,168){\tiny{\white $W_{3,4}$}}
\put(-142,168){\tiny{\white $W_{3,5}$}}
\put(-145,60){\tiny{\white $W_{\Lr-1,\Lc-3}$}}
\put(-131,60){\tiny{\white $W_{\Lr-1,\Lc-2}$}}
\put(-91,60){\tiny{\white $W_{\Lr-1,\Lc-1}$}}
\put(-50,60){\tiny{\white $W_{\Lr-1,\Lc}$}}
\put(-127,35){\tiny{\white $W_{\Lr,\Lc-2}$}}
\put(-89,35){\tiny{\white $W_{\Lr,\Lc-1}$}}
\put(-48,35){\tiny{\white $W_{\Lr,\Lc}$}}
\caption{\small {Construction of the spatially coupled measurement matrix $A$ as described in Section~\ref{sec:MatrixEnsemble}. The matrix is divided into blocks with size $M$ by $N$. (Number of blocks in each row and each column are respectively $\Lc$ and $\Lr$, hence $m= M \Lr$, $n = N\Lc$). The matrix elements $A_{ij}$ are chosen as $\normal(0,\frac{1}{M}W_{\gr(i),\gr(j)})$. In this figure, $W_{i,j}$ depends only on $|i-j|$ and thus blocks on each diagonal have the same variance.}}  \label{fig:Amatrix}
\end{figure}

%

\subsection{State evolution}
\label{sec:StateEvolution}

State evolution allows an exact asymptotic  analysis of AMP algorithms
in the limit of a large number of dimensions. As indicated by the
name, it bears close resemblance to the  density evolution method in
iterative coding theory \cite{RiU08}. Somewhat surprisingly, this
analysis approach is asymptotically exact despite the underlying factor 
graph being far from locally tree-like.

State evolution was first developed  in \cite{DMM09} on the
basis of heuristic arguments, and substantial numerical evidence. 
Subsequently, it was proved to hold for Gaussian sensing matrices with
i.i.d. entries, and a broad class of iterative algorithm in \cite{BM-MPCS-2011}.
These proofs were further generalized in \cite{RanganGAMP}, to cover
`generalized' AMP algorithms.

In the present case, state evolution takes the following form.
\footnote{In previous work, the state variable concerned a single scalar, representing 
the mean-squared error in the current reconstruction, averaged across all coordinates. 
In this paper, the dimensionality of the state variable is much larger, because it contains 
$\psi$, an individualized MSE for each coordinate of the reconstruction and also 
$\phi$, a noise variance for the residuals $r^t$ for each measurement coordinate. }

\begin{definition}\label{def:SEmapa}
Given $W\in \reals_+^{\Rows\times\Cols}$ roughly row-stochastic, and
$\delta>0$, the
corresponding \emph{state evolution maps}
$\SEmapa_W:\reals_+^{\Rows}\to\reals_+^{\Cols}$,
$\SEmapb_W:\reals_+^{\Cols}\to\reals_+^{\Rows}$, 
 are defined as follows. 
For $\phi=(\phi_a)_{a\in\Rows}\in\reals_+^{\Rows}$,
$\psi=(\psi_i)_{i\in\Cols}\in\reals_+^{\Cols}$, we let:
\begin{eqnarray}
\SEmapa_{W}(\phi)_{i} &= &\mmse\Big(\sum_{b\in\Rows}W_{b,i}\phi_b^{-1}\Big)\, ,\label{eq:Mapa}\\
\SEmapb_{W}(\psi)_{a} &=& \sigma^2+\frac{1}{\delta}\sum_{i\in\Cols }
W_{a,i}\, \psi_i\, . \label{eq:Mapb}
\end{eqnarray}
We finally define $\SEmap_W=\SEmapa_W \circ \SEmapb_W$.
\end{definition}
In the following, we shall omit the subscripts from $\SEmap_{W}$
whenever clear from the context.

\begin{definition}\label{def:SESequence}
Given  $W\in \reals_+^{\Lr\times \Lc}$ roughly row-stochastic, the
corresponding \emph{state evolution sequence} is the sequence of vectors
$\{\phi(t),\psi(t)\}_{t\ge 0}$,
$\phi(t)=(\phi_{a}(t))_{a\in\Rows}\in\reals_+^{\Rows}$,
$\psi(t)=(\psi_{i}(t))_{i\in\Cols}\in\reals_+^{\Cols}$,
 defined
recursively by $\phi(t) = \SEmapb_{W}(\psi(t))$, $\psi(t+1) = \SEmapa_{W}(\phi(t))$, 
with initial condition
\begin{eqnarray} \label{eqn:initial_cond}
\psi_i(0) = \infty  \mbox{ for all } i\in \Cols\, .
\end{eqnarray}
 Hence, for all $t\ge 0$, 
\begin{eqnarray}
\begin{split}\label{eq:ExplicitSE}
\phi_a(t) &=& \sigma^2+\frac{1}{\delta}\sum_{i\in\Cols }
W_{a,i}\, \psi_i(t)\, ,\\
\psi_{i}(t+1) &= &\mmse\Big(\sum_{b\in\Rows}W_{b,i}\phi_b(t)^{-1}\Big)\,. 
\end{split}
\end{eqnarray}
\end{definition}

The quantities $\psi_i(t)$, $\phi_a(t)$ correspond to the asymptotic
MSE achieved by the AMP algorithm. More precisely, $\psi_i(t)$
corresponds to the asymptotic mean square error $\E\{(x_j^t-x_j)^2\}$ for $j\in
C(i)$, as $N\to\infty$. Analogously, $\phi_a(t)$ is the noise variance
in residuals $r^t_j$ corresponding to rows $j\in R(a)$.  This correspondence is stated formally in Lemma \ref{lemma:SE} below.
The state evolution (\ref{eq:ExplicitSE}) describes the evolution of these
quantities. In particular, the linear operation in Eq.~(\ref{eq:AMP2})
    corresponds to a sum of noise variances as per
    Eq.~(\ref{eq:Mapb}) and the application of denoisers $\eta_t$
    corresponds to a noise reduction as per Eq.~(\ref{eq:Mapa}).

As we will see, the definition of denoiser function $\eta_t$ involves the state vector $\phi(t)$. (Notice that the state vectors $\{\phi(t),\psi(t)\}_{t \ge 0}$ can be precomputed). Hence, $\eta_t$ is `tuned' according to the predicted reconstruction error at iteration $t$.   
\bigskip

%
%
\subsection{General algorithm definition}
\label{sec:GeneralAlgo}

In order to fully define the AMP algorithm (\ref{eq:AMP1}),
(\ref{eq:AMP2}),
we need to provide constructions for the matrix $Q^t$,
the nonlinearities $\eta_t$, and  the vector 
$\ons^t$. In doing this, we exploit the fact that the state evolution
sequence $\{\phi(t)\}_{t\ge 0}$ can be precomputed.

We define the matrix $Q^t$ by
\begin{eqnarray}\label{eq:Q_def}
Q_{ij}^t\equiv\frac{\phi_{\gr(i)}(t)^{-1}}{\sum_{k=1}^{\Lr}W_{k,\gr(j)}\phi_{k}(t)^{-1}}\,.
\end{eqnarray}
Notice that $Q^t$ is block-constant: for any $r,s\in [L]$, the block
$Q^{t}_{R(r),C(s)}$ has all its entries equal.

As mentioned in Section \ref{sec:Introduction}, the function
$\eta_t:\reals^n\to\reals^n$ is chosen to be separable, i.e., for
$v\in\reals^N$:
\begin{eqnarray}\label{eq:eta_def1}
\eta_t(v) =
\big(\eta_{t,1}(v_1),\eta_{t,2}(v_2),\;\dots\;,\eta_{t,N}(v_N)\big)\, .
\end{eqnarray}
We take $\eta_{t,i}$ to be a conditional expectation estimator
for $X\sim p_X$ in gaussian noise:
\begin{eqnarray}\label{eq:eta_def2}
\eta_{t,i}(v_i) = \E\big\{X\,\big|\, X+\, s_{\gr(i)}(t)^{-1/2}Z = v_i\,\big\}\,
,\;\;\;\; s_r(t) \equiv \sum_{u \in \Rows} W_{u,r}\phi_u(t)^{-1}\, .
\end{eqnarray}
Notice that the function  $\eta_{t,i}(\,\cdot\,)$ depends on $i$ only
through the group index $\gr(i)$, and in fact only parametrically
through $s_{\gr(i)}(t)$. It is also interesting to notice that the
denoiser $\eta_{t,i}(\,\cdot\,)$ does not have any tuning parameter to
be optimized over. This was instead the case for the soft-thresholding
AMP algorithm studied in \cite{DMM09} for which the threshold level
had to be adjusted in a non-trivial manner to the sparsity level. This
difference is due to the fact that the prior $p_X$ is assumed to be
known and hence the optimal denoiser is uniquely determined to be the
posterior expectation as per Eq.~(\ref{eq:eta_def2}).

Finally, in order to define the vector $\ons^t_i$, let us introduce
the quantity 
\begin{eqnarray}
\<\eta'_{t}\>_u = \frac{1}{N}\sum_{i\in
  C(u)}\eta'_{t,i}\big(x^t_i+((Q^t\odot A)^*r^t)_i\big)\, .
\end{eqnarray}
The vector $\ons^t$ is then defined by
\begin{eqnarray}
\ons^t_i \equiv
\frac{1}{\delta}\sum_{u \in \Cols} W_{\gr(i),u}\tQ^{t-1}_{\gr(i),u}\,
\<\eta'_{t-1}\>_u\, ,\label{eq:OnsagerDef}
\end{eqnarray}
where we  defined $Q^t_{i,j} = \tQ^t_{r,u}$ for $i\in R(r)$, $j\in
C(u)$. Again $\ons^t_i$ is block-constant: the vector $\ons^t_{C(u)}$
has all its entries equal.

This completes our definition of the AMP algorithm. Let us conclude with a few computational remarks:
\begin{enumerate}
\item The quantities $\tilde{Q}^t$, $\phi(t)$ can be precomputed efficiently
  iteration by iteration, because they are, respectively, 
$\Lr\times \Lc$ and $\Lr$-dimensional, and, as discussed further
below, $\Lr,\Lc$ are much smaller than $m,n$. The most complex part of
this computation is implementing  the iteration (\ref{eq:ExplicitSE}),
which has complexity $O((\Lr+\Lc)^3)$, plus the complexity of
evaluating the $\mmse$ function, which is a one-dimensional integral.
\item The vector $\ons^t$ is also block-constant, so can be efficiently
  computed using Eq.~(\ref{eq:OnsagerDef}).
\item Instead of computing $\phi(t)$ analytically by iteration (\ref{eq:ExplicitSE}),
  $\phi(t)$ can also be estimated from data $x^t, r^t$.
In particular, by generalizing the methods introduced in
\cite{DMM09,MontanariChapter}, we get  the estimator
\begin{eqnarray}
\widehat{\phi}_a(t) = \frac{1}{M}\, \|r^t_{R(a)}\|^2_2\, ,
\end{eqnarray}
where $r^t_{R(a)} = (r^t_j)_{j\in R(a)}$ is the restriction of $r^t$
to the indices in $R(a)$.  An alternative more robust estimator (more resilient to outliers),
would be 
\begin{eqnarray}
\widehat{\phi}_a(t)^{1/2} = \frac{1}{\Phi^{-1}(3/4)}\, |r^t_{R(a)}|_{(M/2)}\, ,
\end{eqnarray}
where $\Phi(z)$ is the Gaussian distribution function, and, for $v\in\reals^K$,
$|v|_{(\ell)}$ is the $\ell$-th largest entry in the vector
$(|v_1|,|v_2|,\dots,|v_K|)$.
(See, e.g., \cite{HuberBook} for background in robust estimation.)
The idea underlying both of the above estimators is that the components
of $r^t_{R(a)}$ are asymptotically i.i.d. with mean zero and variance $\phi_a(t)$.
\end{enumerate}
%
%
\subsection{Choices of parameters, and spatial coupling}
\label{sec:Choices}

In order to prove our main Theorem \ref{thm:MainTheorem},
we use a sensing matrix from the ensemble $\Ens(W,M,N)$ for a suitable
choice of the matrix $W\in \reals^{\Rows\times \Cols}$.
Our construction depends on parameters $\rho\in\reals_+$,
$L,L_0\in\naturals$, and on the `shape function' $\Shape$.
As explained below, $\rho$ will be taken to be small, and hence
we will treat $1/\rho$ as an integer to avoid rounding (which
introduces in any case a negligible error).

Here and below $\cong$ denotes identity between two sets up to a relabeling.
\begin{definition}
A \emph{shape function} is a function $\Shape:\reals\to\reals_+$ 
 continuously differentiable, with support in $[-1,1]$  and such that $\int_\reals
\Shape(u)\,\de u = 1$, and $\Shape(-u) =\Shape(u)$.
\end{definition}

We let  $\Cols \cong \{-2\rho^{-1},\dots,0,1,\dots, L-1\}$, so that 
$\Lc = L+2\rho^{-1} $. Also let $\Cols_0 = \{0,1,\dotsc, L-1\}$. 

The rows are partitioned as follows:
\begin{eqnarray*}
\Rows = \Rows_0 \cup \Big\{\cup_{i=-2\rho^{-1}}^{-1} \Rows_{i}\Big\}\,,
\end{eqnarray*} 
where $\Rows_0\cong\{-\rho^{-1},\dots,0,1,\dots ,L-1+\rho^{-1}\}$, and
$\Rows_i = \{i\L0, \dotsc, (i+1) \L0-1\}$, for $i= -2\rho^{-1}, \dotsc, -1$. Hence, $|\Rows_i|=L_0$, and $\Lr = \Lc+2\rho^{-1}\L0$.

Finally, we take $N$ so that $n= N\Lc$,  
and let $M = N\delta$ so that $m = M \Lr = N(\Lc+2\rho^{-1} \L0)\delta$.
Notice that $m/n = \delta(\Lc+ 2\rho^{-1} \L0)/\Lc$. Since we will take $\Lc$ much larger
than $L_0/\rho$, we in fact have $m/n$ arbitrarily close to $\delta$.

Given these inputs, we construct the corresponding matrix
$W = W(L,L_0,\Shape,\rho)$ as follows.
\begin{enumerate}
\item For $i\in\{-2\rho^{-1},\dots,-1\}$,
and each $a\in\Rows_i$, we let $W_{a,i}=1$. Further, 
  $W_{a,j} = 0$ for all $j\in\Cols\setminus\{i\}$.
\item For  all $a\in\Rows_0\cong\{-\rho^{-1},\dots,0,\dots,
  L-1+\rho^{-1}\}$, we let
\begin{eqnarray}
W_{a,i} = 
\rho\, \Shape\big(\rho\,(a-i)\big)\, & i \in \{-2\rho^{-1},\dots,L-1\}.
\end{eqnarray}
%
%
%
%
\end{enumerate}
\begin{figure}[!t]
\centering
\begin{picture}(80,250)
\put(-170,-15){\includegraphics[width = 5.7in]{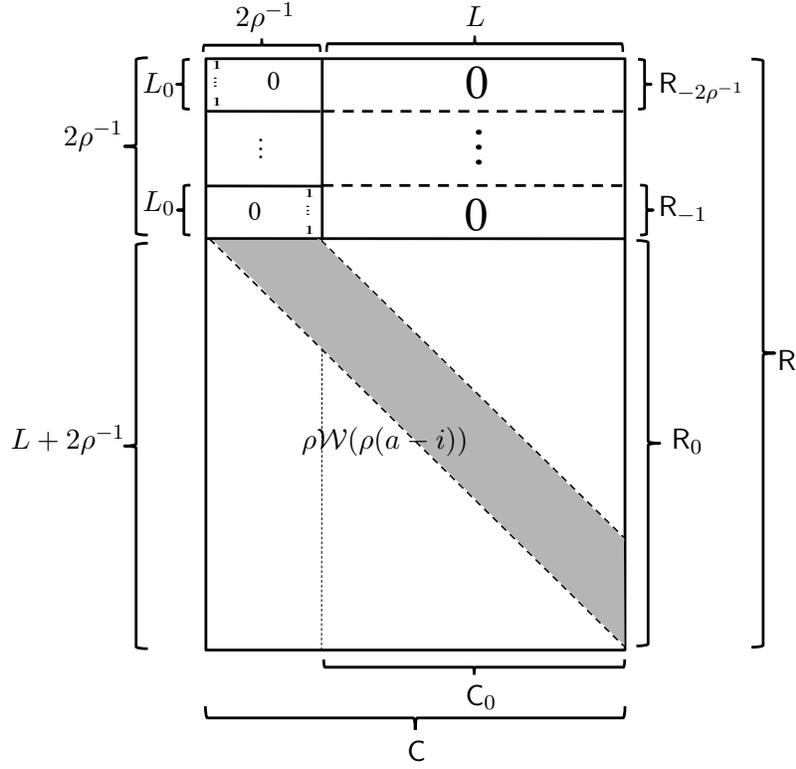}}
\put(-110, 110){$L+2\rho^{-1}$}
\put(-90, 225){$2\rho^{-1}$}
\put(-60, 245){$L_0$}
\put(-60, 200){$L_0$}
\put(-25, 270){$2\rho^{-1}$}
\put(62, 270){$L$}
\put(135, 245){$\Rows_{-2\rho^{-1}}$}
\put(135, 198){$\Rows_{-1}$}
\put(140, 110){$\Rows_0$}
\put(62, 13){$\Cols_0$}
\put(40, -8){$\Cols$}
\put(180, 140){$\Rows$}
\put(0, 110){$\rho \Shape(\rho(a-i))$}
\end{picture}
\vspace{0.5cm}
\caption{\small Matrix W. The shaded region indicates the non zero entries in the lower part of the matrix. As shown (the lower part of ) the matrix $W$ is band diagonal.}  \label{fig:Wmatrix}
\end{figure}
The role of the rows in $\Big\{\cup_{i = -2\rho^{-1}}^{-1} \Rows_i \Big\}$ and the corresponding 
rows in $A$ are to oversample the first few (namely the first $2\rho^{-1} N$) coordinates of the signal
as explained in Section~\ref{sec:Intuition}. Furthermore, the restriction of $W$ to the rows in $\Rows_0$ 
is band diagonal as $\Shape$ is supported on $[-1,1]$.
See Fig.~\ref{fig:Wmatrix} for an illustration of the matrix $W$. 

In the following we occasionally use the shorthand $W_{a-i}\equiv \rho\, \Shape\big(\rho\,(a-i)\big)$.
Note that $W$ is roughly row-stochastic. Also, the restriction of $W$ to the rows in $\Rows_0$ is roughly column-stochastic. This follows from the fact that the function $\Shape(\cdot)$ has continuous (and thus bounded) derivative on the compact interval $[-1,1]$, and $\int_{\reals} \Shape(u) \de u = 1$. Therefore, using the standard convergence of Riemann sums to Riemann integrals and the fact that $\rho$ is small, we get the result. 

We are now in position to restate Theorem \ref{thm:MainTheorem}
in  a more explicit form. 
\begin{thm}\label{thm:MainTheoremNew}
Let $p_X$ be a probability measure on the real line with
$\delta>\uRenyi(p_X)$, and let $\Shape:\reals\to\reals_+$ be  a shape function.
For any $\ve>0$, there exist $L_0, L, \rho$, $t_0$, $\sigma_0^2= \sigma_0(\ve, \delta, p_X)^2$
such that $L_0/(L \rho)\le \ve$, and further the following holds true
for $W = W(L,L_0,\Shape,\rho)$.

For $N\ge 0$, and $A(n)\sim \Ens(W,M,N)$ with
$M=N\delta$, and for all $\sigma^2\le\sigma_0^2$,
$t\ge t_0$, we almost surely have
\begin{eqnarray}
\underset{N\to\infty}{\limsup}\frac{1}{n}\big\|x^t\big(A(n);y(n)\big)-x(n)\big\|^2\le
\ve\, .
\end{eqnarray}
Further, under the same assumptions, we have
\begin{eqnarray}
\underset{N\to\infty}{\limsup}\frac{1}{n}\E\big\{\big\|x^t\big(A(n);y(n)\big)-x(n)\big\|^2\big\}\le
\ve\, .
\end{eqnarray}
\end{thm}

In order to obtain a stronger form of robustness, as per Theorem
\ref{thm:MainTheoremR}, 
we slightly modify the sensing scheme.  
We construct the sensing matrix $\tilde{A}$ from $A$ by appending $2 \rho^{-1}L_0$ rows in the bottom.
\begin{eqnarray}
\tilde{A} = \begin{pmatrix}
\quad \quad A\\
\hline
\multicolumn{1}{c|} 0 & I
\end{pmatrix},
\end{eqnarray} 
where $I$ is the identity matrix of dimensions $2\rho^{-1}L_0$. Note that this corresponds to increasing the number of measurements; however, the asymptotic undersampling rate remains $\delta$, provided that $L_0/(L\rho) \to 0$, as $n \to \infty$.

The reconstruction scheme is modified as follows. Let $x_1$ be the
vector obtained by restricting $x$ to entries in $\cup_{i} C(i)$, where $i\in
\{-2\rho^{-1},\cdots,L-2\rho^{-1} -1\}$. Also, let $x_2$ be the vector
obtained by restricting $x$ to entries in $\cup_{i} C(i)$, where $i \in
\{L-2\rho^{-1}, \cdots, L-1\}$. Therefore, $x = (x_1,x_2)^T$. Analogously,
let $y = (y_1,y_2)^T$ where $y_1$ is given by the restriction  of $y$ 
to  $\cup_{i\in\Rows} R(i)$ and $y_2$ corresponds to the additional
$2\rho^{-1}L_0$ rows.
Define $w_1$ and $w_2$ from the noise vector $w$, analogously. Hence,
\begin{eqnarray}
\begin{pmatrix}
y_1\\y_2
\end{pmatrix} = 
\begin{pmatrix}
\quad \quad A\\
\hline
\multicolumn{1}{c|} 0 & I
\end{pmatrix}
\begin{pmatrix}
x_1\\x_2
\end{pmatrix} +
\begin{pmatrix}
w_1\\w_2
\end{pmatrix} .
\end{eqnarray}

Note that the sampling rate for vector $x_2$ is one, i.e., $y_2$ and $x_2$ are of the same length and are related to each other through the identity matrix $I$. Hence, we have a fairly good approximation of these entries. We use the AMP algorithm as described in the previous section to obtain an estimation of $x_1$. Formally, let $x^t$ be the estimation at iteration $t$ obtained by applying the AMP algorithm to the problem $y_1 = Ax + w_1$. The modified estimation is then $\tilde{x}^t = (x_1^t, y_2)^T$. 

As we will see later, this modification in the sensing matrix and
algorithm, while not necessary, simplifies some technical steps in the proof.
\begin{thm}\label{thm:MainTheoremNewR}
Let $p_X$ be a probability measure on the real line with
$\delta>\uMMSE(p_X) $, and let $\Shape:\reals\to\reals_+$ be  a shape function.
There exist $L_0, L, \rho$, $t_0$ and a finite stability constant $C = C(p_X,\delta)$,
such that $L_0/(L \rho) < \ve$, for any given $\ve > 0$, and the following holds true for the modified reconstruction scheme.

For $t\ge t_0$, we almost surely have,
\begin{eqnarray}
\underset{N\to\infty}{\limsup}\frac{1}{n} \big\|\tilde{x}^t\big(\tilde{A}(n);y(n)\big)-x(n)\big\|^2\le
C\sigma^2.
\end{eqnarray}
Further, under the same assumptions, we have
\begin{eqnarray}
\underset{N\to\infty}{\limsup}\frac{1}{n} \E\big\{\big\|\tilde{x}^t\big(\tilde{A}(n);y(n)\big)-x(n)\big\|^2 \big\}\le
C\sigma^2.
\end{eqnarray}
Finally, in the asymptotic case where $\ell = L\rho \to \infty$, $\rho \to 0$, $\L0 \to \infty$, we have
 \[
 \lim_{\sigma \to 0} \frac{1}{\sigma^2} \Big\{ \lim_{t\to \infty} \underset{N\to\infty}{\limsup}\frac{1}{n} \big\|\tilde{x}^t\big(\tilde{A}(n);y(n)\big)-x(n)\big\|^2 \Big\}
  \le \frac{4\delta- 2\uMMSE(p_X)}{\delta - \uMMSE(p_X)} \,.
 \]
\end{thm}
It is obvious that Theorems~\ref{thm:MainTheoremNew} and
\ref{thm:MainTheoremNewR} respectively imply Theorems
\ref{thm:MainTheorem} and \ref{thm:MainTheoremR}. We shall therefore
focus on the proofs of Theorems~\ref{thm:MainTheoremNew} and
\ref{thm:MainTheoremNewR} in the rest of the paper. 

Notice that the results of Theorems~\ref{thm:MainTheoremNew} and~\ref{thm:MainTheoremNewR}
only deal with a linear subsequence $n = N\Lc$ with $N\to \infty$. However, this is sufficient to prove the claim
of Theorems~\ref{thm:MainTheorem} and \ref{thm:MainTheoremR}. More specifically,
suppose that $n$ is not a multiple of $\Lc$. Let $n'$ be the smallest number greater than $n$ which is divisible by $\Lc$,
i.e., $n'= \lceil n/\Lc\rceil \Lc$, and let $\hat{x} = (x,0)^T \in \reals^{n'}$ be obtained by padding $x$ with zeros.
Let $\hat{x}^t$ denote the Bayes optimal AMP estimate of $\hat{x}$ and $x^t$ be the restriction of $\hat{x}^t$ to the first $n$ entries. 
We have $(1/n) \|x^t - x\|^2 \le (n'/n) (1/n') \|\hat{x}^t - \hat{x}\|^2$. The result of Theorem~\ref{thm:MainTheorem} follows by applying Theorem~\ref{thm:MainTheoremNew} (for the sequence $n = N\Lc$, $N \to \infty$), and noting that $n'/n \le (1+ \Lc/n) \to 1$, as $N \to \infty$. 
Similar comment applies to Theorems~\ref{thm:MainTheoremNewR} and~\ref{thm:MainTheoremR}.

\section{Advantages of spatial coupling}

Within the construction proposed in this paper, spatially coupled
sensing matrices have independent heteroscedastic entries (entries
with different variances). In addition to this, we  also oversample a
few number of  coordinates of the signal, namely the first
$2\rho^{-1}N$ coordinates. In this section we informally discuss the
various components of this scheme.

It can be instructive to compare this
construction with the case of homoscedastic Gaussian matrices
(i.i.d. entries).
For the reader familiar with coding theory, this comparison is
analogous to the comparison between regular LDPC codes and spatially
coupled regular LDPC codes. Regular LDPC codes have been known since
Gallager \cite{GallagerThesis,MMRU05} to achieve the channel capacity, as the degree gets
large, under maximum likelihood decoding. However their performances under practical (belief propagation)
decoding is rather poor. When the code ensemble is modified via
spatial coupling, the belief propagation performances improve to
become asymptotically equivalent to the maximum likelihood
performances.  Hence spatially coupled LDPC codes achieve capacity
under practical decoding schemes.

Similarly, standard (non-spatially coupled) sensing matrices achieve
the information theoretic limit under computationally unpractical
recovery schemes \cite{WuVerdu}, but do not perform ideally under
practical reconstruction algorithms. Consider for instance Bayes optimal AMP. Within the standard ensemble,
the state evolution recursion reads
\begin{eqnarray}\label{eq:ExplicitSE_standard2}
\begin{split}
\phi(t) &=\sigma^2+\frac{1}{\delta}\psi(t)\, ,\\
\psi(t+1) &= \mmse\big(\phi(t)^{-1}\big)\,. 
\end{split}
\end{eqnarray}
Let $\tilde{\delta}(p_X)\equiv \sup_{s \ge 0} s\cdot \mmse(s)>\uRenyi(p_X)$.
It is immediate to see that the last recursion develops two (or
possibly more) stable
fixed points for $\delta<\tilde{\delta}(p_X)$ and all $\sigma^2$ small
enough. The smallest fixed point, call it $\phi_{\rm good}$, corresponds to correct
reconstruction and  is such that $\phi_{\rm good}= O(\sigma^2)$ as
$\sigma\to 0$. The
largest fixed point, call it $\phi_{\rm bad}$, corresponds to incorrect
reconstruction and is such that $\phi_{\rm bad}=\Theta(1)$ as
$\sigma\to 0$. A study of the above
recursion shows that $\lim_{t\to\infty}\phi(t) =\phi_{\rm bad}$. State
evolution converges to the `incorrect' fixed point, hence predicting a
large MSE for AMP.  

On the contrary, for
$\uRenyi(p_X)<\delta<\tilde{\delta}(p_X)$
the recursion~\eqref{eq:ExplicitSE_standard2} converges (for
appropriate choices of $W$ as in the previous section) to the `ideal'
fixed point $\lim_{t\to\infty}\phi_a(t)= \phi_{\rm good}$ for all $a$
(except possibly those near the boundaries). This is
illustrated in Fig.~\ref{fig:phi-profile}. We also refer to
\cite{HassaniEtAl} for a survey of examples of the same phenomenon and
to \cite{KrzakalaEtAl,JavanmardMon12} for further discussion in
compressed sensing.

The above discussion also clarifies why the posterior expectation
denoiser is useful. Spatially coupled sensing matrices do not yield
better performances than the ones dictated by the best fixed point in
the `standard' recursion (\ref{eq:ExplicitSE_standard2}). In
particular, replacing the Bayes optimal denoiser by another denoiser
$\eta_t$ amounts, roughly, to replacing $\mmse$ in
Eq.~\eqref{eq:ExplicitSE_standard2} by the MSE of another denoiser,
hence leading to worse performances.

In particular, if the posterior expectation denoiser is replaced by
soft thresholding, the resulting state evolution recursion always has
a unique stable fixed point for homoscedastic matrices
\cite{DMM09}. This suggests that spatial coupling does not lead to any
improvement for soft thresholding AMP and hence (via the
correspondence of \cite{BayatiMontanariLASSO}) for LASSO or $\ell_1$
reconstruction. This expectation is indeed confirmed numerically in \cite{JavanmardMon12}.
%
%
\section{Key lemmas and proof of the main theorems}
\label{sec:Lemmas}

Our proof is based in a crucial way on state evolution. This
effectively reduces the analysis of the algorithm (\ref{eq:AMP1}), (\ref{eq:AMP2}) to 
the analysis of the deterministic recursion (\ref{eq:ExplicitSE}).
\begin{lemma}\label{lemma:SE}
Let $W\in\reals_+^{\Rows\times\Cols}$ be a roughly row-stochastic matrix (see Eq.~\eqref{eqn:almost_row_stochastic})and $\phi(t)$, $Q^t$, $\ons^t$ be defined as in Section
\ref{sec:GeneralAlgo}.
Let $M=M(N)$ be such that $M/N\to\delta$, as $N \to \infty$. Define $m=M\Lr$, $n=N\Lc$,
and for each $N\ge 1$, let $A(n)\sim \Ens(W,M,N)$. 
Let $\{(x(n),w(n))\}_{n\ge 0}$ be a converging sequence of instances
with parameters $(p_X,\sigma^2)$.
Then, for all $t\ge 1$, almost surely we have 
\begin{eqnarray}
\underset{N\to\infty}{\limsup}\frac{1}{N}\|x^t_{C(i)}(A(n);y(n))-x_{C(i)}\|_2^2=
\mmse\Big(\sum_{a\in \Rows}W_{a,i}\phi_a(t-1)^{-1} \Big)\, .
\end{eqnarray}
for all $i \in \Cols$.
\end{lemma}
This lemma is a straightforward generalization of
\cite{BM-MPCS-2011}. Since a formal proof does not require new ideas,
but a significant amount of new notations, it is presented in  a separate
publication \cite{JM-StateEvolution} which covers an even
more general setting.
In the interest of self-containedness, and to develop useful intuition
on  state evolution, we present an heuristic derivation of
the state evolution equations (\ref{eq:ExplicitSE}) in Section \ref{sec:SE_Heuristics}.

The next Lemma provides the needed analysis of the recursion 
(\ref{eq:ExplicitSE}).
\begin{lemma}\label{lem:phi_convergence}
Let $\delta>0$, and $p_X$ be a probability measure on the real line. Let $\Shape:\reals\to\reals_+$ be  a shape function.

$(a)$ If $\delta>\uRenyi(p_X)$, then for any $\ve>0$, there exist $\sigma_0 = \sigma_0(\ve, \delta, p_X),\rho, L_* > 0$,
such that for any $\sigma^2\in [0,\sigma_0^2], L_{0}> 3/\delta$, and $L > L_*$, 
the following holds for $W = W(L,L_0,\Shape,\rho)$:
\begin{eqnarray}
\lim_{t\to\infty} \frac{1}{L} \sum_{a = -\rho^{-1}}^{L +\rho^{-1} - 1} \phi_a(t) \le \ve.
\end{eqnarray}

$(b)$ If further $\delta>  \uMMSE(p_X) $, then there exist $\rho,L_* > 0$, and 
a finite stability constant $C= C(p_X,\delta)$, such that for $L_0 > 3/\delta$, and $L >L_*$,
the following holds for $W = W(L,L_0,\Shape,\rho)$.
\begin{eqnarray}
\lim_{t\to\infty} \frac{1}{L} \sum_{a = -\rho^{-1}}^{L - \rho^{-1} - 1} \phi_a(t) \le C\sigma^2.
\end{eqnarray}
Finally, in the asymptotic case where $\ell = L\rho \to \infty$, $\rho \to 0$, $\L0 \to \infty$, we have
 \begin{eqnarray}\label{eqn:phi-C-bound}
 \lim_{\sigma \to 0} \lim_{t\to \infty} \frac{1}{\sigma^2\, L} \sum_{a=-\rho^{-1}}^{L-\rho^{-1}-1} \phi_a(t) \le 
 \frac{3\delta- \uMMSE(p_X)}{\delta - \uMMSE(p_X)} \,.
 \end{eqnarray}
\end{lemma}
The proof of this lemma is deferred to Section \ref{sec:AnalysisLemma}
and is indeed the technical core of the paper.

Now, we have in place all we need to prove our main results.
\begin{proof}[Proof (Theorem \ref{thm:MainTheoremNew})]
Recall that $\Cols \cong \{-2\rho^{-1}\cdots,L-1\}$. Therefore,
\begin{eqnarray}
\begin{split}
\underset{N\to\infty}{\limsup}\frac{1}{n} \big\|x^t\big(A(n);y(n)\big)-x(n)\big\|^2 &\le
\frac{1}{\Lc} \sum_{i \in \Cols} \underset{N\to\infty}{\limsup}\frac{1}{N} \big\|x^t_{C(i)}\big(A(n);y(n)\big)-x_{C(i)}(n)\big\|^2\\
& \stackrel{(a)}{\le} \frac{1}{\Lc} \sum_{i= -2\rho^{-1}}^{L-1} \mmse \left( \sum_{a \in \Rows}W_{a,i} \phi_a(t-1)^{-1}\right)\\
&\stackrel{(b)}{\le} \frac{1}{\Lc} \sum_{i= -2\rho^{-1}}^{L-1} \mmse \left( \sum_{a \in \Rows_0}W_{a,i} \phi_a(t-1)^{-1}\right)\\
&\stackrel{(c)}{\le} \frac{1}{\Lc} \sum_{i= -2\rho^{-1}}^{L-1} \mmse \left( \frac{1}{2} \phi_{i+\rho^{-1}}(t-1)^{-1}\right)\\
&\stackrel{(d)}{\le} \frac{1}{\Lc} \sum_{a= -\rho^{-1}}^{L+\rho^{-1}-1}  2\phi_{a}(t-1).
\end{split}
\end{eqnarray}
Here, $(a)$ follows from Lemma~\ref{lemma:SE}; $(b)$ follows from the
fact that $\mmse$ is non-increasing; (c) holds because of the following facts: $(i)$ $\phi_a(t)$ is nondecreasing 
in $a$ for every $t$ (see Lemma~\ref{lemma:NonDecreasing} below). $(ii)$ Restriction of $W$ to the rows in $\Rows_0$ is roughly column-stochastic. $(iii)$ $\mmse$ is non-increasing; $(d)$ follows from the inequality $\mmse(s) \le 1/s$. The result is immediate due to Lemma~\ref{lem:phi_convergence}, Part $(a)$.

Now, we prove the claim regarding the expected error. Let $f_n = \frac{1}{n} \|x^t(A(n);y(n)) - x(n)\|^2$. Since $\underset{n \to \infty}{\lim \sup}\, f_n \le \ve$, there exists $n_0$ such that $f_n \le 2 \ve$ for $n \ge n_0$. Applying reverse Fatou's lemma to the bounded sequence $\{f_n\}_{n \ge n_0}$, we have $\underset{N \to \infty}{\lim \sup}\, \E f_n  \le \E [\underset{N \to \infty}{\lim \sup}\, f_n]  \le \ve$.     
\end{proof}
\begin{proof}[Proof (Theorem \ref{thm:MainTheoremNewR})]
The proof proceeds in a similar manner to the proof of Theorem~\ref{thm:MainTheoremNew}. 
\begin{eqnarray}\label{eqn:proofMTR}
\begin{split}
\underset{N\to\infty}{\limsup}\frac{1}{n} \big\|\tilde{x}^t\big(\tilde{A}(n)&;y(n)\big)-x(n)\big\|^2\\
 &\le \frac{1}{\Lc} \Big\{ \sum_{i = - 2\rho^{-1}}^{L - 2\rho^{-1}-1} \underset{N\to\infty}{\limsup} \frac{1}{N} \big\|x^t_{C(i)}\big(A(n);y(n)\big)-x_{C(i)}(n)\big\|^2 + \lim_{N\to\infty}\frac{1}{N} \big\|w_2(n)\big\|^2 \Big\}\\
& \le \frac{1}{\Lc} \Big\{\sum_{i= -2\rho^{-1}}^{L-2\rho^{-1}-1} \mmse \left( \sum_{a \in \Rows}W_{a,i} \phi_a(t-1)^{-1}\right) + \lim_{N\to\infty}\frac{1}{N} \big\| w_2(n)\big\|^2 \Big\}\\
& \le \frac{1}{\Lc} \Big\{\sum_{i= -2\rho^{-1}}^{L-2\rho^{-1}-1} \mmse \left( \sum_{a \in \Rows_0}W_{a,i} \phi_a(t-1)^{-1}\right) + \lim_{N\to\infty}\frac{1}{N} \big\| w_2(n)\big\|^2 \Big\}\\
& \le \frac{1}{\Lc} \Big\{\sum_{i= -2\rho^{-1}}^{L-2\rho^{-1}-1} \mmse \left( \frac{1}{2} \phi_{i+\rho^{-1}}(t-1)^{-1}\right)
+  \lim_{N\to\infty}\frac{1}{N} \big\|w_2(n)\big\|^2 \Big\}\\
&\le \frac{1}{\Lc} \Big\{ \sum_{a= -\rho^{-1}}^{L-\rho^{-1}-1}  2\phi_{a}(t-1)
+  \lim_{N\to\infty}\frac{1}{N} \big\|w_2(n)\big\|^2 \Big\} \le C\, \sigma^2,
\end{split}
\end{eqnarray}
where the last step follows from Part $(b)$ in Lemma~\ref{lem:phi_convergence}, and Part $(b)$ in Definition~\ref{def:Converging}.

The claim regarding the expected error follows by a similar argument to the one in the proof of Theorem~\ref{thm:MainTheoremNew}.

Finally, in the asymptotic case, where $\ell =L\rho \to \infty$, $\L0 \to \infty$, $\rho \to 0$, we have $\sum_{a\in \Rows_0} W_{a,i} = \sum_{a\in \Rows_0} \rho \Shape(\rho(a-i)) \to \int \Shape(u)\, \de u = 1$, and using Eq.~\eqref{eqn:phi-C-bound} in Eq.~\eqref{eqn:proofMTR}, we obtain the desired result.
\end{proof}
%
%
\section{Numerical experiments}
We consider a Bernoulli-Gaussian distribution $p_{X} = (1-\ve)\delta_0 +
\ve\, \gamma_{0,1}$. Recall that $\gamma_{\mu,\sigma}(\de x) = (2\pi \sigma^2)^{-1/2} \exp\{-(x-\mu)^2/(2\sigma^2)\} \de x$.
 We construct a random signal
$x(n) \in \reals^n$ by sampling i.i.d. coordinates
$x(n)_i \sim p_{X}$. We have $\uRenyi(p_{X}) = \ve$ by Proposition~\ref{propo:Renyi} and
\begin{eqnarray}
\eta_{t,i} (v_i) = \frac{\ve \gamma_{1+s_{\gr(i)}^{-1}}(v_i)}{\ve \gamma_{1+s_{\gr(i)}^{-1}}(v_i) + (1-\ve) \gamma_{s_{\gr(i)}^{-1}}(v_i)}\cdot \frac{1}{1+s_{\gr(i)}^{-1}} v_i.
\end{eqnarray}
In the experiments, we use $\ve = 0.1$, $\sigma = 0.01$, $\rho = 0.1$, $M = 6$, $N = 50$, $L = 500$,
$L_0 = 5$. 
\subsection{Evolution of the AMP algorithm}

Our first set of experiments aims at illustrating the evolution of the profile $\phi(t)$
defined by state evolution versus iteration $t$, and comparing the predicted errors by the state evolution
with the empirical errors.

Figure~\ref{fig:phi-profile} shows the evolution of profile $\phi(t) \in \reals^{\Lr}$, given by the state
evolution recursion~\eqref{eq:ExplicitSE}. As explained in Section~\ref{sec:Intuition}, in the spatially coupled sensing matrix, additional
measurements are associated to the first few coordinates of $x$,
namely, $2\rho^{-1} N = 1000$ first coordinates.
This ensures that the values of these coordinates are recovered up to a mean
 square error of order $\sigma^2$. This is reflected
in the figure as the profile $\phi$ becomes of order $\sigma^2$ on the first few
entries after a few iterations (see $t=5$ in the figure). As the
iteration proceeds, the contribution of these components is correctly
subtracted from all the measurements, and essentially they are removed
from the problem. Now, in the resulting problem the first few
variables are effectively oversampled and the algorithm reconstructs
their values up to a mean square error of $\sigma^2$. Correspondingly,
the profile $\phi$ falls to a value of order $\sigma^2$ in the next
few coordinates. As the process is iterated, all the variables are
progressively reconstructed and the profile $\phi$ follows a traveling wave with  constant
velocity. After a sufficient number of iterations ($t = 800$ in the
figure), $\phi$ is uniformly of order $\sigma^2$.
\begin{figure}[!t]
\centering
\includegraphics*[viewport = -10 190 610 610, width =
3.5in]{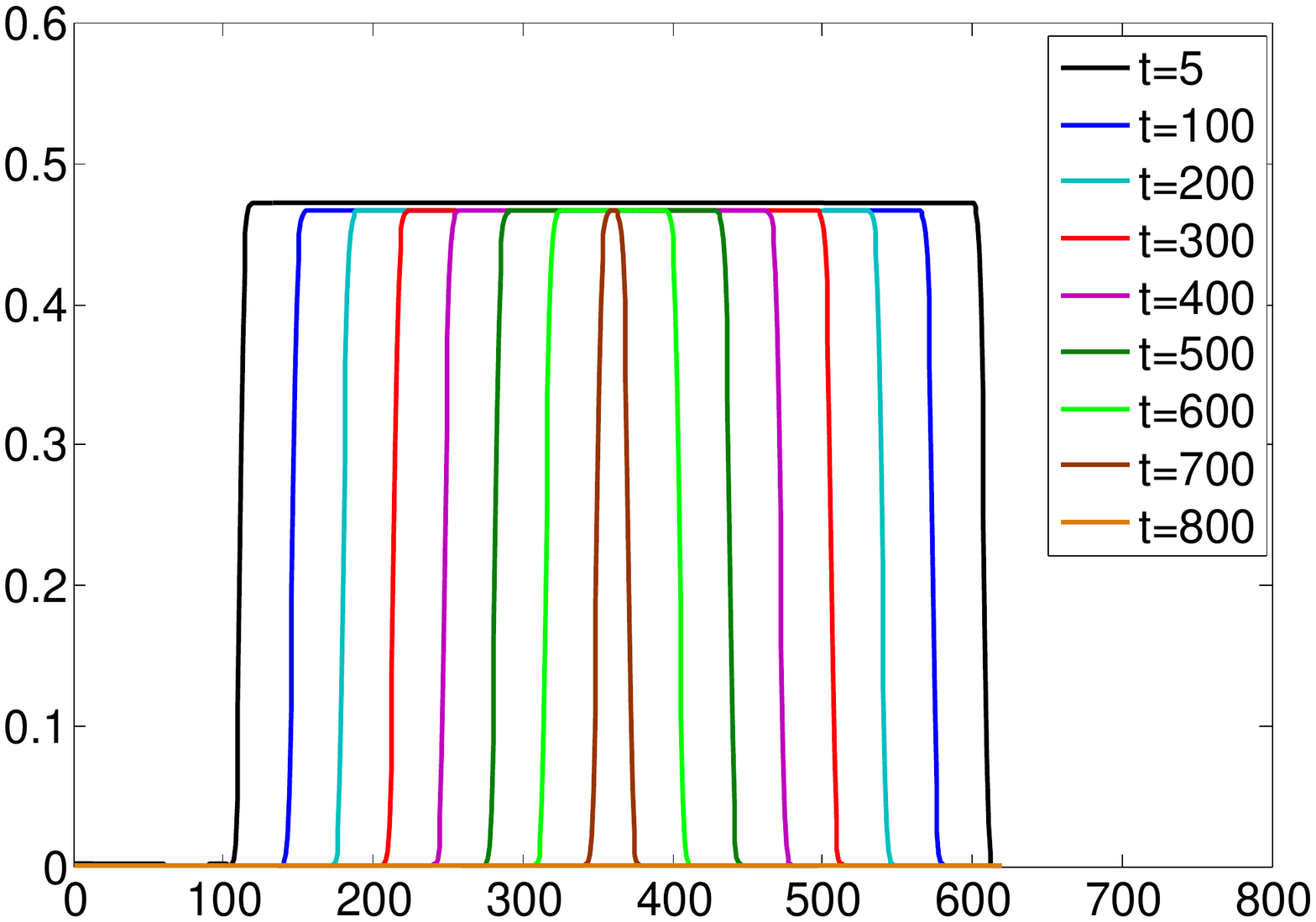}
\put(-120,-9){$a$}
\put(-265,88){$\phi_a(t)$}
\caption{Profile $\phi_a(t)$ versus $a$ for several iteration numbers.}\label{fig:phi-profile}
\end{figure} 

Next, we numerically verify that the deterministic state evolution recursion predicts
the performance of the AMP at each iteration. Define the
empirical and the predicted mean square errors respectively by
\begin{eqnarray}
\MSE_{\rm AMP}(t)&=& \frac{1}{n} \|x^t(y) - x\|_2^2,\\
\MSE_{\rm SE}(t) & = &\frac{1}{\Lc} \sum_{i\in \Cols} \mmse\Big( \sum_{a \in \Rows} W_{a,i} \phi_a^{-1}(t-1)\Big).
\end{eqnarray}

\noindent The values of $\MSE_{\rm AMP}(t)$ and $\MSE_{\rm SE}(t)$ 
are depicted versus $t$ in Fig.~\ref{fig:SE-AMP}. (Values of $\MSE_{\rm
  AMP(t)}$ and the error bars correspond to $M = 30$ Monte Carlo
instances). This verifies that the state evolution provides an
iteration-by-iteration prediction of the AMP performance. We observe  that  $\MSE_{\rm AMP}(t)$ (and $\MSE_{\rm SE}(t)$) decreases linearly versus $t$. 
\begin{figure}[!t]
\centering
\includegraphics*[viewport = -20 180 610 610, width =
3.7in]{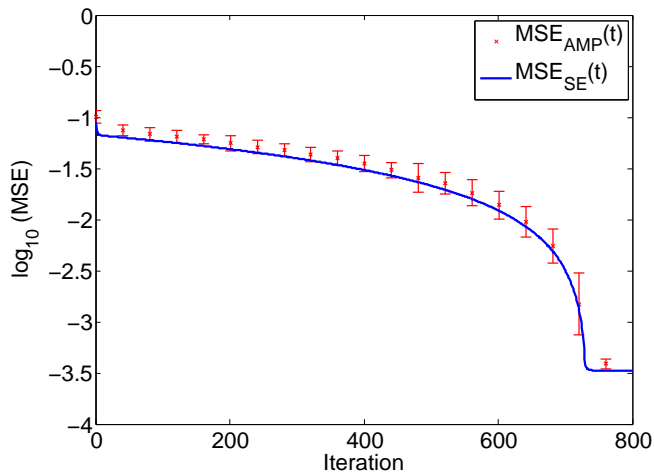}
\caption{Comparison of $\MSE_{{\rm AMP}}$ and $\MSE_{{\rm SE}}$ across iteration.}\label{fig:SE-AMP}

\end{figure} 

 \subsection{Phase diagram}
Consider a noiseless setting and let $\mathcal{A}$ be a sensing matrix--reconstruction algorithm scheme. 
The curve $\ve \mapsto \delta_{\mathcal{A}}(\ve)$ describes the
sparsity-undersampling tradeoff of $\cA$ if the following happens in the large-system limit $n,m \to
\infty$, with $m/n =\delta$.  The scheme $\cA$ does (with high
probability) correctly recover the original signal provided $\delta >
\delta_{\mathcal{A}}(\ve)$, while for $\delta <
\delta_{\mathcal{A}}(\ve)$ the algorithm fails with high probability. 

The goal of this section is to numerically compute the sparsity-undersampling tradeoff
curve for the proposed scheme (spatially coupled sensing matrices and Bayes optimal AMP ).
We consider a set of sparsity parameters  $\ve \in \{0.1,0.2,0.3,0.4,0.5\}$, and for
each value of $\ve$, evaluate the empirical phase transition
through a logit fit (we omit details, but follow the methodology
described in \cite{DMM09}). As shown in Fig~\ref{fig:phi-diagram}, the numerical results are consistent with
the claim that this scheme achieves the information theoretic lower bound $\delta > \uRenyi(p_X) = \ve$.
(We indeed expect the gap to decrease further by taking larger values of $L$).

\begin{figure}[!t]
\centering
\includegraphics*[viewport = -10 180 610 610, width =
3.7in]{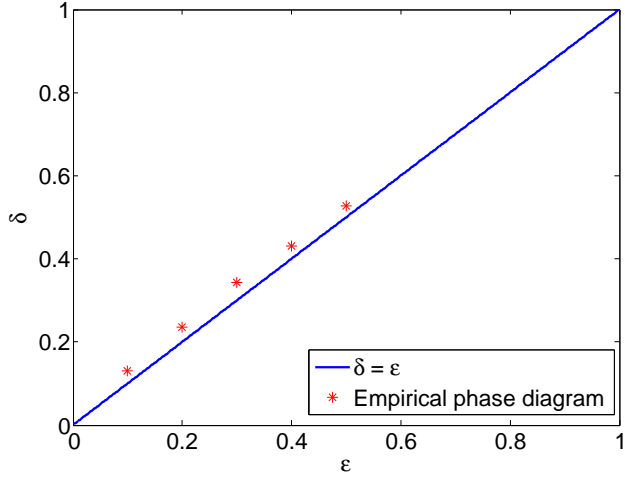}
\caption{Phase diagram for the spatially coupled sensing matrices and Bayes optimal AMP.}\label{fig:phi-diagram}
\end{figure} 

In~\cite{JavanmardMon12}, we numerically show that the spatial coupling phenomenon is significantly more robust and general than suggested by constructions in the present paper. Namely, we consider the problem of sampling a signal with sparse support in frequency domain and propose a sampling scheme that acquires a random subset of Gabor coefficients of the signal. This scheme offers one venue (out of many) for implementing the idea of spatial coupling. Note that the corresponding sensing matrix, in this context, does not have gaussian entries. As shown numerically for the mixture model, the combination of this scheme and the Bayes optimal AMP achieves the fundamental lower bound $\delta > \uRenyi(p_X)$.

%
%
\section{State evolution: an heuristic derivation}
\label{sec:SE_Heuristics}

This section presents an heuristic derivation of the state evolution
equations  (\ref{eq:ExplicitSE}). Our objective is to provide some
basic intuition: a proof in a more general setting will appear in a
separate publication \cite{JM-StateEvolution}. An heuristic derivation similar to
the present one, for the special cases of sensing matrices with
i.i.d. entries was presented in \cite{BM-MPCS-2011}.

Consider the recursion~\eqref{eq:AMP1}-\eqref{eq:AMP2}, and introduce the following modifications: $(i)$ At each iteration, replace the random matrix $A$ with a new independent copy $A^t$; $(ii)$ Replace the observation vector $y$ with $y^t = A^tx + w$; $(iii)$ Eliminate the last term in the update equation for $r^t$. Then, we have the following update rules:
\begin{eqnarray}
x^{t+1} & = & \eta_t(x^t+(Q^t\odot A^t)^*r^t)\, ,\label{eq:AMP_mod1}\\
r^t & = & y^t-A^t x^t\, ,\label{eq:AMP_mod2}
\end{eqnarray}
where $A^0, A^1, A^2, \cdots$ are i.i.d. random matrices distributed according to the ensemble $\mathcal{M}(W,M,N)$, i.e.,
\begin{eqnarray}
A_{ij}^t \sim \normal\Big(0,\frac{1}{M} W_{\gr(i),\gr(j)}\Big).
\end{eqnarray} 
Rewriting the recursion by eliminating $r^t$, we obtain:
\begin{eqnarray}\label{eq:AMP_mod3}
\begin{split}
x^{t+1} & =  \eta_t((Q^t\odot A^t)^*y^t + (I -(Q^t\odot A^t)^* A^t) x^t)\\
 & = \eta_t(x +  (Q^t\odot A^t)^* w + B^t (x^t - x))\, ,
\end{split}
\end{eqnarray}
where $B^t = I -(Q^t\odot A^t)^* A^t \in \reals^{n \times n}$. Note
that the recursion~\eqref{eq:AMP_mod3} does not correspond to the AMP
update rules defined per Eqs.~\eqref{eq:AMP1} and~\eqref{eq:AMP2}. 
In particular, it does not correspond to any practical algorithm since the sensing matrix $A$ is a fixed input to a reconstruction algorithm and is not resampled at each iteration. 
However, it is much easier to analyze, 
since $A^t$ is independent of $x^t$ and therefore the distribution of
$(Q^t\odot A^t)^*r^t$ can be easily characterized.
Also, it is useful for presenting the intuition behind the AMP
algorithm and to emphasize the role of the
term $\ons^t \odot r^{t-1}$ in the update rule for $r^t$. As it emerges
from the proof of \cite{BM-MPCS-2011}, this term does asymptotically
cancel dependencies across iterations.

By virtue of the central limit theorem, each entry of $B^t$ is approximately normal. More specifically, $B_{ij}^t$ is approximately normal with mean zero and variance $(1/M) \sum_{r\in \Rows} W_{r,\gr(i)} W_{r,\gr(j)} (Q^t_{r,\gr(i)})^2$, for $i,j \in [n]$. Define $\hat{\tau}_t(s) = \lim_{N \to \infty} \|x^t_{C(s)} - x_{C(s)}\|^2 / N$, for $s \in \Cols$. It is easy to show that distinct entries in $B^t$ are approximately independent. Also, $B^t$ is independent of $\{B^s\}_{1\le s\le t-1}$, and in particular, of $x^t - x$. Hence, $B^t(x^t - x)$ converges to a vector, say $v$, with i.i.d. normal entries, and for $i \in [n]$,
\begin{eqnarray}
\E\{v_i\} = 0, \quad \E\{v_i^2\} = \frac{N}{M} \sum_{u \in \Cols} \sum_{r \in \Rows} W_{r,\gr(i)} W_{r,u} (Q^t_{r,\gr(i)})^2\, \hat{\tau}_t(u).
\end{eqnarray}

Conditional on $w$, $(Q^t\odot A^t)^*w$ is a vector with i.i.d. zero-mean normal
entries . Also, the variance of its $i^{th}$ entry, for $i\in[n]$, is 
\begin{eqnarray}
\frac{1}{M} \sum_{r \in \Rows} W_{r,\gr(i)} (Q^t_{r,\gr(i)})^2 \|w_{R(r)}\|^2 ,
\end{eqnarray}
which converges to $\sum_{r \in \Rows} W_{r,\gr(i)} (Q^t_{r,\gr(i)})^2 \sigma^2 $, by the law of large numbers. With
slightly more work, 
it can be shown that these entries are approximately independent of the ones of $B^t(x^t - x)$. 

Summarizing, the $i^{th}$ entry of the vector in the argument of $\eta_t$ in Eq.~\eqref{eq:AMP_mod3} converges to $X + \tau_t(\gr(i))^{1/2} Z$ with $Z\sim \normal(0,1)$ independent of $X$, and
\begin{eqnarray}\label{eq:AMP_mod4}
\begin{split}
\tau_t(s) &= \sum_{r \in \Rows} W_{r,s} (Q^t_{r,s})^2 \big\{\sigma^2 +  \frac{1}{\delta} \sum_{u \in \Cols} W_{r,u} \, \hat{\tau}_t(u) \big\},\\
%
\end{split}
\end{eqnarray} 
for $s \in \Cols$. In addition, using Eq.~\eqref{eq:AMP_mod3} and invoking Eqs.~\eqref{eq:eta_def1},~\eqref{eq:eta_def2}, each entry of $x^{t+1}_{C(s)} - x_{C(s)}$ converges to $\eta_{t,s}(X + \tau_t(s)^{1/2}Z) - X$, for $s\in \Cols$. Therefore,
\begin{eqnarray}\label{eq:AMP_mod5}
\begin{split}
\hat{\tau}_{t+1}(s) &= \lim_{N \to \infty} \frac{1}{N} \|x^{t+1}_{C(s)} - x_{C(s)}\|^2\\
&= \E\{[\eta_{t,s}(X + \tau_t(s)^{1/2} Z) - X]^2\} = \mmse(\tau_t(s)^{-1}).
\end{split}
\end{eqnarray}

Using Eqs.~\eqref{eq:AMP_mod4} and~\eqref{eq:AMP_mod5}, we obtain:
\begin{eqnarray}
\tau_{t+1}(s) = \sum_{r \in \Rows} W_{r,s} (Q^{t+1}_{r,s})^2 \big\{\sigma^2 +  \frac{1}{\delta} \sum_{u \in \Cols} W_{r,u} \, \mmse(\tau_t(u)^{-1}) \big\}.
\end{eqnarray}

Applying the change of variable $\tau_t(u)^{-1} = \sum_{b \in \Rows} W_{b,u} \phi_b(t)^{-1}$, and substituting for $Q^{t+1}_{r,s}$ from Eq.~\eqref{eq:Q_def}, we obtain the state evolution recursion, Eq.~\eqref{eq:ExplicitSE}.

In conclusion, we showed that the state evolution recursion would hold
if the matrix $A$ was resampled independently from the ensemble
$\mathcal{M}(W,M,N)$, at each iteration. However, in our proposed AMP algorithm,
the matrix $A$ is constant across iterations, and the above argument is
not valid since $x^t$ and $A$ are dependent. 
The dependency between $A$ and $x^t$ cannot
be neglected. Indeed, state evolution does not apply to the following
naive iteration in which we dropped the memory term $\ons^t \odot r^{t-1}$:
\begin{eqnarray}
x^{t+1} & = & \eta_t(x^t+(Q^t\odot A)^*r^t)\, ,\label{eq:AMP_naive1}\\
r^t & = & y^t-A x^t\, .\label{eq:AMP_naive2}
\end{eqnarray}
Indeed,
the term $\ons^t \odot r^{t-1}$ leads to an asymptotic cancellation of
the dependencies between $A$ and $x^t$  as proved in \cite{BM-MPCS-2011,JM-StateEvolution}. 
     
%
\section{Analysis of  state evolution: Proof of Lemma \ref{lem:phi_convergence}}
\label{sec:AnalysisLemma}

This section is devoted to the analysis of the state evolution
recursion for spatially coupled matrices $A$, hence proving Lemma
\ref{lem:phi_convergence}.

In order to prove Lemma \ref{lem:phi_convergence}, we
will construct a free energy functional $\Energy_{\Shape}(\phi)$ such that the fixed points of the state
evolution are the stationary points of $\Energy_{\Shape}$. We then assume
by contradiction that the claim of the lemma does not hold, i.e.,
$\phi(t)$ converges to a fixed point $\phi(\infty)$ with $\phi_a(\infty)\gg\sigma^2$ for a
significant fraction of the indices $a$. We then obtain a
contradiction by describing an infinitesimal deformation of this fixed
point (roughly speaking, a shift to the right) that decreases its free energy.

\subsection{Outline}

A more precise outline of the proof is given below:
\begin{enumerate}
\item[($i$)] We establish some useful properties of the state evolution sequence $\{\phi(t), \psi(t)\}_{t \ge 0}$. This includes a monotonicity property as well as a lower and an upper bound for the state vectors.

\item[($ii$)] We define a modified state evolution sequence, denoted by $\{\phi^{\rm mod}(t), \psi^{\rm mod} (t)\}_{t \ge 0}$. This sequence dominates the original state vectors (see Lemma~\ref{lem:modified_dominance}) and hence it suffices to focus on the modified state evolution to get the desired result. As we will see the modified state evolution is more amenable to analysis.

\item[($iii$)] We next introduce continuum state evolution which serves as the continuous analog of the modified state evolution. (The continuum states are functions rather than vectors). The bounds on the continuum state evolution sequence lead to bounds on the modified state vectors.

\item[($iv$)] Analysis of the continuum state evolution incorporates
  the definition of a free energy functional defined on the space of
  non-negative measurable functions with bounded support. The energy
  is constructed in a way to ensure that the fixed points of the
  continuum state evolution are the stationary points of the free
  energy. Then, we show that if the undersampling rate is greater than
  the information dimension, the solution of the continuum state
  evolution can be made as small as $O(\sigma^2)$. If this were not
  the case, the (large) fixed point  could be perturbed slightly in
  such a way that  the free energy decreases to the  first order. However, since the fixed point is a stationary point of the free energy, this leads to a contradiction.     

\end{enumerate}
\subsection{Properties of the state evolution sequence}
Throughout this section $p_X$ is a given probability distribution over
the real line, and $X\sim p_X$. Also, we will take $\sigma > 0$. The result for the noiseless model (Corollary~\ref{coro:Noiseless}) follows by letting $\sigma \downarrow 0$. Recall the inequality
\begin{eqnarray}\label{eq:mmse_ineq}
\mmse(s)\le \min({\rm Var}(X),\frac{1}{s})\, .
\end{eqnarray}

\begin{definition}
For two vectors $\phi, \tilde{\phi} \in \reals^{K}$, we write $\phi
\succeq \tilde{\phi}$ if all $\phi_r \ge \tilde{\phi}_r$ for 
$r\in\{1,\dots,K\}$.
\end{definition}

\begin{propo}
\label{propo:monotonicity_TW}
For any $W \in \reals_+^{\Rows \times \Cols}$, the maps $\SEmapa_{W}:\reals_+^{\Rows}\to\reals_+^{\Cols}$ and 
$\SEmapb_{W}:\reals_+^{\Cols}\to\reals_+^{\Rows}$, as defined in Definition~\ref{def:SEmapa}, are monotone;
i.e., if $\phi \succeq \tilde{\phi}$ then $\SEmapa_{W}(\phi) \succeq
\SEmapa_W(\tilde{\phi})$, and if $\psi \succeq \tilde{\psi}$ then $\SEmapb_{W}(\psi) \succeq
\SEmapb_W(\tilde{\psi})$. Consequently, $\SEmap_W$ is also monotone.
\end{propo}
\begin{proof}
It follows immediately from the fact that $s\mapsto \mmse(s)$ is a
monotone decreasing function and the positivity of the matrix $W$.
\end{proof}

\begin{propo}\label{propo:monotone_map_discrete}
The state evolution sequence $\{\phi(t),\psi(t)\}_{t\ge 0}$ with
initial condition $\psi_i(0)=\infty$, for $i\in \Cols$, is monotone decreasing, in the
sense that $\phi(0)\succeq\phi(1)\succeq \phi(2)\succeq\dots$ and
$\psi(0)\succeq\psi(1)\succeq \psi(2)\succeq\dots$.
\end{propo}
\begin{proof}
Since $\psi_i(0)=\infty$ for all $i$, we have
$\psi(0)\succeq\psi(1)$. The thesis follows from the monotonicity of
the state evolution map.
\end{proof}

\begin{propo}\label{propo:monotone_noise}
The state evolution sequence $\{\phi(t),\psi(t)\}_{t\ge 0}$ 
is monotone increasing in $\sigma^2$. Namely, let $0\le \sigma_1\le
\sigma_2$ and $\{\phi^{(1)}(t),\psi^{(1)}(t)\}_{t\ge 0}$, $\{\phi^{(2)}(t),\psi^{(2)}(t)\}_{t\ge 0}$ 
be the state evolution sequences corresponding to setting, respectively,
$\sigma^2=\sigma_1^2$ and $\sigma^2=\sigma_2^2$ in
Eq.~(\ref{eq:ExplicitSE}), with identical initial conditions.
Then $\phi^{(1)}(t)\preceq\phi^{(2)}(t)$, 
$\psi^{(1)}(t)\preceq\psi^{(2)}(t)$ for all $t$. 
\end{propo}
\begin{proof}
Follows immediately from Proposition \ref{propo:monotonicity_TW} and
the monotonicity of the one-step mapping (\ref{eq:ExplicitSE}).
\end{proof}

\begin{lemma}\label{lemma:UB_SE}
Assume $\delta L_0> 3$. Then
there exists $t_0$ (depending only on $p_X$), such that, for all $t\ge t_0$ and all
$i\in\{-2\rho^{-1},\dots,-1\}$, $a\in\Rows_i$, we have
\begin{eqnarray}
\psi_i(t)&\le&
\mmse\Big(\frac{L_0}{2\sigma^2}\Big)\le\frac{2\sigma^2}{L_0}\, ,\\
\phi_a(t)&\le&\sigma^2+\frac{1}{\delta}\mmse\Big(\frac{L_0}{2\sigma^2}\Big)\le
\Big(1+\frac{2}{\delta L_0}\Big)\sigma^2\, .
\end{eqnarray}
\end{lemma}
\begin{proof}
Take $i \in \{-2\rho^{-1},\cdots,-1\}$. For $a \in\Rows_i$, we have $\phi_a(t) =
\sigma^2+(1/\delta)\psi_i(t)$. Further
from $ \mmse(s)\le 1/s$, we deduce that 
\begin{eqnarray}
\begin{split}
\psi_i(t+1) &= \mmse\Big( \sum_{b \in \Rows} W_{b,i} \phi_b(t)^{-1}\Big)
 \le \Big( \sum_{b \in \Rows} W_{b,i} \phi_b(t)^{-1}\Big)^{-1}\\
 & \le \Big(\sum_{a\in\Rows_i} W_{a,i}\phi_a(t)^{-1} \Big) ^{-1} = \Big( \L0 \phi_a(t)^{-1}\Big)^{-1} = \frac{\phi_a(t)}{\L0}.
 \end{split}
\end{eqnarray}
Here we used the facts that $W_{a,i} =1$, for $a\in \Rows_i$ and $|\Rows_i| = \L0$. Substituting in the earlier relation, we get $\psi_i(t+1)\le (1/\L0) (\sigma^2 + (1 /\delta) \psi_i(t))$. Recalling that $\delta \L0 >3$, we have $\psi_i(t) \le 2\sigma^2 / \L0$, for all $t$ sufficiently large. Now, using this in the equation for $\phi_a(t)$, $a \in \Rows_i$, we obtain
\begin{eqnarray} \label{eqn:phi_b1}
\phi_a(t) = \sigma^2 + \frac{1}{\delta} \psi_i(t) \le \Big( 1+ \frac{2}{\delta \L0}\Big) \sigma^2.
\end{eqnarray}
We prove the other claims
by repeatedly substituting in the previous bounds. In particular,
\begin{eqnarray}
\begin{split} \label{eqn:psi_b_1}
\psi_i(t) &= \mmse\Big( \sum_{b \in \Rows} W_{b,i} \phi_b(t-1)^{-1}\Big)
\le \mmse\Big( \sum_{a \in \Rows_i} W_{a,i} \phi_a(t)^{-1}\Big)\\
&= \mmse(\L0 \phi_a(t)^{-1}) \le \mmse \Big( \frac{\L0}{ (1+\frac{2}{\delta \L0}) \sigma^2} \Big) \le \mmse\Big( \frac{\L0}{2 \sigma^2}\Big),
\end{split}
\end{eqnarray}
where we used Eq.~\eqref{eqn:phi_b1} in the penultimate inequality. Finally,
\begin{eqnarray}
\phi_a(t) \le \sigma^2 + \frac{1}{\delta} \psi_i(t) \le \sigma^2 + \frac{1}{\delta} \mmse \Big( \frac{\L0}{2 \sigma^2}\Big),
\end{eqnarray}
where the inequality follows from Eq.~\eqref{eqn:psi_b_1}.
\end{proof}

Next we prove a lower bound on the state evolution sequence.
Here and below
$\Cols_0\equiv\Cols\setminus\{-2\rho^{-1},\dots,-1\}\cong\{0,\dots ,L-1\}$. Also, recall that $\Rows_0 \equiv \{-\rho^{-1},\dots,0,\dots ,L-1 + \rho^{-1}\}$. (See Fig.~\ref{fig:Wmatrix}).
\begin{lemma}\label{lemma:LB_SE}
For any $t\ge 0$, and any $i\in\Cols_0$, $\psi_i(t)\ge
\mmse(2 \sigma^{-2})$. Further, for any $a\in
\Rows_0$ and any 
$t\ge 0$ we have $\phi_a(t)\ge
\sigma^2+(2\delta)^{-1}\mmse(2\sigma^2)$.
\end{lemma}
\begin{proof}
Since $\phi_a(t)\ge\sigma^2$ by definition, we have, for $i\ge 0$, $\psi_i(t)\ge
\mmse(\sigma^{-2}\sum_bW_{bi})\ge \mmse(2\sigma^{-2})$, where we used
the fact that the restriction of $W$ to columns in $\Cols_0$ is roughly column-stochastic. Plugging this into the expression for $\phi_a$, we get 
\begin{eqnarray}
\phi_a(t)\ge\sigma^2+\frac{1}{\delta}\sum_{i \in \Cols} W_{a,i}\;
\mmse(2\sigma^{-2})\ge
\sigma^2+\frac{1}{2\delta}\mmse(2\sigma^{-2}) \, .
\end{eqnarray}
\end{proof}

Notice that for $L_{0,*} \ge 4$ and for all $L_0>L_{0,*}$, the upper bound for $\psi_i(t)$, $i\in \{-2\rho^{-1},\cdots, -1\}$, given in Lemma \ref{lemma:UB_SE} is below the lower
bound for $\psi_i(t)$, with $i \in \Cols_0$, given in Lemma \ref{lemma:LB_SE}; i.e., for all $\sigma$,
\begin{eqnarray}\label{eq:ConditionL0}
\mmse\Big(\frac{L_0}{2\sigma^2}\Big)\le \mmse\Big(\frac{2}{\sigma^2}\Big)\, .
\end{eqnarray}
%

\subsection{Modified state evolution}

First of all, by Proposition \ref{propo:monotone_noise} we can assume,
without loss of generality $\sigma>0$.

Motivated by the monotonicity properties of the state evolution sequence mentioned
in Lemmas~\ref{lemma:UB_SE} and~\ref{lemma:LB_SE}, 
we introduce a new state evolution recursion that dominates the
original one and yet is more amenable to analysis.  Namely, we define the modified state evolution maps
$\Mapa_W:\reals_+^{\Rows_0}\to\reals_+^{\Cols_0}$,
$\Mapb_W:\reals_+^{\Cols_0}\to\reals_+^{\Rows_0}$. 
For $\phi=(\phi_a)_{a\in\Rows_0}\in\reals_+^{\Rows_0}$,
$\psi=(\psi_i)_{i \in\Cols_0}\in\reals_+^{\Cols_0}$, and for all
$i\in\Cols_0$, $a\in\Rows_0$, let:
\begin{eqnarray}
\Mapa_{W}(\phi)_{i} &= &\mmse\Big(\sum_{b\in\Rows_0}W_{b-i}\phi_b^{-1}\Big)\, ,\\
\Mapb_{W}(\psi)_{a} &=& \sigma^2+\frac{1}{\delta}\sum_{i\in\integers}
W_{a-i}\, \psi_i\, .
\end{eqnarray}
where, in the last equation we set  by convention,
$\psi_i(t)=\mmse(L_0/(2\sigma^2))$ for 
$i\le -1$, and $\psi_i=\infty$ for $i\ge L$, and recall the shorthand
$W_{a-i}\equiv \rho\, \Shape\big(\rho\,(a-i)\big)$ 
introduced in Section \ref{sec:Choices}. We also let $\Map_W =
\Mapa_W\circ\Mapb_W$. 
\begin{definition}\label{def:Modified_SE}
The \emph{modified state evolution sequence} is the sequence
$\{\phi(t),\psi(t)\}_{t\ge 0}$ with $\phi(t) = \Mapb_W(\psi(t))$ and   
$\psi(t+1)=\Mapa_W(\phi(t))$ for all $t\ge 0$, and $\psi_i(0) =\infty$
for all $i\in\Cols_0$. We also adopt the convention that, for $i \ge L$, $\psi_i(t) = + \infty$ and for $i \le -1$, $\psi_i(t) = \mmse(\L0/ (2 \sigma^2))$, for all $t$. 
\end{definition}
Lemma \ref{lemma:UB_SE} then implies the following.
\begin{lemma}\label{lem:modified_dominance}
Let $\{\phi(t),\psi(t)\}_{t\ge 0}$ denote the state evolution sequence as
per Definition \ref{def:SESequence}, and  $\{\phi^{\rm
  mod}(t),\psi^{\rm mod}(t)\}_{t\ge 0}$  denote the modified state
evolution sequence as per Definition \ref{def:Modified_SE}. Then,
there exists $t_0$ (depending only on $p_X$), such that, for all $t\ge t_0$,
$\phi(t)\preceq \phi^{\rm mod}(t-t_0)$ and $\psi(t)\preceq \psi^{\rm mod}(t-t_0)$.  
\end{lemma}
\begin{proof}
Choose $t_0 = t(\L0,\delta)$ as given by Lemma~\ref{lemma:UB_SE}.  We prove the claims by induction on $t$. For the induction basis ($t = t_0$), we have from Lemma~\ref{lemma:UB_SE}, $\psi_i(t_0) \le  \mmse(\L0 /(2\sigma^2)) = \psi_i^{\rm mod}(0)$, for $i \le -1$. Also, we have $\psi^{\rm mod}_i(0) = \infty \ge \psi_i(t_0)$, for $i \ge 0$. Further, 
\begin{eqnarray}
\phi^{\rm mod}_a(0) = \Mapb_W(\psi^{\rm mod}(0))_a \ge \SEmapb_W(\psi^{\rm mod}(0))_a \ge  \SEmapb_W(\psi(t_0))_a = \phi_a(t_0),
\end{eqnarray}
for $a \in \Rows_0$. Here, the last inequality follows from monotonicity of $\SEmapb_W$ (Proposition~\ref{propo:monotonicity_TW}). Now, assume that the claim holds for $t$; we prove it for $t+1$. For $i \in \Cols_0$, we have
\begin{eqnarray}
\begin{split} \label{eqn:psi_induction}
\psi_i^{\rm mod}(t+1-t_0) &= \Mapa_W(\phi^{\rm mod}(t - t_0))_i = \SEmapa_W(\phi^{\rm mod}(t-t_0))_i\\
&\ge \SEmapa_W(\phi(t))_i  = \psi_i(t+1),
\end{split}
\end{eqnarray}
where the inequality follows from monotonicity of $\SEmapa_W$ (Proposition~\ref{propo:monotonicity_TW}) and the induction hypothesis. In addition, for $a \in \Rows_0$,
\begin{eqnarray}
\begin{split}
\phi_a^{\rm mod}(t+1-t_0) &= \Mapb_W(\psi^{\rm mod}(t+1-t_0))_a \ge \SEmapb_W(\psi^{\rm mod}(t+1-t_0))_a\\
& \ge \SEmapb_W(\psi(t+1))_a = \phi_a(t+1).
\end{split}
\end{eqnarray}
Here, the last inequality follows from monotonicity of $\SEmapb_W$ and Eq.~\eqref{eqn:psi_induction}.
\end{proof}
By Lemma~\ref{lem:modified_dominance}, we can now focus on the modified state evolution sequence in order to prove Lemma~\ref{lem:phi_convergence}. Notice that the mapping $\Map_W$ has a
particularly simple description in terms of a shift-invariant state
evolution mapping. Explicitly, define
$\SEmapa_{W,\infty}:\reals^{\integers}\to\reals^{\integers}$,
$\SEmapb_{W,\infty}:\reals^{\integers}\to\reals^{\integers}$, by
letting, for $\phi,\psi\in\reals^{\integers}$ and all $i,a\in \integers$:
\begin{eqnarray}
\SEmapa_{W,\infty}(\phi)_{i} &= &\mmse\Big(\sum_{b\in\integers}W_{b-i}\phi_b^{-1}\Big)\, ,\\
\SEmapb_{W,\infty}(\psi)_{a} &=& \sigma^2+\frac{1}{\delta}\sum_{i\in\integers}
W_{a-i}\, \psi_i\, .
\end{eqnarray}
Further, define the embedding
$\Embed:\reals^{\Cols_0}\to\reals^{\integers}$ by letting
\begin{eqnarray}
(\Embed\psi)_i = 
\begin{cases}
\mmse(L_0/(2\sigma^2)) & \mbox{if $i<0$,}\\
\psi_i &\mbox{if $0\le i\le L-1$,}\\
+\infty &\mbox{if $i\ge L$,}
\end{cases}
\end{eqnarray}
And the restriction mapping
$\Restr_{a,b}:\reals^{\integers}\to\reals^{b-a+1}$ by
$\Restr_{a,b}\psi = (\psi_a,\dots,\psi_{b})$. 
\begin{lemma}\label{lemma:Shift}
With the above definitions, $\Map_W =
\Restr_{0,L-1}\circ\SEmap_{W,\infty}\circ\Embed$.
\end{lemma}
\begin{proof}
Clearly, for any $\psi = (\psi_i)_{i \in \Cols_0}$, we have $\SEmapb_{W} \circ\Embed (\psi)_a = \Mapb_W \circ\Embed (\psi)_a$ for $ a \in \Rows_0$, since the definition of the embedding $\Embed$ is consistent with the convention adopted in defining the modified state evolution. Moreover, for $i \in \Cols_0 \cong \{0,\dots,L-1\}$, we have
\begin{eqnarray}
\begin{split}
\SEmapa_{W,\infty}(\phi)_i &= \mmse\Big(\sum_{b \in \integers} W_{b - i} \phi_b^{-1} \Big) = \mmse \Big( \sum_{-\rho^{-1} \le b \le L -1 + \rho^{-1} } W_{b - i} \phi_b^{-1}\Big)\\
&= \mmse\Big( \sum_{b \in \Rows_0 } W_{b - i} \phi_b^{-1}\Big) = \Mapa_W(\phi)_i.
\end{split}
\end{eqnarray}
Hence, $\SEmapa_{W,\infty} \circ \SEmapb_{W,\infty} \circ \Embed(\psi)_i =  \Mapa_W \circ \Mapb _W \circ \Embed (\psi)_i$, for $i \in \Cols_0$. Therefore, $\Restr_{0,L-1} \circ \SEmap_{W,\infty} \circ\Embed (\psi) =  \Map_W \circ\Embed (\psi)$, for any $\psi \in \reals_+^{\Cols_0}$, which completes the proof.
\end{proof}

We will say that a vector $\psi\in\reals^K$ is \emph{nondecreasing} if, for
every $1\le i<j\le K$, $\psi_i\le \psi_j$.
\begin{lemma}\label{lemma:NonDecreasing}
If $\psi\in\reals^{\Cols_0}$ is nondecreasing, with $\psi_i\ge
\mmse(L_0/(2\sigma^2))$ for all $i$,
then $\Map_W(\psi)$ is 
nondecreasing as well. In particular, if $\{\phi(t),\psi(t)\}_{t\ge 0}$ is
the modified
state evolution sequence, then $\phi(t)$ and $\psi(t)$ are nondecreasing for all
$t$.
\end{lemma}
\begin{proof}
By Lemma \ref{lemma:Shift}, we know that $\Map_W=
\Restr_{0,L-1}\circ\SEmap_{W,\infty}\circ\Embed$. We first notice that, by the
assumption $\psi_i\ge \mmse(L_0/(2\sigma^2))$, we have that 
$\Embed(\psi)$ is nondecreasing. 

Next, if $\psi\in\reals^{\integers}$ is nondecreasing,
$\SEmap_{W,\infty}(\psi)$ is nondecreasing as well. In fact, the mappings
$\SEmapa_{W,\infty}$ and $\SEmapb_{W,\infty}$ both preserve the
nondecreasing property, since both are shift invariant, and $\mmse(\,\cdot \,)$ is a decreasing function.
Finally, the restriction of a nondecreasing vector is obviously
nondecreasing. 

This proves that $\Map_W$ preserves the nondecreasing
property. To conclude that $\psi(t)$ is nondecreasing for all $t$,
notice that the condition  $\psi_i(t)\ge
\mmse(L_0/(2\sigma^2))$ is satisfied at all $t$ by Lemma
\ref{lemma:LB_SE} and condition (\ref{eq:ConditionL0}). The claim for $\psi(t)$ follows by induction.

Now, since $\Mapb_{W}$ preserves the nondecreasing property, we have $\phi(t) = \Mapb_{W}(\psi(t))$ is nondecreasing for all $t$, as well.
\end{proof}

\subsection{Continuum state evolution}

We start by defining the continuum state evolution mappings.
For  $\Omega \subseteq \reals$, let $\mathscr{M}(\Omega)$ be the space of non-negative measurable functions on $\Omega$ (up to
measure-zero redefinitions). Define $\CMapa_{\Shape}: \mathscr{M}([-1,\ell+1]) \to \mathscr{M}([0,\ell])$ and $\CMapb_{\Shape}: \mathscr{\mathscr{M}}([0,\ell]) \to \mathscr{M}([-1,\ell+1])$ as follows. For $\phi \in \mathscr{M}([-1,\ell+1]), \psi \in \mathscr{M}([0,\ell])$, and for all $x \in [0,\ell], y \in [-1,\ell+1]$, we let
\begin{align}
\CMapa_{\Shape}(\phi)(x) &= \mmse\Big(\int_{-1}^{\ell+1} \Shape(x-z) \phi(z)^{-1} \de z \Big),\label{eq:CMAP1}\\
\CMapb_{\Shape}(\psi)(y) &= \sigma^2 + \frac{1}{\delta} \int_{\reals} \Shape(y-x) \psi(x) \de x\,,\label{eq:CMAP2}
\end{align}
where we adopt the convention that $\psi(x) = \mmse(\L0/(2\sigma^2))$ for $x < 0$, and $\psi(x) = \infty$ for $x > \ell$. 
\begin{definition}\label{def:ContinuumSE}
The continuum state evolution sequence  is the sequence
$\{\phi(\,\cdot\, ;t),\psi(\,\cdot\,;t)\}_{t\ge 0}$, with $\phi(t) = \CMapb_{\Shape}(\psi(t))$ and   
$\psi(t+1)=\CMapa_{\Shape}(\phi(t))$ for all $t\ge 0$, and $\psi(x;0) =\infty$
for all $x\in[0,\ell]$.
\end{definition}
Recalling Eq.~\eqref{eq:mmse_ineq}, we have $\psi(x;t) = \CMapa_{\Shape}(\phi(t-1))(x) \le {\rm Var}(X)$, for $t\ge 1$. Also, $\phi(x;t) = \CMapb_{\Shape}(\psi(t))(x) \le \sigma^2 + (1/\delta) {\rm Var}(X)$, for $t \ge 1$. Define,
\begin{eqnarray}
\Phi_M = 1 + \frac{1}{\delta}\, {\rm Var}(X).
\end{eqnarray}
Assuming $\sigma < 1$, we have $\phi(x;t) < \Phi_M$, for all $t\ge 1$.

The point of introducing continuum state evolution is that by construction of the matrix $W$ and the continuity of $\Shape$, when $\rho$ is small, one can approximate summation by integration and study the evolution of the continuum states which are represented by functions rather than vectors. This observation is formally stated in lemma below.
\begin{lemma}\label{lemma:ContinuumLimit}
Let $\{\phi(\,\cdot\, ;t),\psi(\,\cdot\,;t)\}_{t\ge 0}$ be the
continuum state evolution sequence and $\{\phi(t),\psi(t)\}_{t\ge 0}$
be the modified discrete state evolution sequence, with parameters
$\rho$ and $L=\ell/\rho$. Then for any $t\ge 0$
\begin{align}
&\lim_{\rho\to 0 }\frac{1}{L}\sum_{i=0}^{L-1}
\big|\psi_i(t)-\psi(\rho i;t)\big| = 0\, ,\label{eq:FirstSEApprox}\\
&\lim_{\rho\to 0 }\frac{1}{L}\sum_{a=-\rho^{-1}}^{L-\rho^{-1} - 1}
\big|\phi_a(t)-\phi(\rho a;t)\big| = 0\,\label{eq:SecondSEApprox} .
\end{align}
\end{lemma}
Lemma ~\ref{lemma:ContinuumLimit} is proved in Appendix~\ref{app:ContinuumLimit}.
\begin{coro}\label{cor:monotone_map_continuous}
The continuum state evolution sequence $\{\phi(\,\cdot\, ;t),\psi(\,\cdot\,;t)\}_{t\ge 0}$, with initial condition $\psi(x) = \mmse(\L0/(2\sigma^2))$ for $x < 0$, and $\psi(x) = \infty$ for $x > \ell$, is monotone decreasing, in the sense that $\phi(x;0) \ge \phi(x;1) \ge \phi(x;2) \ge \cdots$ and $\psi(x;0) \ge \psi(x;1) \ge \psi(x;2) \ge \cdots$, for all $x \in [0,\ell]$.
\end{coro}
\begin{proof}
Follows immediately from Lemmas \ref{propo:monotone_map_discrete} and
\ref{lemma:ContinuumLimit}.
\end{proof}
\begin{coro}\label{cor:nondecreasing_continuous}
Let $\{\phi(\,\cdot\, ;t),\psi(\,\cdot\,;t)\}_{t\ge 2}$ be the
continuum state evolution sequence.
Then for any $t$, $x\mapsto \psi(x;t)$ and $x \mapsto \phi(x;t)$ are nondecreasing Lipschitz continuous functions.
\end{coro}
\begin{proof}
Nondecreasing property of functions $x\mapsto \psi(x;t)$, and $x \mapsto \phi(x;t)$ follows immediately from Lemmas \ref{lemma:NonDecreasing} and \ref{lemma:ContinuumLimit}. Further, since $\psi(x;t)$ is bounded for $t \ge 1$, and $\Shape(\,\cdot\,)$ is Lipschitz continuous, recalling Eq.~\eqref{eq:CMAP2}, the function $x\mapsto \phi(x;t)$ is Lipschitz continuous as well, for $t\ge 1$. Similarly, since $ \sigma^2 < \phi(x;t) < \Phi_M$, invoking Eq.~\eqref{eq:CMAP1}, the function $x\mapsto \psi(x;t)$ is Lipschitz continuous for $t\ge 2$. 
\end{proof}
%
%

\subsubsection{Free energy}

A key role in our analysis is played by the free energy functional. In order to define the free energy, we first provide some preliminaries. 
Define the mutual information
between  $X$ and a noisy observation of $X$ at signal-to-noise ratio $s$ by
\begin{eqnarray}
\Info(s) \equiv I(X;\sqrt{s}X+Z)\,,
\end{eqnarray}
with $Z\sim\normal(0,1)$ independent of $X\sim p_X$.
Recall the relation~\cite{Guo05mutualinformation} 
\begin{eqnarray}
\frac{\de\phantom{s}}{\de s}\,\Info(s) = \frac{1}{2}\,\mmse(s)\,.
\end{eqnarray}
Furthermore, the following identities relate the scaling law of mutual information under weak noise to R\'enyi information dimension~\cite{WuVerduMMSE}. 

\begin{propo}\label{propo:info_dimension} 
Assume $H (\lfloor X \rfloor) < \infty$. Then
\begin{eqnarray}\label{eqn: info_dimension}
\begin{split}
\underset{s \to \infty}{\liminf} \frac{\Info(s)}{\frac{1}{2} \log s} &= \lRenyi(p_X),\\
\underset{s \to \infty}{\limsup} \frac{\Info(s)}{\frac{1}{2} \log s} &= \uRenyi(p_X).
\end{split}
\end{eqnarray}
\end{propo}

Now we are ready to define the free energy functional.
\begin{definition}
Let $\Shape(\,\cdot\,)$ be a shape function, and $\sigma,\delta>0$ be given. The corresponding \emph{free energy} is the
functional $\Energy_\Shape: \mathscr{M}([-1,\ell+1])\to \overline{\reals}$ defined as
follows for $\phi \in \mathscr{M}([-1,\ell+1])$:
\begin{eqnarray}\label{eqn:free_spec}
\begin{split}
\Energy_{\Shape}(\phi) = \frac{\delta}{2}\int_{-1}^{\ell-1} 
\Big\{\frac{\varsigma^2(x)}{\phi(x)}+\log \phi(x)\Big\} \de x + 
\int_{0}^{\ell}\Info\Big(\int \Shape(x-z)\phi(z)^{-1} \de z\Big) \de x,
\end{split}
\end{eqnarray}
where 
\begin{eqnarray} \label{eqn:new_sigma}
\varsigma^2(x) = \sigma^2 + \frac{1}{\delta} \Big(\int_{y \le 0} \Shape(y-x) \de y\Big) \mmse \Big( \frac{\L0}{2\sigma^2}\Big). 
\end{eqnarray}
\end{definition}
The name `free energy' is motivated by the connection with statistical
physics, whereby $\Energy_{\Shape}(\phi)$ is the asymptotic
log-partition function for the Gibbs-Boltzmann measure corresponding
to the posterior distribution of $x$ given $y$. 
(This connection is however immaterial for our proof and we will
not explore it further, see for instance \cite{KrzakalaEtAl}.)

Notice that this is where the R\'enyi information comes into the picture. The mutual information appears in the expression of the free energy and the mutual information is related to the R\'enyi information via Proposition~\ref{propo:info_dimension}.

Viewing $\Energy_{\Shape}$ as a function defined on the Banach space
$L_2([-1,\ell])$, we will denote by $\nabla E_{\Shape}(\phi)$ its
Fr\'echet derivative at $\phi$. This will be identified, via standard duality,
with a function in $L_2([-1,\ell])$. It is not hard to show
that the Fr\'echet derivative exists on $\{\phi:\, \phi(x)\ge \sigma^2\}$ and is such that
\begin{eqnarray}\label{eqn:grad_energy}
\begin{split}
\nabla\Energy_{\Shape}(\phi)(y) = \frac{\delta}{2\phi^2(y)} \Big\{\phi(y) - \varsigma^2(y)
- \frac{1}{\delta} \int_{0}^{\ell} \Shape(x-y) \mmse \big(\int \Shape(x-z) \phi(z)^{-1} \de z \big) \de x \Big\}, 
\end{split}
\end{eqnarray}
for $-1 \le y \le \ell-1$. Note that the condition $\phi(x)\ge
\sigma^2$ is immediately satisfied by the state evolution sequence
since, by Eq.~(\ref{eq:CMAP2}), $\CMapb_{\Shape}(\psi)(y) \ge\sigma^2$
for all $y$ (because $\Shape(y-x)$, $\psi(x;t)\ge0$); see also Definition \ref{def:ContinuumSE}.

The specific choice of the free energy in Eq.~\eqref{eqn:free_spec} ensures that the fixed points of the continuum state evolution are the stationary points of the free energy.

\begin{coro}\label{coro:fix_gradiant}
If $\{\phi, \psi\}$ is the fixed point of the continuum state evolution, then $\nabla\Energy_W(\phi)(y)=0$, for $-1 \le y\le \ell - 1$. 
\end{coro}
\begin{proof}
We have $\phi = \CMapb_{\Shape}(\psi)$ and $\psi = \CMapa_{\Shape}(\phi)$, whereby for $-1\le y \le \ell-1$,
\begin{eqnarray}
\begin{split}
\phi(y) &= \sigma^2  + \frac{1}{\delta} \int \Shape(y-x) \psi(x) \de x\\
           &= \sigma^2 + \frac{1}{\delta} \Big(\int_{x \le 0} \Shape(y-x) \de x\Big) \mmse \left( \frac{\L0}{2\sigma^2}\right)+\\
           &\quad  \frac{1}{\delta} \int_{0}^{\ell} \Shape(y-x) \mmse \Big( \int_{-1}^{\ell+1} \Shape(x-z) \phi(z)^{-1} \de z\Big) \de x\\
           &  =  \varsigma^2(y) + \frac{1}{\delta} \int_{0}^{\ell} \Shape(y-x) \mmse \Big( \int_{-1}^{\ell+1} \Shape(x-z) \phi(z)^{-1} \de z\Big) \de x.
\end{split}
\end{eqnarray}
The result follows immediately from Eq.~\eqref{eqn:grad_energy}.
\end{proof}
\begin{definition}
Define the \emph{potential} function $V: \reals_+ \to \reals_+$ as follows.
\begin{eqnarray}\label{eqn:potential_def}
V(\phi) = \frac{\delta}{2} \Big( \frac{\sigma^2}{\phi} + \log \phi \Big) + \Info(\phi^{-1}).
\end{eqnarray}
\end{definition}

As we will see later, the analysis of the continuum state evolution involves a decomposition of the free energy functional into three terms and a careful treatment of each term separately. The definition of the potential function $V$ is motivated by that decomposition. 

Using Eq.~\eqref{eqn: info_dimension}, we have for $\phi \ll 1$,
\begin{eqnarray}
\begin{split}
V(\phi) &\lesssim \frac{\delta}{2}(\frac{\sigma^2}{\phi} + \log \phi) + \frac{1}{2} \uRenyi(p_X)  \log(\phi^{-1})\\
& = \frac{\delta \sigma^2}{2 \phi} + \frac{1}{2} [\delta - \uRenyi(p_X)] \log(\phi).
\end{split}
\end{eqnarray}
Define 
\begin{eqnarray}
\phi^* = \sigma^2+ \frac{1}{\delta} \mmse\Big(\frac{\L0}{2\sigma^2}\Big).
\end{eqnarray}
Notice that $\sigma^2 < \phi^* \le (1+2/(\delta \L0)) \sigma^2 < 2\sigma^2$, given that $\delta \L0 > 3$. The following proposition upper bounds $V(\phi^*)$ and its proof is deferred to Appendix~\ref{app:V_properties}.
\begin{propo}
\label{propo:V_properties}
There exists $\sigma_2> 0$, such that, for $\sigma \in (0, \sigma_2]$, we have
\begin{eqnarray}\label{eqn:V_prop1}
V(\phi^*) \le \frac{\delta}{2} + \frac{\delta - \uRenyi(p_X)}{4} \log(2\sigma^2).
\end{eqnarray}
\end{propo}

%
%
Now, we write the energy functional in terms of the potential function.
\begin{eqnarray} \label{eqn:energy-potential}
\Energy_{\Shape}(\phi) = \int_{-1}^{\ell-1} V(\phi(x)) \;\de x +\frac{\delta}{2} \int_{-1}^{\ell-1} \frac{\varsigma^2(x) - \sigma^2}{\phi(x)}\; \de x +  \tEnergy_{\Shape}(\phi),
\end{eqnarray}
with,
\begin{eqnarray}
\tEnergy_{\Shape}(\phi) = \int_{0}^{\ell} \big\{ \Info(\Shape \ast {\phi}(y)^{-1}) - \Info(\phi(y-1)^{-1}) \big\} \de y.
\end{eqnarray}
%

\subsubsection{Analysis of the continuum state evolution}
Now we are ready to study the fixed points of the continuum state evolution.
\begin{lemma}\label{lem:main_continuum}
Let $\delta > 0$, and $p_X$ be a probability measure on the real line with $\delta > \bar{d}(p_X)$. For any $\kappa > 0$, there exist $\ell_0$, $\sigma_0^2 = \sigma_0(\kappa, \delta, p_X)^2$, such that, for any $\ell > \ell_0$ and $\sigma \in (0,\sigma_0]$, and any fixed point of continuum state evolution,$\{\phi, \psi \}$, with $\psi$ and $\phi$ nondecreasing Lipschitz functions and $\psi(x) \ge \mmse (\L0/(2\sigma^2))$, the following holds.
\begin{eqnarray}
\int_{-1}^{\ell - 1} |\phi(x) - \phi^*| \;\de x \le \kappa \ell.
\end{eqnarray}
\end{lemma}
\begin{proof}
The claim is trivial for $\kappa \ge \Phi_M$, since $\phi(x) \le \Phi_M$. Fix $\kappa < \Phi_M$, and choose $\sigma_1$, such that $\phi^* < \kappa/2$, for $\sigma \in (0,\sigma_1]$. Since $\phi$ is a fixed point of continuum state evolution, we have $\nabla \Energy_{\Shape}(\phi) = 0$, on the interval $[-1,\ell-1]$ by Corollary~\ref{coro:fix_gradiant}. Now, assume that $\int_{-1}^{\ell - 1} |\phi(x) - \phi^*| > \kappa \ell$. We introduce an infinitesimal perturbation of $\phi$ that decreases the energy in the first order; this contradicts the fact $\nabla \Energy_{\Shape}(\phi) = 0$ on the interval $[-1,\ell-1]$.
\begin{claim}\label{claim:slope_phi}
For each fixed point of continuum state evolution that satisfies the hypothesis of Lemma~\ref{lem:main_continuum}, the following holds. For any $K > 0$, there exists $\ell_0$, such that, for $\ell > \ell_0$ there exist $x_1 < x_2 \in [0,\ell -1)$, with $x_2 - x_1  = K$ and $\kappa/2+ \phi^* < \phi(x)$, for $x\in [x_1,x_2]$.
\end{claim}
Claim~\ref{claim:slope_phi} is proved in Appendix~\ref{app:slope_phi}.

Fix $K >2$ and let $x_0 = (x_1+x_2)/2$. Thus, $x_0 \ge 1$. For $a \in (0,1]$, define
\begin{eqnarray}
\phi_a (x) = \begin{cases}
\phi(x), & \text{for }   x_2 \le x,\\
\phi(\frac{x_2 - x_0}{x_2 - x_0 - a}\;x - \frac{ax_2}{x_2 -x_0 - a}), & \text{for } x \in [x_0 +a, x_2),\\
\phi(x-a), &\text{for } x \in [-1+ a, x_0+a),\\
\phi^*, &\text{for } x \in [-1,-1+a).
\end{cases}
\end{eqnarray}
\begin{figure}[!t]
\centering
\includegraphics*[viewport = 80 120 700 500, width = 3.6in]{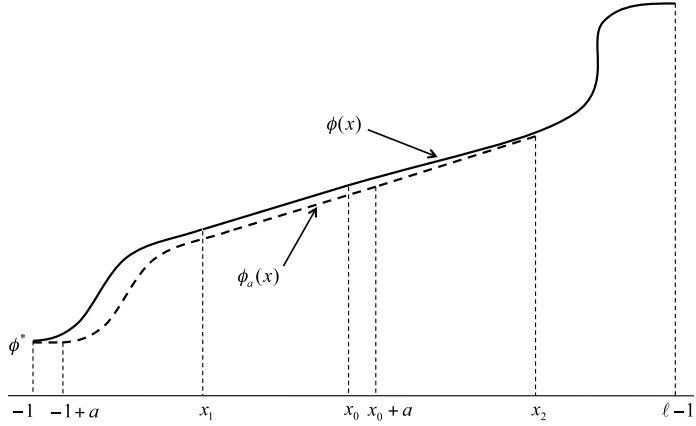}
\caption{\small An illustration of function $\phi(x)$ and its perturbation $\phi_a(x)$.}  \label{fig:Phi}
\end{figure}
See Fig.~\ref{fig:Phi} for an illustration. (Note that from Eq.~\eqref{eq:CMAP2}, $\phi(-1) = \phi^*$). In the following, we bound the difference of the free energies of functions $\phi$ and $\phi_a$.
\begin{propo}\label{propo:sigma_bound}
For each fixed point of continuum state evolution, satisfying the hypothesis of Lemma~\ref{lem:main_continuum}, there exists a constant $C(K)$, such that 
\begin{eqnarray*}
\int_{-1}^{\ell-1} \Big\{ \frac{\varsigma^2(x) - \sigma^2}{\phi_a(x)} - \frac{\varsigma^2(x) - \sigma^2}{\phi(x)}\Big\}\; \de x
\le C(K) a.
\end{eqnarray*}
\end{propo}
We refer to Appendix~\ref{app:sigma_bound} for the proof of Proposition~\ref{propo:sigma_bound}.
\begin{propo} \label{propo:Etilde}
For each fixed point of continuum state evolution, satisfying the hypothesis of Lemma~\ref{lem:main_continuum}, there exists a constant $C(\kappa,K)$, such that,
\begin{eqnarray*}
\tEnergy_{\Shape}(\phi_a) - \tEnergy_{\Shape} (\phi) \le C(\kappa,K) a.
\end{eqnarray*}
\end{propo}
Proof of Proposition~\ref{propo:Etilde} is deferred to Appendix~\ref{app:Etilde}.

Using Eq.~\eqref{eqn:energy-potential} and Proposition~\ref{propo:Etilde}, we have
\begin{eqnarray}\label{eqn:energy_diff}
\Energy_{\Shape}(\phi_a) - \Energy_{\Shape}(\phi) 
\le \int_{-1}^{\ell-1} \big\{V(\phi_a(x)) - V(\phi(x)) \big\} \de x + C(\kappa,K) a,
\end{eqnarray}
where the constants $(\delta/2) C(K)$ and $C(\kappa,K)$ are absorbed in $C(\kappa,K)$. 

We proceed by proving the following proposition. Its proof is deferred to Appendix~\ref{app:potential_Etilde}.
\begin{propo} \label{propo:potential_Etilde}
For any $C = C(\kappa,K)$, there exists $\sigma_0$, such that for $\sigma \in (0,\sigma_0]$ the following holds.
\begin{eqnarray}
\int_{-1}^{\ell-1} \big\{V(\phi_a(x)) - V(\phi(x)) \big\} \de x < - 2 C(\kappa,K) a.
\end{eqnarray}
\end{propo} 

Fix $C(\kappa,K) > 0$. As a result of Eq.~\eqref{eqn:energy_diff} and Proposition~\ref{propo:potential_Etilde},
\begin{eqnarray}\label{eqn:diff_energy}
\begin{split}
\Energy_{\Shape}(\phi_a) - \Energy_{\Shape}(\phi) &
< \int_{-1}^{\ell-1} \big\{V(\phi_a(x)) - V(\phi(x)) \big\} \de x + C(\kappa,K) a\\
&\le -C(\kappa,K)a\, .
\end{split}
\end{eqnarray}

Since $\phi$ is a Lipschitz function by assumption, it is easy to see that $\|\phi_a - \phi\|_2 \le C\,a$, for some constant $C$. By Taylor expansion of the free energy functional around function $\phi$, we have
\begin{eqnarray}\label{eqn:taylor_energy}
\begin{split}
\langle \nabla \Energy_{\Shape}(\phi), \phi_a - \phi \rangle &= 
\Energy_{\Shape}(\phi_a) - \Energy_{\Shape}(\phi) + o(\|\phi_a - \phi\|_2)\\
&\le -C(\kappa,K)a + o(a).
\end{split}
\end{eqnarray}
However, since $\{\phi, \psi\}$ is a fixed point of the continuum state evolution, we have $\nabla \Energy_{\Shape}(\phi) = 0$ on the interval $[-1,\ell-1]$ (cf. Corollary~\ref{coro:fix_gradiant}). Also, $\phi_a - \phi$ is zero out of $[-1,\ell-1]$. Therefore, $\langle \nabla \Energy_{\Shape}(\phi), \phi_a - \phi \rangle = 0$, which leads to a contradiction in Eq~\eqref{eqn:taylor_energy}. This implies that our first assumption $\int_{-1}^{\ell-1} |\phi(x) - \phi^*| \;\de x > \kappa \ell$ is false. The result follows.
\end{proof}

%
\subsubsection{Analysis of the continuum state evolution: robust reconstruction}

Next lemma pertains to the robust reconstruction of the signal. Prior to stating the lemma, we need to establish some definitions. Due to technical reasons in the proof, we consider an alternative decomposition of $\Energy_{\Shape}(\phi)$ to Eq.~\eqref{eqn:energy-potential}.

Define the potential function $V_{\rob}: \reals_{+} \to \reals_{+}$ as follows.
\begin{eqnarray}\label{eqn:potential_rob}
V_{\rob}(\phi) = \frac{\delta}{2}\Big(\frac{\sigma^2}{\phi} + \log \phi \Big),
\end{eqnarray}
and decompose the Energy functional as:
\begin{eqnarray} \label{eqn:energy-potential_rob}
\Energy_{\Shape}(\phi) = \int_{-1}^{\ell-1} V_{\rob}(\phi(x)) \;\de x +\frac{\delta}{2} \int_{-1}^{\ell-1} \frac{\varsigma^2(x) - \sigma^2}{\phi(x)}\; \de x +  \tEnergy_{\Shape,\rob}(\phi),
\end{eqnarray}
with,
\begin{eqnarray}
\tEnergy_{\Shape,\rob}(\phi) = \int_{0}^{\ell} \Info(\Shape \ast {\phi}(y)^{-1}) \de y.
\end{eqnarray}

\begin{lemma}\label{lem:main_continuum2}
Let $\delta > 0$, and $p_X$ be a probability measure on the real line with $\delta > \uMMSE(p_X)$. For any $0 <\alpha <1$, there exist $\ell_0 = \ell_0(\alpha)$, $\sigma_0^2 = \sigma_0(p_X, \delta, \alpha)^2$, such that , for any $\ell > \ell_0$ and $\sigma \in (0,\sigma_0]$, and for any fixed point of continuum state evolution, $\{\phi, \psi \}$, with $\psi$ and $\phi$ nondecreasing Lipschitz functions and $\psi(x) \ge \mmse(\L0/(2\sigma^2))$, the following holds.
\begin{eqnarray}
\int_{-1}^{\ell - 1} |\phi(x) - \phi^*| \;\de x \le C\sigma^2 \ell\,,
\end{eqnarray}
with $C = \frac{2\delta}{(1-\alpha)(\delta - \uMMSE(p_X))}$.
\end{lemma}
\begin{proof}

Suppose $\int_{-1}^{\ell-1} |\phi(x) - \phi^*| \de x > C\sigma^2 \ell$, for the given $C$. Similar to the proof of Lemma~\ref{lem:main_continuum}, we obtain an infinitesimal perturbation of $\phi$ that decreases the free energy in the first order, contradicting the fact $\nabla \Energy_{\Shape}(\phi) = 0$ on the interval $[-1,\ell-1]$.

By definition of upper MMSE dimension (Eq.~\eqref{eqn:ummse_def}), for any $\ve > 0$, there exists $\phi_1$, such that, for $\phi \in [0,\phi_1]$,
\begin{eqnarray}\label{eq:mmse_approx}
\mmse(\phi^{-1}) \le (\uMMSE(p_X) + \ve) \phi.
\end{eqnarray}
Henceforth, fix $\ve$ and $\phi_1$.
\begin{claim}\label{claim:slope_phi2}
For each fixed point of continuum state evolution that satisfies the hypothesis of Lemma~\ref{lem:main_continuum2}, the following holds. For any $K > 0$, $0< \alpha< 1$, there exist $\ell_0 = \ell_0(\alpha)$ and $\sigma_0 = \sigma_0(\ve,\alpha, p_X, \delta)$, such that for $\ell > \ell_0$  and $\sigma \in (0,\sigma_0]$, there exist $x_1 < x_2 \in [0,\ell-1)$, with $x_2 - x_1 = K$ and $C\sigma^2 (1-\alpha) \le \phi(x) \le \phi_1$, for $x \in [x_1,x_2]$.
\end{claim}

Claim~\ref{claim:slope_phi2} is proved in Appendix~\ref{app:slope_phi2}. For positive values of $a$, define
\begin{eqnarray}
\phi_a(x) = \begin{cases}
\phi(x), & \text{for } x \le x_1, x_2 \le x,\\
(1-a)\phi(x) & \text{for } x \in (x_1,x_2).
\end{cases}
\end{eqnarray}
Our aim is to show that $\Energy_{\Shape}(\phi_a) - \Energy_{\Shape}(\phi) \le -c\;a$, for some constant $c > 0$.

Invoking Eq.~\eqref{eqn:energy-potential}, we have
\begin{align}
\begin{split}
\Energy_{\Shape}(\phi_a) - \Energy_{\Shape}(\phi) &= 
\int_{-1}^{\ell-1} \{V_{\rob} \phi_a(x)) - V_{\rob}(\phi(x))\}\;\de x \\
&\quad+ \frac{\delta}{2}\int_{-1}^{\ell-1}(\varsigma^2(x) - \sigma^2)\left(\frac{1}{\phi_a(x)} - \frac{1}{\phi(x)}\right)\;\de x
+ \tEnergy_{\Shape,\rob}(\phi_a) - \tEnergy_{\Shape,\rob}(\phi).
\end{split}
\end{align}

The following proposition bounds each term on the right hand side separately.
\begin{propo}\label{pro:R_terms}
For the function $\phi(x)$ and its perturbation $\phi_a(x)$, we have
\begin{align}
&\int_{-1}^{\ell-1} \{V_{\rob}(\phi_a(x)) - V_{\rob}(\phi(x))\}\;\de x  \le
\frac{\delta}{2} K \log(1-a) + K\frac{\delta a}{2C(1-\alpha)(1-a)} ,\label{eqn:R_terms1}\\
&\int_{-1}^{\ell-1}(\varsigma^2(x) - \sigma^2)\left(\frac{1}{\phi_a(x)} - \frac{1}{\phi(x)}\right)\;\de x \le
K\frac{a}{C(1-\alpha)(1-a)},\label{eqn:R_terms2}\\
&\tEnergy_{\Shape,\rob}(\phi_a) - \tEnergy_{\Shape,\rob}(\phi) \le 
-\frac{\uMMSE(p_X)+\ve}{2} (K+2) \log(1-a).\label{eqn:R_terms3}
\end{align}
\end{propo}
We refer to Appendix~\ref{app:R_terms} for the proof of Proposition~\ref{pro:R_terms}.

Combining the bounds given by Proposition~\ref{pro:R_terms}, we obtain
\begin{eqnarray}\label{eqn:allterms}
\begin{split}
\Energy_{\Shape}(\phi_a) &- \Energy_{\Shape}(\phi) \le 
\frac{K}{2} \log(1-a) \Big\{ \delta - (\uMMSE(p_X) + \ve)(1+\frac{2}{K}) \Big\}+ K \frac{\delta a}{C(1-\alpha)(1-a)}.
\end{split}
\end{eqnarray}
Since $\delta > \uMMSE(p_X)$ by our assumption, and $C = \frac{2\delta}{(1-\alpha)(\delta - \uMMSE(p_X))}$, there exist $\ve, a$ small enough and $K$ large enough, such that
\begin{eqnarray*}
c = \delta - (\uMMSE(p_X) +\ve)(1+\frac{2}{K}) -\frac{2 \delta}{C(1-\alpha)(1-a)} > 0.
\end{eqnarray*}
Using Eq.~\eqref{eqn:allterms}, we get
\begin{eqnarray}
\Energy_{\Shape}(\phi_a) - \Energy_{\Shape}(\phi) \le - \frac{cK}{2}a.
\end{eqnarray}
By an argument analogous to the one in the proof of Lemma~\ref{lem:main_continuum}, this is in contradiction with $\nabla \Energy_{\Shape}(\phi) = 0$. The result follows.
\end{proof}
 
\subsection{Proof of Lemma~\ref{lem:phi_convergence}}

By Lemma~\ref{lem:modified_dominance}, $\phi_a(t) \le \phi^{mod}_a(t - t_0)$, for $a \in \Rows_0\cong \{\rho^{-1}, \cdots, L-1+\rho^{-1}\}$ and $t \ge t_1(L_0,\delta)$. Therefore, we only need to prove the claim for the modified state evolution. The idea of the proof is as follows. In the previous section, we analyzed the continuum state evolution and showed that at the fixed point, the function $\phi(x)$ is close to the constant $\phi^*$. Also, in Lemma~\ref{lemma:ContinuumLimit}, we proved that the modified state evolution is essentially approximated by the continuum state evolution as $\rho \to 0$. Combining these results implies the thesis.

\begin{proof}[Proof (Part(a))]{ 
By monotonicity of continuum state evolution (cf. Corollary~\ref{cor:monotone_map_continuous}), $\lim_{t\to \infty} \phi(x;t) = \phi(x)$ exists. Further, by continuity of state evolution recursions, $\phi(x)$ is a fixed point. Finally, $\phi(x)$ is a nondecreasing Lipschitz function (cf. Corollary~\ref{cor:nondecreasing_continuous}). Using Lemma~\ref{lem:main_continuum} in conjunction with the Dominated Convergence theorem, we have, for any $\ve > 0$ 
\begin{eqnarray}
\lim_{t \to \infty} \frac{1}{\ell}  \int_{-1}^{\ell-1} |\phi(x;t) - \phi^*| \de x \le \frac{\ve}{4},
\end{eqnarray}
for $\sigma \in (0,\sigma_0^2]$ and $\ell > \ell_0$. Therefore, there exists $t_2 > 0$ such that $\frac{1}{\ell}  \int_{-1}^{\ell-1} |\phi(x;t_2) - \phi^*| \de x \le \ve/2$. 
\noindent Moreover, for any $t \ge 0$,
\begin{eqnarray}\label{eqn:int_approximation}
\frac{1}{\ell} \int_{-1}^{\ell-1} |\phi(x;t) - \phi^*| \de x = \lim_{\rho \to 0} \frac{\rho}{\ell} \sum_{a= -\rho^{-1}}^{L-\rho^{-1} - 1} |\phi(\rho a;t) - \phi^*| =\lim_{\rho \to 0} \frac{1}{L} \sum_{a=-\rho^{-1}}^{L-\rho^{-1}-1} |\phi(\rho a;t) - \phi^*|.
\end{eqnarray}
 
\noindent By triangle inequality, for any $t \ge 0$,
 \begin{eqnarray}
 \begin{split}
 \lim_{\rho \to 0} \frac{1}{L}  \sum_{a=-\rho^{-1}}^{L-\rho^{-1}-1} |\phi_a(t) - \phi^*|
 &\le  \lim_{\rho \to 0} \frac{1}{L}  \sum_{a=-\rho^{-1}}^{L-\rho^{-1}-1} |\phi_a(t) - \phi(\rho a; t)|
 + \lim_{\rho \to 0} \frac{1}{L}  \sum_{a=-\rho^{-1}}^{L-\rho^{-1}-1} |\phi(\rho a;t) - \phi^*|\\
 & = \frac{1}{\ell} \int_{-1}^{\ell-1} |\phi(x;t) - \phi^*| \de x,
\end{split}
\end{eqnarray}
where the last step follows from Lemma~\ref{lemma:ContinuumLimit} and Eq.~\eqref{eqn:int_approximation}. Since the sequence $\{\phi(t)\}$ is monotone decreasing in $t$, we have
\begin{eqnarray}
\begin{split}
\lim_{\rho \to 0} \lim_{t \to \infty} \frac{1}{L} \sum_{a=-\rho^{-1}}^{L-\rho^{-1}-1} \phi_a(t)  &\le
\lim_{\rho \to 0} \frac{1}{L} \sum_{a=-\rho^{-1}}^{L-\rho^{-1}-1} \phi_a(t_2)\\
&\le \lim_{\rho \to 0} \frac{1}{L} \sum_{a=-\rho^{-1}}^{L-\rho^{-1}-1} (|\phi_a(t_2) - \phi^*| + \phi^*)\\
&\le \frac{1}{\ell} \int_{-1}^{\ell-1} |\phi(x;t_2) - \phi^*| \de x + \phi^* \\
&\le \frac{\ve}{2} + \phi^*. 
\end{split}
\end{eqnarray}


\noindent Finally,
\begin{eqnarray}\label{eqn:final}
\begin{split}
\lim_{t \to \infty} \sum_{a = -\rho^{-1}}^{L +\rho^{-1}-1} \phi_a(t) & \le \frac{2\rho^{-1}}{L} \Phi_M + \frac{\ve}{2} + \phi^*\\
&\le \frac{2\rho^{-1} }{L_*} \Phi_M + \frac{\ve}{2} + 2\sigma_0.
\end{split}
\end{eqnarray}
\noindent Clearly, by choosing $L_*$ large enough and $\sigma_0$ sufficiently small, we can ensure that the right hand side of Eq.~\eqref{eqn:final} is less than $\ve$.
}\end{proof} 

\begin{proof}[Proof (Part(b))]{
Consider the following two cases.
\begin{itemize}
\item $\sigma \le \sigma_0$:  In this case, proceeding along the same lines as the proof of Part $(a)$, and using Lemma~\ref{lem:main_continuum2} in lieu of Lemma~\ref{lem:main_continuum}, we have
\begin{align}
\lim_{t \to \infty} \frac{1}{L} \sum_{a=-\rho^{-1}}^{L-\rho^{-1}-1} \phi_a(t) &\le C \sigma^2 + \phi^* 
\le \left(\frac{2\delta}{(1-\alpha)(\delta - \uMMSE(p_X))} + 1+ \frac{2}{\delta \L0} \right) \sigma^2\,.\label{eqn:C-limit}
\end{align}

\item $\sigma > \sigma_0$: Since $\phi_a(t) \le \sigma^2 + (1/\delta) {\rm Var}(X)$ for any $t >0$, we have
\begin{eqnarray}
\lim_{t \to \infty} \frac{1}{L} \sum_{a=-\rho^{-1}}^{L-\rho^{-1}-1} \phi_a(t) \le \sigma^2 + \frac{1}{\delta} {\rm Var}(X).
\end{eqnarray}
\end{itemize}
Choosing 
\[
C = \max\Big\{\frac{2\delta}{(1-\alpha)(\delta - \uMMSE(p_X))} + 1+ \frac{2}{\delta \L0},\, 1+ \frac{{\rm Var}(X)}{\delta \sigma_0^2}\Big\}\,,
\]
 proves the claim in both cases. 
 
 Finally, in the asymptotic case where $\ell = L\rho \to \infty$, $\rho \to 0$, $\L0 \to \infty$, we have $\alpha \to 0$ and using Eq.~\eqref{eqn:C-limit}, we get
 \[
 \lim_{\sigma \to 0} \lim_{t\to \infty} \frac{1}{\sigma^2\, L} \sum_{a=-\rho^{-1}}^{L-\rho^{-1}-1} \phi_a(t) \le 
 \frac{3\delta -  \uMMSE(p_X)}{\delta - \uMMSE(p_X)} \,.
 \]
}\end{proof}

\section*{Acknowledgements}

A.M. would like to thank Florent Krzakala, Marc M\'ezard,
Fran\c{c}ois Sausset, Yifan Sun and Lenka Zdeborov\'a for a
stimulating exchange about their results.
A.J. is supported by a Caroline and Fabian Pease Stanford
Graduate Fellowship. This work was partially supported by the NSF
CAREER award CCF- 0743978, the NSF grant DMS-0806211, and the AFOSR
grant FA9550-10-1-0360.

\newpage
\appendix
\section{Dependence of the algorithm on the prior $p_X$}
\label{app:Prior}

In this appendix we briefly discuss the impact of a wrong estimation of the prior $p_X$ on the AMP algorithm.
Namely, suppose that instead of the true prior $p_X$, we have an approximation of $p_X$
denoted by $\ptx$. The only change in the algorithm is
in the posterior expectation denoiser. That is to say, the denoiser $\eta$ in Eq.~(\ref{eq:AMP1})
will be replaced by a new denoiser $\teta$. We will quantify the
discrepancy  between $p_X$ and $\ptx$ through their Kolmogorov-Smirnov distance
$\Dks (\px,\ptx)$. Denoting by $F_X(z) = p_X((-\infty,z])$ and
$F_{\tX}(z) = \ptx((-\infty,z])$ the corresponding distribution
functions, we have 
\begin{eqnarray*}
\Dks (\px,\ptx) =\sup_{z\in \reals} \,
\big|F_X(z)-F_{\tX}(z)\big|\,. 
\end{eqnarray*}
The next lemma 
establishes a bound on the pointwise distance between $\eta$ and $\teta$
in terms of $\Dks (\px,\ptx)$.

Note that state evolution (\ref{eq:ExplicitSE}) applies also to the
algorithm with the mismatched denoiser, provided the
$\mmse(\,\cdot\,)$ function is replaced by the mean square error for
the non-optimal denoiser $\teta$. Hence the bound on
$|\eta(y)-\teta(y)|$ given below can be translated into a bound on
the performance of AMP with the mismatched prior.
 A full study of this issue goes beyond the scope of this paper and will be the object of a forthcoming
publication.

For the sake of simplicity we shall assume that $\px, \ptx$ have
bounded supports. The general case requires a more careful consideration.
\begin{lemma}
Let $\eta : \reals \to \reals$ be the Bayes optimal estimator for
estimating $X\sim \px$ in Gaussian noise
$\eta(y) = \E(X| X+ Z = y)$, with $Z\sim\normal(0,1)$. Define denoiser $\teta$ similarly, with 
respect to $\ptx$. Assume that $\px$ is supported in $[-M,M]$.
%
Then for any $\ptx$ supported in $[-M,M]$, we have
\begin{eqnarray*}
%
|\eta(y) - \teta(y)| \le \,\frac{M(15+10 M|y|)}{\E\{e^{-X^2/2}\}}\, \Dks(\px,\ptx)\, e^{2M|y|}\,.
\end{eqnarray*}
\end{lemma}
\begin{proof}
Throughout the proof we let $\Delta \equiv \Dks(\px,\ptx)$, and $\Delta_1 \equiv \E\{e^{-X^2/2}\}$.

Let $\gamma(z) = \exp(-z^2/2)/\sqrt{2\pi}$ be the Gaussian density. We
then have $\eta(y) = \E\{X\gamma(X-y)\}/\E\{\gamma(X-y)\}$.
Let $p_W$ be the probability measure with Radon-Nikodym derivative
with respect to $p_X$ given by 
\begin{align*}
\frac{\de p_W}{\de p_X}(x) = \frac{e^{-x^2/2}}{\E\{e^{-X^2/2}\}}\, .
\end{align*}
We define $p_{\tW}$ analogously from the measure $\ptx$ and let $W,\tW$
be two random variables with law $p_W$ and $p_{\tW}$, respectively. We then have
\begin{align}
\eta(y) = \frac{\E\{We^{yW}\}}{\E\{e^{yW}\}} \, .\label{eq:EtaFormula}
\end{align}
Letting
$F_{W}$, $F_{\tW}$ denote the corresponding distribution functions, we
have
\begin{align*}
F_W(x)=
\int_{-\infty}^x \de p_W(w) = \frac{\int_{-\infty}^x e^{-z^2/2}\, \de p_X(z)}{\E\{e^{-X^2/2}\}}
= \frac{e^{-x^2/2}F_X(x)+\int_{-\infty}^{x}ze^{-z^2/2}F_X(z)\,\de
  z}{\int_{-\infty}^{\infty}ze^{-z^2/2}F_X(z)\,\de z}\, .
\end{align*}
Letting $N_W(x)$ be the numerator in this expression, we have
\begin{align*}
\big|N_W(x)-N_{\tW}(x)\big|\le \big|F_X(x)-F_{\tX}(x)\big|+
\int_{-\infty}^{x} |z|\,e^{-z^2/2}\big|F_X(z)-F_\tX(z)\big|\, \de z \le
3\Delta\, .
\end{align*}
Proceeding analogously for the denominator, we have
\begin{align*}
\big|\E\{e^{-X^2/2}\} - \E\{e^{-\tX^2/2}\}\big| \le 
\int_{-\infty}^{\infty} |z|\,e^{-z^2/2}\big|F_X(z)-F_\tX(z)\big|\, \de z \le
2\Delta\, .
\end{align*}
Combining these bounds, we obtain
\begin{eqnarray}
\begin{split}\label{eqn:frac-trick}
 \big|F_W(x)-F_{\tW}(x)\big| & = \Big|\frac{N_W(x)}{\E\{e^{-X^2/2}\}} - \frac{N_{\tW}(x)}{\E\{e^{-\tX^2/2}\}} \Big|\\
 &\le \Big|\frac{N_W(x) - N_{\tW}(x)}{\E\{e^{-X^2/2}\}} \Big| + \Big| N_{\tW}(x) \Big(\frac{1}{\E\{e^{-X^2/2}\}}  - \frac{1}{\E\{e^{-\tX^2/2}\}} \Big)\Big|\\
 & = \frac{\big|N_W(x) - N_{\tW}(x)\big|}{\E\{e^{-X^2/2}\}} + F_{\tW}(x)\, \frac{\big|\E\{e^{-X^2/2}\} - \E\{e^{-\tX^2/2}\}\big|}{\E\{e^{-X^2/2}\}}\\
 &\le \frac{3\Delta}{\Delta_1} + \frac{2\Delta}{\Delta_1} = \frac{5\Delta}{\Delta_1}\,.
\end{split}
\end{eqnarray}
Since, the above inequality holds for any $x \in \reals$, we get
\begin{align}\label{eqn:bound-DSKW}
\Dks(p_W,p_{\tW})\le \frac{5\Delta}{\Delta_1}\, .
\end{align}
Consider now Eq.~(\ref{eq:EtaFormula}). We have
\begin{align*}
\big|\E\{e^{yW}\}-\E\{e^{y\tW}\}\big| &= |y|\,\int
e^{yx}\big|F_W(x)-F_{\tW}(x)\big|\,\de x\\
&\le |y| \Dks(p_W,p_{\tW}) \int_{-M}^M e^{yx} \, \de x
\le e^{M|y|} \Dks(p_W,p_{\tW})\,.
\end{align*}
We proceed analogously for the numerator, namely,
\begin{align*}
\big|\E\{We^{yW}\}-\E\{\tW e^{y\tW}\}\big| &= \int(1+|yx|)
e^{yx}\big|F_W(x)-F_{\tW}(x)\big|\,\de x\\
& \le \Dks(p_W,p_{\tW}) \int_{-M}^M (1+|yx|)e^{yx}\,\de x
\le 2M(1+ M |y|)e^{M|y|}
\Dks(p_W,p_{\tW})\, .
\end{align*}
Combining these bounds and proceeding along similar lines to Eq.~\eqref{eqn:frac-trick}, we obtain
%
\begin{align}\label{eqn:eta-diff-bound}
\big|\teta(y)-\eta(y)\big| \le \frac{2M(1+M|y|) + \teta(y)}{\E\{e^{yW}\}}\, e^{M|y|}\, \Dks(p_W,p_{\tW})\,.
\end{align}
Note that $\teta(y)\in [-M,M]$ since $\ptx$ is supported on $[-M,M]$, and thus $|\teta(y)| \le M$. Also, $p_W$ is supported on $[-M,M]$ since it is absolutely continuous with respect to $p_X$ and $p_X$ is supported on $[-M,M]$. Therefore, $\E\{e^{yW}\} \ge e^{-M|y|}$. Using these bounds in Eq.~\eqref{eqn:eta-diff-bound}, we obtain
\begin{align}
\big|\teta(y)-\eta(y)\big| \le {M(3+2M|y|)}\, e^{2M|y|}\, \Dks(p_W,p_{\tW})\,.
\end{align}
The result follows by plugging in the bound given by Eq.~\eqref{eqn:bound-DSKW}. 
\end{proof}

\section{Lipschitz continuity of AMP}
\label{app:Lipschitz}
Let $x^t$ be the Bayes optimal AMP estimation at iteration $t$ as
given by Eqs.~\eqref{eq:AMP1},~\eqref{eq:AMP2}. We show that for each fixed
iteration number $t$, the mapping $y \to x^t(y)$ is locally Lipschitz
continuous.
\begin{lemma}\label{app:lem-lip}
For any $R,B>0$, $t\in\naturals$,  there exists $L=L(R,B;t)<\infty$ such
that  for any  $y, \ty\in\reals^m$ with $\|y\|, \|\ty\| \le R$, and
any matrix $A$ with $\|A\|_2\le B$ we have
\begin{eqnarray}
\|x^t(y)-x^t(\ty)\|\le L\, \|y-\ty\|\, .
\end{eqnarray}
\end{lemma}
Note that in the statement we assume $\|A\|_2$ to be finite. This
happens as long as the entries of $A$ are bounded and hence almost
surely within our setting. 

Also, we assume $\|y\|$, $\|\ty\|\le R$ for some fixed $R$.
In  other words, we prove that the algorithm is locally Lipschitz. 
We can obtain an algorithm that is globally Lipschitz 
by defining $x^t(y)$ via the AMP iteration for $\|y\|\le R$, and by an
arbitrary bounded Lipschitz extension for $\|y\|\ge R$.
Notice that $\|y\|\le B\|x\|+\|w\|$, and, by the law of large numbers, $\|x\|^2\le (\E\{X^2\}+\eps)n$, $\|w\|^2\le
(\sigma^2+\eps)m$ with probability converging to $1$. 
Hence, the globally Lipschitz modification of AMP achieves the same
performance as the original AMP, almost surely. (Note that $R$ can depend on $n$).

\begin{proof}[Proof (Lemma~\ref{app:lem-lip})]
Suppose that we have two measurement vectors $y$ and $\ty$.
Note that the state evolution is completely characterized in terms of 
prior $\px$ and noise variance $\sigma^2$, and can be precomputed
(independent of measurement vector).

Let $(x^t, r^{t})$ correspond to the AMP with measurement vector $y$ and
$(\tx^t, \tr^{t})$ correspond to the AMP with measurement vector $\ty$. (To clarify, note that $x^t \equiv x^t(y)$ and $\tx^t \equiv x^t(\ty)$).
Further define
\begin{eqnarray*}
\xi_t = \max(\|x^{t} - \tx^{t}\|, \|r^{t} - \tr^{t}\|, \|y - \ty\|)\,.
\end{eqnarray*}
We show that 
\begin{eqnarray}\label{eqn:xi}
\xi_t \le C_{t} (1+ \|y\|)\, \xi_{t-1}\,,
\end{eqnarray}
for a constant $C_t$. This establishes the claim since
\begin{eqnarray*}
\|x^{t} - \tx^{t}\| \le \xi_t \le C_{t} C_{t-1}\dotsc C_2\, (1+\|y\|)^{t-1} \xi_{1} = C_t C_{t-1} \dotsc C_2\, (1+\|y\|)^{t-1} \|y - \ty\|\,,
\end{eqnarray*}
where the last step holds since $x_i^1 = \tx_i^1 = \E\{X\}$ and $r^1 - \tr^1 = y - \ty$.

In order to prove Eq.~\eqref{eqn:xi}, we need to prove the following two claims.
\begin{claim}\label{claim:bound-r}
For any fixed iteration number $t$, there exists a constant $C_t$, such that 
\[
\|r^t\| \le C_t \max(\|x^1\|, \|y\|)\,.
\]
\end{claim}
\begin{proof}[Proof (Claim~\ref{claim:bound-r})]{
Define $\lambda_t = \max(\|x^{t+1}\|, \|r^t\|, \|y\|)$. Then, 
\[
\|r^t\| \le \|y\| + \|A\|_2 \|x^t\| + \|\ons^{t}\|_{\infty} \|r^{t-1}\|.
\]
Note that $A$ has bounded operator by assumption. Also, the posterior mean $\eta$ is a smooth function with bounded derivative. Therefore, recalling the definition of $\ons^{t}$,
\[
\ons^{t} \equiv \frac{1}{\delta} \sum_{u\in \Cols} W_{\gr(i),u} \tQ^{t-1}_{\gr(i),u} \<\eta'_{t-1}\>_u \,,
\]
we have $\|\ons^t\|_{\infty} \le C_{1,t}$ for some constant $C_{1,t}$. Hence,
$\|r^t\| \le C_{2,t} \lambda_{t-1}$.
Moreover,
\begin{align*}
\|x^{t+1}\|  = \|\eta_t(x^t + (Q^t\odot A)^* r^t)\| \le C (\|x^t \| + \|Q^t\odot A\|_2 \|r^t\|)
\le C_{3,t} \max(\|x^t\|, \|r^t\|)\,,
\end{align*}
for some constant $C_{3,t}$. In the first inequality, we used the fact that $\eta$ is Lipschitz continuous. Therefore, $\lambda_t \le C'_t \lambda_{t-1}$, where $C'_t = \max(1, C_{2,t}, C_{3,t}, C_{2,t}\, C_{3,t})$, and
\[
\|r^t\| \le \lambda_t \le C'_t \cdots C'_{1} \lambda_0 \le C'_t \cdots C'_1 \max(\|x^1\|,\|y\|),
\]
with $x^1_i = \E\{X\}$, for $i \in [n]$.
}\end{proof}

\begin{claim}\label{claim:ons-diff-bound}
For any fixed iteration number $t$, there exists a constant $C_t$, such that 
\[
\|\ons^{t} \odot r^{t-1} - \tilde{\ons}^{t} \odot \tr^{t-1}\| \le C_t (1+\|y\|) \max(\|x^{t-1} - \tx^{t-1}\|,\|r^{t-1} - \tr^{t-1}\|)\,.
\]
\end{claim}
\begin{proof}[Proof (Claim~\ref{claim:ons-diff-bound})]{
Using triangle inequality, we have
\begin{align}\label{eqn:tri-onsager}
\|\ons^{t} \odot r^{t-1} - \tilde{\ons}^{t} \odot \tr^{t-1}\|
\le \|(\ons^{t} - \tilde{\ons}^{t}) \odot r^{t-1}\|
+\|\tilde{\ons}^{t} \odot (r^{t-1} - \tr^{t-1})\|\,.
\end{align}
Since $\eta'$ is Lipschitz continuous, we have
\[
\|\ons^{t} - \tilde{\ons}^{t}\| \le C_{1,t} (\|x^{t-1} - \tx^{t-1}\| + \|r^{t-1} - \tr^{t-1}\|)\,,
\]
for some constant $C_{1,t}$. Also, as discussed in the proof of Claim~\ref{claim:bound-r}, the Onsager terms $\ons^t$ are
uniformly bounded. Applying these bounds to the right hand side of Eq.~\eqref{eqn:tri-onsager}, we obtain
\begin{align*}
\|\ons^{t} \odot r^{t-1} - \tilde{\ons}^{t} \odot \tr^{t-1}\| &\le C_{1,t}\, (\|x^{t-1} - \tx^{t-1}\| + \|r^{t-1} - \tr^{t-1}\|)\, \|r^{t-1}\| + C_{2,t}\, \|r^{t-1} - \tr^{t-1}\|\\
&\le C_t (1+\|y\|) \max(\|x^{t-1} - \tx^{t-1}\|,\|r^{t-1} - \tr^{t-1}\|)\,,
\end{align*}
for some constants $C_{1,t}, C_{2,t}, C_t$. The last inequality here follows from the bound given in Claim~\ref{claim:bound-r}.
}\end{proof}

Now, we are ready to prove Eq.~\eqref{eqn:xi}. We write
\begin{align}
\|x^{t} - \tx^{t}\| &= \|\eta_{t-1}(x^{t-1} + (Q^{t-1}\odot A)^* r^{t-1}) - \eta_{t-1}(\tx^{t-1} + (Q^{t-1}\odot A)^* \tr^{t-1})\|\nonumber\\
&\le C \left(\|x^{t-1} - \tx^{t-1}\| + \|Q^{t-1}\odot A\|_2 \|r^{t-1}-\tr^{t-1}\| \right) \nonumber \\
&\le C_{1,t} \max(\|x^{t-1} - \tx^{t-1}\|, \|r^{t-1}- \tr^{t-1}\|, \|y-\ty\|) = C_{1,t}\, \xi_{t-1}\,,\label{eq:x}
\end{align}
for some constant $C_{1,t}$. Furthermore,
\begin{align}
\|r^{t} - \tr^{t}\| &\le \|y-\ty\| + \|A\|_2 \|x^t - \tx^t\| + \|\ons^{t} \odot r^{t-1} - \tilde{\ons}^{t} \odot \tr^{t-1}\| \nonumber\\
&\le \|y-\ty\| + \|A\|_2 \,C_{1,t}\, \xi_{t-1} + C'_t (1+\|y\|) \max(\|x^{t-1} - \tx^{t-1}\|,\|r^{t-1} - \tr^{t-1}\|) \nonumber \\
&\le C_{2,t}\,(1+\|y\|)\, \xi_{t-1}\,,\label{eq:r}
\end{align}
for some constant $C_{2,t}$ and using Eq.~\eqref{eq:x} and Claim~\ref{claim:ons-diff-bound} in deriving the second inequality.
Combining Eqs.~\eqref{eq:x} and~\eqref{eq:r}, we obtain
\[
\xi_t \le \max(1, C_{1,t}, C_{2,t})\,(1+\|y\|)\, \xi_{t-1}\,.
\]
\end{proof}
%

\section{Proof of Lemma~\ref{lemma:ContinuumLimit}}
\label{app:ContinuumLimit}
We prove the first claim, Eq.~(\ref{eq:FirstSEApprox}). The second one
follows by a similar argument. The proof uses induction on $t$. It is
a simple exercise to show that the  induction basis ($t=1$) holds (the
calculation follows the same lines as the induction step). Assuming
the claim for $t$, we write,
for $i\in \{0,1,\dots,L-1\}$
\begin{eqnarray}
\begin{split}
|\psi_i(t+1) - \psi(\rho i;t+1) | &=
 \Big| \mmse \Big( \sum_{b \in \Rows_0} W_{b-i}\; [\sigma^2 + \frac{1}{\delta} \sum_{j \in \integers} W_{b-j} \psi_j(t)]^{-1}\Big)\\
 & \quad - \mmse \Big( \int_{-1}^{\ell+1} \Shape(z-\rho i)\; [\sigma^2
 + \frac{1}{\delta} \int_{\reals} \Shape(z- y) \psi(y;t)\de y]^{-1} \de z\Big)
\Big| \\
 &\le \Big| \mmse \Big( \sum_{b \in \Rows_0} W_{b-i}\; [\sigma^2 + \frac{1}{\delta} \sum_{j \in \integers} W_{b-j} \psi_j(t)]^{-1}\Big) \\
 & \quad -\mmse \Big( \sum_{b \in \Rows_0} W_{b-i}\; [\sigma^2 + \frac{1}{\delta} \sum_{j \in \integers} W_{b-j} \psi(\rho j;t)]^{-1}\Big) \Big| \\
&\quad + \Big| \mmse \Big( \sum_{b \in \Rows_0} \rho \Shape(\rho(b-i))\; [\sigma^2 + \frac{1}{\delta} \sum_{j \in \integers} \rho\Shape(\rho(b-j)) \psi(\rho j;t)]^{-1}\Big)\\
&\quad - \mmse \Big( \int_{-1}^{\ell+1} \Shape(z-\rho i)\; [\sigma^2 +
\frac{1}{\delta} \int_{\reals} \Shape(z- y) \psi(y;t) \de y]^{-1} \de z\Big)
\Big|.\label{eq:BigUBMMSE}
\end{split}
\end{eqnarray}
Now, we bound the two terms on the right hand side separately. Note
that the arguments of $\mmse(\,\cdot\,)$ in the above terms are at
most $2/\sigma^2$. Since $\mmse$ has a continuous derivative, there
exists a constant $C$ such that $|\frac{\de}{\de s}\; \mmse(s)| \le C$, for $s \in
[0, 2/\sigma^2]$. Then, considering the first term in the
upper bound (\ref{eq:BigUBMMSE}), we have
\begin{eqnarray}\label{eqn:first_term}
\begin{split}
\Big| \mmse \Big( \sum_{b \in \Rows_0} W_{b-i}\; &[\sigma^2 + \frac{1}{\delta} \sum_{j \in \integers} W_{b-j} \psi_j(t)]^{-1}\Big) 
 -\mmse \Big( \sum_{b \in \Rows_0} W_{b-i}\; [\sigma^2 + \frac{1}{\delta} \sum_{j \in \integers} W_{b-j} \psi(\rho j;t)]^{-1}\Big) \Big|\\
 &\le C \Big|  \sum_{b \in \Rows_0} W_{b-i}\; \Big([\sigma^2 + \frac{1}{\delta} \sum_{j \in \integers} W_{b-j} \psi_j(t)]^{-1}
 -   [\sigma^2 + \frac{1}{\delta} \sum_{j \in \integers} W_{b-j} \psi(\rho j;t)]^{-1} \Big) \Big|\\
 &\le \frac{C}{\sigma^4} \sum_{b \in \Rows_0} W_{b-i}\;  \frac{1}{\delta} \Big|\sum_{j =-\infty}^{L-1} W_{b-j}(\psi(\rho j;t) - \psi_j(t))\Big|\\
 &\le \frac{C}{\delta \sigma^4} \sum_{b \in \Rows_0} W_{b-i}\;  \sum_{j =-\infty}^{L-1} W_{b-j} |\psi(\rho j;t) - \psi_j(t)|\\
 & = \frac{C}{\delta \sigma^4} \sum_{j=0}^{L-1} \Big( \sum_{b \in
   \Rows_0} W_{b-i} W_{b - j} \Big)\; |\psi(\rho j;t) - \psi_j(t)|\\
&\le  \frac{C}{\delta \sigma^4}
\Big( \sum_{i\in\integers} W_{i}^2 \Big)\;
 \sum_{j=0}^{L-1} |\psi(\rho j;t) - \psi_j(t)|\\
 & \le \frac{C'\rho}{\delta \sigma^4} \sum_{j=0}^{L-1} |\psi(\rho j;t) - \psi_j(t)|.
\end{split}
\end{eqnarray}
Here we used $\sum_{i\in\integers}
W_{i}^2=\sum_{i\in\integers }\rho^2\Shape(\rho i)^2\le
C\sum_{|i|\le \rho^{-1}}\rho^2\le C\rho$ (where the first
inequality follows from the fact that $\Shape$ is bounded).

To bound the second term in Eq.~(\ref{eq:BigUBMMSE}), note that
\begin{eqnarray}\label{eqn:second_term}
\begin{split}
&\Big| \mmse \Big( \sum_{b \in \Rows_0} \rho \Shape(\rho(b-i))\; [\sigma^2 + \frac{1}{\delta} \sum_{j \in \integers} \rho\Shape(\rho(b-j)) \psi(\rho j;t)]^{-1}\Big)\\
&\quad \quad \quad \quad \quad- \mmse \Big( \int_{-1}^{\ell+1}
\Shape(z-\rho i)\; [\sigma^2 + \frac{1}{\delta} \int_{\reals}
\Shape(z- y) \psi(y;t) \de y]^{-1} \de z\Big)
\Big|\\
& \le  C \Big| \sum_{b \in \Rows_0} \rho \Shape(\rho(b-i))\; [\sigma^2 + \frac{1}{\delta} \sum_{j \in \integers} \rho\Shape(\rho(b-j)) \psi(\rho j;t)]^{-1}\\
&\quad \quad \quad \quad \quad  - \int_{-1}^{\ell+1} \Shape(z-\rho i
)\; [\sigma^2 + \frac{1}{\delta} \int_{\reals} \Shape(z- y) \psi(y;t)
\de y]^{-1} \de z
\Big|\\
&\le C \Big| \sum_{b \in \Rows_0} \rho \Shape(\rho(b-i))\; [\sigma^2 + \frac{1}{\delta} \sum_{j \in \integers} \rho\Shape(\rho(b-j)) \psi(\rho j;t)]^{-1}\\
&\quad \quad \quad \quad \quad  - \sum_{b \in \Rows_0} \rho
\Shape(\rho(b-i))\; [\sigma^2 + \frac{1}{\delta} \int_{\reals}
\Shape(\rho b- y) \psi(y;t) \de y]^{-1} \Big|\\
&+ C \Big| \sum_{b \in \Rows_0} \rho \Shape(\rho(b-i))\; [\sigma^2 +
\frac{1}{\delta} \int_{\reals} \Shape(\rho b- y) \psi(y;t) \de y]^{-1}
\de z\\
&\quad \quad \quad \quad \quad  - \int_{-1}^{\ell+1} \Shape(z-\rho
i)\; [\sigma^2 + \frac{1}{\delta} \int_{\reals} \Shape(z- y) \psi(y;t)
\de y]^{-1} \de z \Big|\\
&\le \frac{C}{\delta \sigma^4} \sum_{b \in \Rows_0} \rho
\Shape(\rho(b-i)) \Big|\sum_{j \in \integers} \rho F_1(\rho
b;\rho j)-\int_{\reals} F_1(\rho b;y)\de y\Big|\\
& + C \Big|\sum_{b \in \Rows_0} \rho F_2(\rho b) -
\int_{-1}^{\ell+1} F_2(z)\de z \Big|
\end{split}
\end{eqnarray}
where $F_1(x;y) = \Shape(x-y)\psi(y;t)$ and $F_2(z) =  \Shape(z-\rho
i)\; [\sigma^2 + \frac{1}{\delta} \int_{\reals} \Shape(z- y) \psi(y;t)
\de y]^{-1}$.
Since the functions $\Shape(\,\cdot\,)$ and $\psi(\,\cdot\,)$ have
continuous (and thus bounded) derivative on compact interval $[0,\ell]$, the
same is true for $F_1$ and $F_2$. 
Using the standard convergence of Riemann sums to Riemann integrals,
 right hand side of Eq.~\eqref{eqn:second_term} can be bounded by $C_3 \rho / \delta \sigma^4$, for some constant $C_3$. Let $\epsilon_i(t) = |\psi_i(t) - \psi(\rho i;t)|$. Combining Eqs.~\eqref{eqn:first_term} and~\eqref{eqn:second_term}, we get
\begin{eqnarray}
\epsilon_i(t+1) \le \frac{\rho}{\delta \sigma^4} \left(C' \sum_{j=0}^{L-1} \epsilon_j(t) + C_3 \right). 
\end{eqnarray}
Therefore,
\begin{eqnarray}
\frac{1}{L} \sum_{i=0}^{L-1} \epsilon_i(t+1) \le \frac{\ell}{\delta \sigma^4} \left( \frac{C'}{L} \sum_{j=0}^{L-1} \epsilon_j(t) \right) + \frac{C_3 \rho}{\delta \sigma^4}.
\end{eqnarray}
The claims follows from the induction hypothesis.
\section{Proof of Proposition~\ref{propo:V_properties}}
\label{app:V_properties}
 By Eq.~\eqref{eqn: info_dimension}, for any $\ve > 0$, there exists $\phi_0$, such that for $0 \le \phi \le \phi_0$, 
\begin{eqnarray}
\Info(\phi^{-1}) \le \frac{\uRenyi(p_X)+ \ve}{2} \log(\phi^{-1}).
\end{eqnarray}
Therefore,
\begin{eqnarray}
V(\phi) \le \frac{\delta \sigma^2}{2 \phi} + \frac{\delta - \uRenyi(p_X) - \ve}{2} \log \phi.
\end{eqnarray}

Now let $\ve = (\delta - \uRenyi(p_X))/2$ and $\sigma_2 = \sqrt{\phi_0/2}$. Hence, for $\sigma \in (0,\sigma_2]$, we get $\phi^* < 2\sigma^2 \le \phi_0$. Plugging in $\phi^*$ for $\phi$ in the above equation, we get
\begin{eqnarray}
\begin{split}
V(\phi^*) &\le \frac{\delta \sigma^2}{2 \phi^*} + \frac{\delta - \uRenyi(p_X)}{4} \log \phi^* \\
&< \frac{\delta}{2} +   \frac{\delta - \uRenyi(p_X)}{4}  \log(2\sigma^2)\,.
\end{split}
\end{eqnarray}
%

\section{Proof of Claim~\ref{claim:slope_phi}}
\label{app:slope_phi}
Recall that $\kappa < \Phi_M$ and $\phi(x)$ is nondecreasing. Let 
\begin{eqnarray*}
0 < \theta = \frac{\Phi_M - \kappa}{\Phi_M - \frac{\kappa}{2}} < 1.
\end{eqnarray*}
We show that $\phi(\theta \ell - 1) \ge \kappa/2 + \phi^*$. If this is not true, using the nondecreasing property of $\phi(x)$, we obtain
\begin{eqnarray}
\begin{split}
\int_{-1}^{\ell-1} |\phi(x) - \phi^*|\; \de x &= \int_{-1}^{\theta
  \ell-1} |\phi(x) - \phi^*|\; \de x + \int_{\theta \ell-1}^{\ell-1}
|\phi(x) - \phi^*|\; \de x\\
& < \frac{\kappa}{2} \theta \ell  +  \Phi_M (1-\theta) \ell\\
& = \kappa \ell,
\end{split}
\end{eqnarray}
contradicting our assumption. Therefore, $\phi(x) \ge \kappa/2 + \phi^*$, for $\theta \ell - 1 \le x \le \ell-1$. For given $K$, choose $\ell_0 = K/(1-\theta)$. Hence, for $\ell > \ell_0$, interval $[\theta\ell-1,\ell-1)$ has length at least $K$. The result follows.
\section{Proof of Proposition~\ref{propo:sigma_bound}}
\label{app:sigma_bound}
We first establish some properties of function $\varsigma^2(x)$.
\begin{remark}\label{rem:sigma_prop1}
The function $\varsigma^2(x)$ as defined in Eq.~\eqref{eqn:new_sigma}, is
non increasing in $x$. Also, $\varsigma^2(x) = \sigma^2 +
(1/\delta)\;\mmse(\L0/(2\sigma^2))$, for $x\le -1$ and $\varsigma^2(x) =
\sigma^2$, for $x \ge 1$.  For $\delta \L0 >3$, we have $ \sigma^2 \le \varsigma^2(x) < 2 \sigma^2$.
\end{remark}
\begin{remark} \label{rem:sig_lip}
The function $\varsigma^2(x)/\sigma^2$ is Lipschitz continuous. More
specifically, there exists a constant $C$, such that,
$|\varsigma^2(\alpha_1) - \varsigma^2(\alpha_2)| < C \sigma^2 |\alpha_2 -
\alpha_1|$, for any two values $\alpha_1, \alpha_2$. Further, if
$L_0\delta>3$ we can take $C<1$.
\end{remark}
The proof of Remarks~\ref{rem:sigma_prop1} and~\ref{rem:sig_lip} are immediate from Eq.~\eqref{eqn:new_sigma}.

To prove the proposition, we split the integral over the intervals $[-1,-1+a), [-1+a,x_0+a),[x_0+a,x_2),[x_2,\ell-1)$, and bound each one separately.
Firstly, note that
\begin{eqnarray} \label{eqn:sigma_term1}
\int_{x_2}^{\ell-1} \Big \{ \frac{\varsigma^2(x) - \sigma^2}{\phi_a(x)} -
\frac{\varsigma^2(x) - \sigma^2}{\phi(x)}\Big\} \; \de x = 0,
\end{eqnarray}
since $\phi_a(x)$ and $\phi(x)$ are identical  for $x\ge x_2$.

Secondly, let $\alpha = (x_2-x_0)/(x_2-x_0-a)$, and $\beta = (ax_2)/(x_2-x_0-a)$. Then,
\begin{eqnarray}\label{eqn:sigma_term2}
\begin{split}
&\int_{x_0+a}^{x_2} \Big \{ \frac{\varsigma^2(x) - \sigma^2}{\phi_a(x)} -
\frac{\varsigma^2(x) - \sigma^2}{\phi(x)}\Big\} \; \de x\\
&\quad = \int_{x_0}^{x_2} \frac{\varsigma^2(\frac{x+\beta}{\alpha}) -
  \sigma^2}{\phi(x)} \frac{\de x}{\alpha} - 
\int_{x_0+a}^{x_2} \frac{\varsigma^2(x) - \sigma^2}{\phi(x)} \de x\\
& \quad= \int_{x_0}^{x_2} \Big \{ \frac{1}{\alpha} \frac{\varsigma^2(
  \frac{x+\beta}{\alpha}) - \sigma^2}{\phi(x)} - \frac{\varsigma^2(x) -
  \sigma^2}{\phi(x)}\Big\} \; \de x
+ \int_{x_0}^{x_0+a}\frac{\varsigma^2(x) - \sigma^2}{\phi(x)}\; \de x\\
&\quad \stackrel{(a)}{\le} \frac{1}{\sigma^2} \int_{x_0}^{x_2} \Big|
\frac{1}{\alpha} \varsigma^2\big( \frac{x+\beta}{\alpha}\big)  -
\varsigma^2(x)\Big|\;\de x + \Big(1-\frac{1}{\alpha} \Big) \int_{x_0}^{x_2}
\frac{\sigma^2}{\phi(x)} \; \de x + \int_{x_0}^{x_0+a}
\frac{\sigma^2}{\phi(x)}\;\de x\\
&\quad \le \frac{1}{\sigma^2} \int_{x_0}^{x_2}
\left(1-\frac{1}{\alpha}\right)\; \varsigma^2\big(\frac{x+\beta}{\alpha}\big)
\;\de x
+\frac{1}{\sigma^2} \int_{x_0}^{x_2}
\Big|\varsigma^2\big(\frac{x+\beta}{\alpha}\big) - \varsigma^2(x)\Big| \;\de x 
+ \frac{K}{2} \left(1-\frac{1}{\alpha}\right) + a\\
&\quad \le \left(1-\frac{1}{\alpha}\right) K + \frac{1}{\sigma^2}
\int_{x_0}^{x_2} \Big|\varsigma^2\big(\frac{x+\beta}{\alpha}\big) -
\varsigma^2(x)\Big| \;\de x + \frac{K}{2}\left(1-\frac{1}{\alpha}\right) +a \\
&\quad \stackrel{(b)}{\le} \left(1-\frac{1}{\alpha}\right) K +
CK^2\left(1-\frac{1}{\alpha} \right) + C K \, a+ \frac{K}{2}\left(1-\frac{1}{\alpha}\right) +a\\
&\quad \le C(K) a,
\end{split}
\end{eqnarray}
where $(a)$ follows from the fact $\sigma^2 \le \phi(x)$ and Remark~\ref{rem:sigma_prop1}; $(b)$ follows from Remark~\ref{rem:sig_lip}.

Thirdly, recall that $\phi_a(x) = \phi(x-a)$, for $x \in [-1+a,x_0+a)$. Therefore,
\begin{eqnarray}\label{eqn:sigma_term3}
\begin{split}
&\int_{-1+a}^{x_0+a} \Big \{ \frac{\varsigma^2(x) - \sigma^2}{\phi_a(x)}
- \frac{\varsigma^2(x) - \sigma^2}{\phi(x)}\Big\} \; \de x\\
&\quad = \int_{-1}^{x_0} \frac{\varsigma^2(x+a) - \sigma^2}{\phi(x)}\;
\de x 
- \int_{-1+a}^{x_0+a} \frac{\varsigma^2(x) - \sigma^2}{\phi(x)}\; \de x\\
&\quad = \int_{-1}^{x_0} \frac{\varsigma^2(x+a) - \varsigma^2(x)}{\phi(x)}\;
\de x 
- \int_{x_0}^{x_0+a} \frac{\varsigma^2(x) - \sigma^2}{\phi(x)} \;\de x 
+ \int_{-1}^{-1+a} \frac{\varsigma^2(x) - \sigma^2}{\phi(x)} \;\de x \\
&\quad \le 0 + 0 + \int_{-1}^{-1+a}\frac{\sigma^2}{\phi(x)}\;\de x \\
&\quad \le a,
\end{split}
\end{eqnarray}
where the first inequality follows from Remark~\ref{rem:sigma_prop1}
and the second follows from $\phi(x)\ge \sigma^2$.

Finally, using the facts $\sigma^2 \le \varsigma^2(x) \le 2 \sigma^2$,
and $\sigma^2 \le \phi(x)$, we have
\begin{eqnarray}\label{eqn:sigma_term4}
\begin{split}
\int_{-1}^{-1+a} \Big \{ \frac{\varsigma^2(x) - \sigma^2}{\phi_a(x)} -
\frac{\varsigma^2(x) - \sigma^2}{\phi(x)}\Big\} \; \de x
\le a.
\end{split}
\end{eqnarray}

Combining Eqs.~\eqref{eqn:sigma_term1},~\eqref{eqn:sigma_term2},~\eqref{eqn:sigma_term3}, and~\eqref{eqn:sigma_term4} implies the desired result.

\section{Proof of Proposition~\ref{propo:Etilde}}
\label{app:Etilde}
\begin{proof}
Let $\tEnergy_{\Shape}(\phi_a) = \tEnergy_{\Shape,1}(\phi_a) + \tEnergy_{\Shape,2}(\phi_a) + \tEnergy_{\Shape,3}(\phi_a)$, where
\begin{eqnarray}\label{eqn:Ephia_split}
\begin{split}
\tEnergy_{\Shape,1}(\phi_a) &= \int_{x_0+a}^{\ell-1} \{\Info(\Shape
\ast \phi_a(y)^{-1}) - \Info(\phi_a(y-1)^{-1})\} \de y,\\
\tEnergy_{\Shape,2}(\phi_a) &= \int_{a}^{x_0+a} \{\Info(\Shape \ast
\phi_a(y)^{-1}) - \Info(\phi_a(y-1)^{-1})\} \de y,\\
\tEnergy_{\Shape,3}(\phi_a) &= \int_{0}^{a} \{\Info(\Shape \ast
\phi_a(y)^{-1}) - \Info(\phi_a(y-1)^{-1})\} \de y.
\end{split}
\end{eqnarray}
Also let $\tEnergy_{\Shape}(\phi) = \tEnergy_{\Shape,1}(\phi) +  \tEnergy_{\Shape,2,3}(\phi)$, where
\begin{eqnarray}\label{eqn:Ephi_split}
\begin{split}
\tEnergy_{\Shape,1}(\phi) &=  \int_{x_0+a}^{\ell-1} \{\Info(\Shape
\ast \phi(y)^{-1}) - \Info(\phi(y-1)^{-1})\} \de y,\\
\tEnergy_{\Shape,2,3}(\phi) &= \int_{0}^{x_0+a} \{\Info(\Shape \ast
\phi(y)^{-1}) - \Info(\phi(y-1)^{-1})\} \de y.
\end{split}
\end{eqnarray}

The following remark is used several times in the proof.
\begin{remark}\label{rem:ILiptchiz}
For any two values $0 \le \alpha_1 < \alpha_2$,  
\begin{eqnarray}
\Info(\alpha_2) - \Info(\alpha_1)  = \int_{\alpha_1}^{\alpha_2}
\frac{1}{2}\mmse(z) \de z \le \int_{\alpha_1}^{\alpha_2} \frac{1}{2z}
\de z
 = \frac{1}{2}\log\Big(\frac{\alpha_2}{\alpha_1}\Big) \le \frac{1}{2} \left(\frac{\alpha_2}{\alpha_1} - 1 \right).
\end{eqnarray}
\end{remark}
\smallskip

\noindent$\bullet$ Bounding $\tEnergy_{\Shape,1}(\phi_a) - \tEnergy_{\Shape,1}(\phi)$. \\
Notice that the functions $\phi(x)=\phi_a(x)$, for $x_2 \le x$. Also $\kappa/2 < \phi_a(x) \le  \phi(x) \le \Phi_M$, for $x_1< x < x_2$. Let $\alpha = (x_2 - x_1)/(x_2 - x_1 - a)$, and $\beta = (a x_2)/(x_2 - x_1 - a)$. Then, $\phi_a(x) = \phi(\alpha x - \beta)$ for $x \in [x_0+a, x_2)$. Hence,
\begin{align}
&\tEnergy_{\Shape,1}(\phi_a) - \tEnergy_{\Shape,1}(\phi) \nonumber\\
&= \int_{x_0+a}^{x_2+1} \Info(\Shape \ast \phi_a(y)^{-1}) -
\Info(\Shape \ast \phi(y)^{-1}) \; \de y 
+ \int_{x_0+a}^{x_2+1}  \Info( \phi(y-1)^{-1}) -
\Info(\phi_a(y-1)^{-1}) \; \de y \nonumber\\
& \le \frac{1}{2} \int_{x_0+a}^{x_2+1} \frac{1}{\Shape \ast
  \phi(y )^{-1}} (\Shape \ast \phi_a(y )^{-1} - \Shape \ast
\phi(y )^{-1}) \; \de y \nonumber\\
& \le \frac{\Phi_M}{2} \int_{x_0+a}^{x_2+1} \Big(\int_{x_0 + a -
  1}^{x_2} \Shape(y-z) \phi_a(z)^{-1} \; \de z - 
\int_{x_0 + a - 1}^{x_2} \Shape(y-z) \phi(z)^{-1}  \; \de z \Big) \de y \nonumber\\
& = \frac{\Phi_M}{2} \int_{x_0+a}^{x_2+1} \Big(\int_{x_0 + a }^{x_2}
\Shape(y-z) \phi(\alpha z - \beta)^{-1} \; \de z
 + \int_{x_0 + a - 1}^{x_0+a} \Shape(y-z) \phi(z-a)^{-1} \; \de z \nonumber\\
 &\quad \quad \quad \quad \quad \quad \quad - \int_{x_0 + a - 1}^{x_2}
 \Shape(y-z) \phi(z)^{-1} \; \de z \Big) \de y \nonumber\\
& \le \frac{\Phi_M}{2} \int_{x_0+a}^{x_2+1} \Big\{ \int_{x_0 }^{x_2}
\Big(\frac{1}{\alpha} \Shape(y-\frac{z+\beta}{\alpha}) -
\Shape(y-z)\Big) \phi(z)^{-1}\; \de z \nonumber\\
& \quad \quad \quad \quad \quad \quad \quad+ \int_{x_0-1}^{x_0} \Big(
\Shape(y-z-a) - \Shape(y-z)\Big) \phi(z)^{-1}\; \de z \nonumber\\
& \quad \quad \quad \quad \quad \quad \quad + \int_{x_0-1}^{x_0+a-1}
\Shape(y-z) \phi(z)^{-1}\; \de z \Big\} \de y \nonumber\\
& \le \frac{\Phi_M}{2} \int_{x_0+a}^{x_2+1} \Big\{ \int_{x_0 }^{x_2}
\Big(\Shape(y-\frac{z+\beta}{\alpha}) - \Shape(y-z)\Big)
\phi(z)^{-1}\; \de z \nonumber\\
& \quad \quad \quad \quad \quad \quad \quad+ \int_{x_0-1}^{x_0} \Big(
\Shape(y-z-a) - \Shape(y-z)\Big) \phi(z)^{-1}\; \de z \nonumber\\
& \quad \quad \quad \quad \quad \quad \quad + \int_{x_0-1}^{x_0+a-1}
\Shape(y-z) \phi(z)^{-1}\; \de z \Big\} \de y \nonumber\\
& \le C_1(1 - \frac{1}{\alpha} ) + C_2 \frac{\beta}{\alpha} + C_3 \;a \le C_4\; a. \label{eqn:term1_1}
\end{align}
Here $C_1,C_2,C_3,C_4$ are some constants that depend only on $K$ and
$\kappa$. The last step follows from the facts that 
$\Shape(\,\cdot\,)$ is a bounded Lipschitz function and $\phi(z)^{-1} \le 2/\kappa$ for $z \in [x_1,x_2]$. Also, note that in the first inequality, $\Info( \phi(y-1)^{-1}) - \Info(\phi_a(y-1)^{-1}) \le 0$, since $\phi(y-1)^{-1} \le \phi_a(y-1)^{-1}$, and $\Info(\,\cdot\,)$ is nondecreasing.

\smallskip

\noindent$\bullet$ Bounding $\tEnergy_{\Shape,2}(\phi_a) - \tEnergy_{\Shape,2,3}(\phi)$. \\
We have
\begin{eqnarray} \label{eqn:term2_1}
\begin{split}
\tEnergy_{\Shape,2}(\phi_a) &=  \int_{x_0+a-1}^{x_0+a} \{\Info(\Shape
\ast \phi_a(y)^{-1}) - \Info(\phi_a(y-1)^{-1})\} \de y\\
& \quad +  \int_{a}^{x_0+a-1} \{\Info(\Shape \ast \phi_a(y)^{-1}) -
\Info(\phi_a(y-1)^{-1})\} \de y.
\end{split}
\end{eqnarray}
We treat each term separately. For the first term,
\begin{align}
&\int_{x_0+a-1}^{x_0+a} \{\Info(\Shape \ast \phi_a(y)^{-1}) -
\Info(\phi_a(y-1)^{-1})\} \de y \nonumber\\
&= \int_{x_0+a-1}^{x_0+a} \Big\{\Info \left(\int_{x_0+a}^{x_0+a+1}
  \Shape(y-z) \phi_a(z)^{-1} \; \de z +
\int_{x_0+a-2}^{x_0+a} \Shape(y-z) \phi_a(z)^{-1}\; \de z\right) -
\Info(\phi_a(y-1)^{-1})\Big\} \de y \nonumber\\
&= \int_{x_0+a-1}^{x_0+a} \Info \left(\int_{x_0}^{x_0+\alpha}
  \Shape(y-\frac{z+\beta}{\alpha}) \phi(z)^{-1} \; \frac{\de z}{\alpha}
+ \int_{x_0-2}^{x_0} \Shape(y-a-z) \phi(z)^{-1}\; \de z \right)\de y \nonumber\\
& \quad -\int_{x_0-1}^{x_0} \Info (\phi(y-1)^{-1}) \de y \nonumber\\
& = \int_{x_0-1}^{x_0} \Info \left(\int_{x_0}^{x_0+\alpha} \Shape(y+a
  -\frac{z+\beta}{\alpha}) \phi(z)^{-1} \; \frac{\de z}{\alpha}
+ \int_{x_0-2}^{x_0} \Shape(y-z) \phi(z)^{-1}\; \de z \right) \de y \nonumber\\
& \quad -\int_{x_0-1}^{x_0} \Info(\phi(y-1)^{-1}) \de y \nonumber\\
& \le C_5 \; a + \int_{x_0-1}^{x_0} \Big\{\Info
\left(\int_{x_0-2}^{x_0+1} \Shape(y - z) \phi(z)^{-1} \; \de z \right)
- \Info(\phi(y-1)^{-1}) \Big\} \de y \nonumber\\
& =  C_5 \; a + \int_{x_0-1}^{x_0} \Big\{\Info (\Shape \ast
\phi(y)^{-1} ) - \Info(\phi(y-1)^{-1}) \Big\} \de y, \label{eqn:term2_2}
\end{align}
where the last inequality is an application of remark~\ref{rem:ILiptchiz}. More specifically,
\begin{eqnarray*}
\begin{split}
&\Info \left(\int_{x_0}^{x_0+\alpha} \Shape(y+a
  -\frac{z+\beta}{\alpha}) \phi(z)^{-1} \; \frac{\de z}{\alpha} +
\int_{x_0-2}^{x_0} \Shape(y - z) \phi(z)^{-1} \; \de z\right)\\
&\quad - \Info \left(\int_{x_0-2}^{x_0+1} \Shape(y - z) \phi(z)^{-1}
  \; \de z \right)\\
& \le \frac{\Phi_M}{2}\; \left( \int_{x_0}^{x_0+\alpha} \Shape(y+a
  -\frac{z+\beta}{\alpha}) \phi(z)^{-1} \; \frac{\de z}{\alpha} -
  \int_{x_0}^{x_0+1} \Shape(y - z) \phi(z)^{-1} \; \de z \right)\\
& \le \frac{\Phi_M}{2} \int_{x_0+1}^{x_0+\alpha} \Shape(y+a
-\frac{z+\beta}{\alpha}) \phi(z)^{-1} \; \de z \\
& \quad + \frac{\Phi_M}{2} \int_{x_0}^{x_0+1} \left(\Shape(y+a -
  \frac{z+\beta}{\alpha}) - \Shape(y-z) \right) \phi(z)^{-1} \de z\\
&\le C'_1(1 - \frac{1}{\alpha} ) + C'_2 \frac{\beta}{\alpha} + C'_3 \;a \le C_5\; a,  
\end{split}
\end{eqnarray*}
where $C'_1, C'_2,C'_3,C_5$ are constants that depend only on
$\kappa$. Here, the penultimate inequality follows from $\alpha
>1$, and the last one follows from the fact that $\Shape(\,\cdot\,)$ is a bounded Lipschitz function and that $\phi(z)^{-1} \le 2/\kappa$, for $z\in [x_1,x_2]$.

\indent To bound the second term on the right hand side of Eq.~\eqref{eqn:term2_2}, notice that $\phi_a(z) = \phi(z-a)$, for $z \in [-1+a,x_0+a)$, whereby
\begin{eqnarray}\label{eqn:term2_3}
\begin{split}
 \int_{a}^{x_0+a-1} \{\Info(\Shape \ast \phi_a(y)^{-1}) -
 \Info(\phi_a(y-1)^{-1})\} \de y 
  = \int_{0}^{x_0-1} \{\Info(\Shape \ast \phi(y)^{-1}) -
  \Info(\phi(y-1)^{-1})\} \de y.
\end{split}
\end{eqnarray}

\indent Now, using Eqs.~\eqref{eqn:Ephi_split}, \eqref{eqn:term2_1} and~\eqref{eqn:term2_3}, we obtain
\begin{eqnarray}\label{eqn:term2_4}
\begin{split}
\tEnergy_{\Shape,2}(\phi_a) - \tEnergy_{\Shape,2,3}(\phi) &\le C_5\; a  - \int_{x_0}^{x_0+a} \{\Info(\Shape \ast \phi(y)^{-1}) - \Info(\phi(y-1)^{-1})\} \de y\\
&\le C_5\;a +  \int_{x_0}^{x_0+a} \log\left(\frac{\phi(y-1)^{-1}}{\Shape \ast \phi(y)^{-1}}\right)\\
&\le C_5\;a +  a \log(\frac{\Phi_M}{\kappa})  = C_6\;a,
\end{split}
\end{eqnarray}
where $C_6$ is a constant that depends only on $\kappa$.
\smallskip

\noindent$\bullet$ Bounding $\tEnergy_{\Shape,3}(\phi_a).$\\
Notice that $\phi_a(y) \ge
\sigma^2$. Therefore, $\Info(\Shape \ast \phi_a(y)^{-1}) \le
\Info(\sigma^{-2})$, since $\Info(\,\cdot\, )$ is nondecreasing. Recall that $\phi_a(y) = \phi^* < 2\sigma^2$, for $y \in [-1,-1+a)$. Consequently,
\begin{eqnarray} \label{eqn:term3}
\tEnergy_{\Shape,3}(\phi_a) \le \int_{0}^{a} \{\Info(\sigma^{-2}) - \Info({\phi^*}^{-1})\} \de y \le \frac{a}{2} \log\Big(\frac{\phi^*}{\sigma^2}\Big)
 < \frac{a}{2} \log 2, 
\end{eqnarray} 
where the first inequality follows from Remark~\ref{rem:ILiptchiz}.

Finally, we are in position to prove the proposition. Using Eqs.~\eqref{eqn:term1_1},~\eqref{eqn:term2_4} and~\eqref{eqn:term3}, we get
\begin{eqnarray}
\tEnergy_{\Shape}(\phi_a) - \tEnergy_{\Shape}(\phi) \le C_4\;a + C_6\; a +  \frac{a}{2} \log 2
 = C(\kappa,K)\; a .
\end{eqnarray}

\end{proof}

\section{Proof of Proposition~\ref{propo:potential_Etilde}}
\label{app:potential_Etilde}

We have
\begin{eqnarray}\label{eqn:pert_V}
\begin{split}
\int_{-1}^{\ell-1} \big\{V(\phi_a(x)) - V(\phi(x)) \big\} \de x & = 
\int_{x_2}^{\ell-1} \big\{V(\phi_a(x)) - V(\phi(x)) \big\} \de x\\ 
&+ \Big( \int_{x_0+a}^{x_2} V(\phi_a(x)) \de x- \int_{x_0}^{x_2} V(\phi(x)) \de x \Big) \\
& +\Big( \int_{-1+a}^{x_0+a} V(\phi_a(x)) \de x - \int_{-1}^{x_0} V(\phi(x)) \de x  \Big) \\
& +\int_{-1}^{-1+a} V(\phi_a(x))\de x.  
\end{split}
\end{eqnarray}
Notice that the first and the third terms on the right hand side are zero. Also,
\begin{eqnarray}\label{eqn:2-4terms}
\begin{split}
 \int_{x_0+a}^{x_2} V(\phi_a(x)) \de x - \int_{x_0}^{x_2} V(\phi(x)) \de x
 &=  -\frac{a}{x_2-x_0} \int_{x_0}^{x_2} V(\phi(x))\de x,\\
 \int_{-1}^{-1+a} V(\phi_a(x)) \de x  &= a V(\phi^*).
\end{split}
\end{eqnarray}
Substituting Eq.~\eqref{eqn:2-4terms} in Eq.~\eqref{eqn:pert_V}, we get
\begin{eqnarray}\label{eqn:temp_rev}
\int_{-1}^{\ell-1} \big\{V(\phi_a(x)) - V(\phi(x)) \big\} \de x = \frac{a}{x_2-x_0} \int_{x_0}^{x_2} \big\{ V(\phi^*) - V(\phi(x)) \big \}\de x.
\end{eqnarray}

Now we upper bound the right hand side of Eq.~\eqref{eqn:temp_rev}.

By Proposition~\ref{propo:V_properties}, we have
\begin{eqnarray} \label{eqn:V(mu*)_1}
V(\phi^*) \le \frac{\delta}{2} + \frac{\delta - \uRenyi(p_X)}{4} \log(2\sigma^2),
\end{eqnarray}
for $\sigma \in (0,\sigma_2]$. Also, since $\phi(x) >\kappa/2$ for $x \in [x_0,x_2]$, we have $V(\phi(x)) \ge (\delta/2) \log \phi > (\delta/2) \log (\kappa/2)$. Therefore,
\begin{eqnarray}\label{eq:Vphistar-Vphi}
\begin{split}
\frac{1}{2}\int_{-1}^{\ell-1} \big\{V(\phi_a(x)) - V(\phi(x)) \big\} \de x  &= 
\frac{a}{2(x_2-x_0)} \int_{x_0}^{x_2} \big\{ V(\phi^*) - V(\phi(x)) \big \}\de x \\
&< \frac{a}{2} \Big[\frac{\delta}{2} + \frac{\delta - \uRenyi(p_X)}{4} \log(2\sigma^2)
 - \frac{\delta}{2} \log(\frac{\kappa}{2}) \Big].
\end{split}
\end{eqnarray}
It is now obvious that by choosing $\sigma_0 > 0$ small enough, we can ensure that for values $\sigma \in (0,\sigma_0]$, 
\begin{eqnarray}\label{eqn:V(mu*)_2}
\frac{a}{2} \Big[\frac{\delta}{2} + \frac{\delta - \uRenyi(p_X)}{4} \log(2\sigma^2)
 - \frac{\delta}{2} \log(\frac{\kappa}{2}) \Big] < -2C(\kappa,K) a. 
\end{eqnarray}
(Notice that the right hand side of Eq.~\eqref{eqn:V(mu*)_2} does not depend on $\sigma$).

%
\section{Proof of Claim~\ref{claim:slope_phi2}}
\label{app:slope_phi2}
Similar to the proof of Claim~\ref{claim:slope_phi}, the assumption $\int_{-1}^{\ell-1} |\phi(x) - \phi^*| \de x > C\sigma^2 \ell$ implies $\phi(\theta \ell - 1) > C\sigma^2(1-\alpha)$, where 
\begin{equation*}
0 < \theta = \frac{\Phi_M - C\sigma^2}{\Phi_M - {C\sigma^2}{(1-\alpha)}} < 1.
\end{equation*} 

Choose $\sigma$ small enough such that $\phi^* < \phi_1$. Let $\kappa = (\phi_1-\phi^*)(1-\theta)/2$. Applying Lemma~\ref{lem:main_continuum}, there exists $\ell_0$, and $\sigma_0$, such that, $\int_{-1}^{\ell-1} |\phi(x)-\phi^*|\;\de x \le \kappa \ell$, for $\ell > \ell_0$ and $\sigma \in (0,\sigma_0]$. We claim that $\phi(\mu \ell-1) < \phi_1$,
with
\begin{equation*}
\mu = 1 - \frac{\kappa}{\phi_1 - \phi^*} = \frac{1+\theta}{2}.
\end{equation*}
 
 \noindent Otherwise, by monotonicity of $\phi(x)$,
 \begin{eqnarray}
 \begin{split}
 (\phi_1 - \phi^*)  (1-\mu) \ell \le  \int_{\mu\ell-1}^{\ell-1} |\phi(x) - \phi^*|\;\de x < \int_{-1}^{\ell-1} |\phi(x) - \phi^*|\;\de x \le \kappa \ell.
 \end{split}
 \end{eqnarray}
 Plugging in for $\mu$ yields a contradiction.
 
 Therefore, $C\sigma^2(1-\alpha) < \phi(x) < \phi_1$, for $x \in [\theta\ell-1, \mu\ell-1]$, and $(\mu-\theta)\ell = (1-\theta)\ell/2$. Choosing $\ell > \max\{\ell_0, 2K/(1-\theta)\}$ gives the result.

\section{Proof of Proposition~\ref{pro:R_terms}}
\label{app:R_terms}
To prove Eq.~\eqref{eqn:R_terms1}, we write
\begin{eqnarray}
\begin{split}
\int_{-1}^{\ell-1} \{V_{\rob}(\phi_a(x)) - V_{\rob}(\phi(x))\}\;\de x &= -\int_{x_1}^{x_2} \int_{\phi_a(x)}^{\phi(x)} V'(s)\;\de s\; \de x\\
&\le -\int_{x_1}^{x_2} \int_{\phi_a(x)}^{\phi(x)} \frac{\delta}{2s^2} \left(s-\sigma^2 \right)\;\de s\; \de x\\
&= -\frac{\delta}{2}\int_{x_1}^{x_2} \Big\{\log\left(\frac{\phi(x)}{\phi_a(x)}\right) + \frac{\sigma^2}{\phi(x)}- \frac{\sigma^2}{\phi_a(x)} \Big\} \; \de x\\
&\le \frac{\delta}{2} K \log(1-a) + K \frac{\delta a}{2C(1-\alpha)(1-a)} ,
\end{split}
\end{eqnarray}
where the second inequality follows from the fact $C\sigma^2/2 < \phi(x)$, for $x \in [x_1,x_2]$.

Next, we pass to prove Eq.~\eqref{eqn:R_terms2}.
\begin{eqnarray}
\begin{split}
\int_{-1}^{\ell-1}(\varsigma^2(x) - \sigma^2)\left(\frac{1}{\phi_a(x)} - \frac{1}{\phi(x)}\right)\;\de x 
&= \int_{x_1}^{x_2} \frac{\varsigma^2(x) - \sigma^2}{\phi(x)} \left( \frac{1}{1-a} - 1\right)\\ 
&\le \frac{a}{1-a} \int_{x_1}^{x_2} \frac{\sigma^2}{\phi(x)}\de x \le K \frac{a}{C(1-\alpha)(1-a)},
\end{split}
\end{eqnarray}
where the first inequality follows from Remark~\ref{rem:sigma_prop1}.

Finally, we have
\begin{eqnarray} 
\begin{split}
\tEnergy_{\Shape,\rob}(\phi_a) - \tEnergy_{\Shape,\rob}(\phi)
& = \int_{0}^{\ell} \{\Info(\Shape \ast \phi_a(y)^{-1}) - \Info(\Shape \ast \phi(y)^{-1})\} \de y \\
& = \int_{0}^{\ell} \int_{\Shape \ast \phi(y)^{-1}}^{\Shape \ast \phi_a(y)^{-1}} \frac{1}{2}\mmse(s)\;\de s\;\de y\\
&\le \frac{\uMMSE(p_X) + \ve}{2} \int_{0}^{\ell}\int_{\Shape \ast \phi(y)^{-1}}^{\Shape \ast \phi_a(y)^{-1}}s^{-1}\de s\;\de y\\
&\le \frac{\uMMSE(p_X) + \ve}{2} \int_{0}^{\ell} \log\left( \frac{\Shape \ast \phi_a(y)^{-1}}{\Shape \ast \phi(y)^{-1}}\right)\;\de y\\
&\le -\frac{\uMMSE(p_X) + \ve}{2}(K+2) \log(1-a),
\end{split}
\end{eqnarray}
where the first inequality follows from Eq.~\eqref{eq:mmse_approx} and Claim~\ref{claim:slope_phi2}.

\bibliographystyle{amsalpha}

\bibliography{all-bibliography}

\newcommand{\etalchar}[1]{$^{#1}$}
\providecommand{\bysame}{\leavevmode\hbox to3em{\hrulefill}\thinspace}
\providecommand{\MR}{\relax\ifhmode\unskip\space\fi MR }
\providecommand{\MRhref}[2]{%
  \href{http://www.ams.org/mathscinet-getitem?mr=#1}{#2}
}
\providecommand{\href}[2]{#2}
\begin{thebibliography}{MMRU09}

\bibitem[ASZ10]{Saligrama}
S.~Aeron, V.~Saligrama, and Manqi Zhao, \emph{Information theoretic bounds for
  compressed sensing}, IEEE Trans. on Inform. Theory \textbf{56} (2010), 5111
  -- 5130.

\bibitem[BGI{\etalchar{+}}08]{Indyk}
R.~Berinde, A.C. Gilbert, P.~Indyk, H.~Karloff, and M.J. Strauss,
  \emph{{Combining geometry and combinatorics: A unified approach to sparse
  signal recovery}}, 47th Annual Allerton Conference (Monticello, IL),
  September 2008, pp.~798 -- 805.

\bibitem[BIPW10]{IndykLower}
K.~Do Ba, P.~Indyk, E.~Price, and D.~P. Woodruff, \emph{Lower bounds for sparse
  recovery}, Proceedings of the Twenty-First Annual ACM-SIAM Symposium on
  Discrete Algorithms, SODA '10, 2010, pp.~1190--1197.

\bibitem[BLM12]{BM-Universality}
M.~Bayati, M.~Lelarge, and A.~Montanari, \emph{{Universality in Polytope Phase
  Transitions and Message Passing Algorithms}}, {\sf arXiv:1207.7321v1}, 2012.

\bibitem[BM11]{BM-MPCS-2011}
M.~Bayati and A.~Montanari, \emph{{The dynamics of message passing on dense
  graphs, with applications to compressed sensing}}, IEEE Trans. on Inform.
  Theory \textbf{57} (2011), 764--785.

\bibitem[BM12]{BayatiMontanariLASSO}
\bysame, \emph{{The LASSO risk for gaussian matrices}}, IEEE Trans. on Inform.
  Theory \textbf{58} (2012), 1997--2017.

\bibitem[BSB10]{BaronBayesian}
D.~Baron, S.~Sarvotham, and R.~Baraniuk, \emph{{Bayesian Compressive Sensing
  Via Belief Propagation}}, IEEE Trans. on Signal Proc. \textbf{58} (2010),
  269--280.

\bibitem[CD11]{CandesDavenport}
E.~Cand\'es and M.~Davenport, \emph{{How well can we estimate a sparse
  vector?}}, {\sf arXiv:1104.5246v3}, 2011.

\bibitem[CRT06a]{CandesFourier}
E.~Candes, J.~K. Romberg, and T.~Tao, \emph{{Robust uncertainty principles:
  Exact signal reconstruction from highly incomplete frequency information}},
  IEEE Trans. on Inform. Theory \textbf{52} (2006), 489 -- 509.

\bibitem[CRT06b]{CandesStable}
\bysame, \emph{{Stable signal recovery from incomplete and inaccurate
  measurements}}, Communications on Pure and Applied Mathematics \textbf{59}
  (2006), 1207--1223.

\bibitem[CT05]{CandesTao}
E.~J. Cand\'es and T.~Tao, \emph{{Decoding by linear programming}}, IEEE Trans.
  on Inform. Theory \textbf{51} (2005), 4203--4215.

\bibitem[DJM11]{DonohoJohnstoneMontanari}
D.~Donoho, I.~Johnstone, and A.~Montanari, \emph{{Accurate Prediction of Phase
  Transitions in Compressed Sensing via a Connection to Minimax Denoising}},
  {\sf arXiv:1111.1041}, 2011.

\bibitem[DMM09]{DMM09}
D.~L. Donoho, A.~Maleki, and A.~Montanari, \emph{{Message Passing Algorithms
  for Compressed Sensing}}, Proceedings of the National Academy of Sciences
  \textbf{106} (2009), 18914--18919.

\bibitem[DMM10]{DMM_ITW_I}
\bysame, \emph{{Message Passing Algorithms for Compressed Sensing: I.
  Motivation and Construction}}, Proceedings of IEEE Inform. Theory Workshop
  (Cairo), 2010.

\bibitem[DMM11]{DMM-NSPT-11}
D.L. Donoho, A.~Maleki, and A.~Montanari, \emph{{The Noise Sensitivity Phase
  Transition in Compressed Sensing}}, IEEE Trans. on Inform. Theory \textbf{57}
  (2011), 6920--6941.

\bibitem[Don06a]{Donoho1}
D.~L. Donoho, \emph{Compressed sensing}, IEEE Trans. on Inform. Theory
  \textbf{52} (2006), 1289--1306.

\bibitem[Don06b]{Do}
D.~L. Donoho, \emph{{High-dimensional centrally symmetric polytopes with
  neighborliness proportional to dimension}}, Discrete Comput. Geometry
  \textbf{35} (2006), 617--652.

\bibitem[DT05]{DoTa05}
D.~L. Donoho and J.~Tanner, \emph{Neighborliness of randomly-projected
  simplices in high dimensions}, Proceedings of the National Academy of
  Sciences \textbf{102} (2005), no.~27, 9452--9457.

\bibitem[DT10]{DoTa10b}
D.~L. Donoho and J.~Tanner, \emph{Counting the faces of randomly-projected
  hypercubes and orthants, with applications}, Discrete {\&} Computational
  Geometry \textbf{43} (2010), no.~3, 522--541.

\bibitem[FZ99]{Felstrom}
A.J. Felstrom and K.S. Zigangirov, \emph{Time-varying periodic convolutional
  codes with low-density parity-check matrix}, IEEE Trans. on Inform.~Theory
  \textbf{45} (1999), 2181--2190.

\bibitem[Gal63]{GallagerThesis}
R.~G. Gallager, \emph{Low-density parity-check codes}, MIT Press, Cambridge,
  Massachussetts, 1963, Available online at {\tt
  http://web./gallager/www/pages/ldpc.pdf}.

\bibitem[GSV05]{Guo05mutualinformation}
D.~Guo, S.~Shamai, and S.~Verd\'u, \emph{Mutual information and minimum
  mean-square error in gaussian channels}, IEEE Trans. Inform. Theory
  \textbf{51} (2005), 1261--1282.

\bibitem[HMU10]{HassaniEtAl}
S.H. Hassani, N.~Macris, and R.~Urbanke, \emph{{Coupled graphical models and
  their thresholds}}, Proceedings of IEEE Inform. Theory Workshop (Dublin),
  2010.

\bibitem[HR09]{HuberBook}
P.J. Huber and E.~Ronchetti, \emph{Robust statistics (second edition)}, J.
  Wiley and Sons, 2009.

\bibitem[IPW11]{IndykAdaptive}
P.~Indyk, E.~Price, and D.P. Woodruff, \emph{{On the Power of Adaptivity in
  Sparse Recovery}}, IEEE Symposium on the Foundations of Computer Science,
  FOCS, October 2011.

\bibitem[JM12a]{JM-StateEvolution}
A.~Javanmard and A.~Montanari, \emph{{State Evolution for General Approximate
  Message Passing Algorithms, with Applications to Spatial Coupling}}, {\sf
  arXiv:1211.5164}, 2012.

\bibitem[JM12b]{JavanmardMon12}
\bysame, \emph{{Subsampling at information theoretically optimal rates}}, IEEE
  Intl. Symp. on Inform. Theory (ISIT) (Cambridge), July 2012, pp.~2431--2435.

\bibitem[KGR11]{RanganQuantized}
U.S. Kamilov, V.K. Goyal, and S.~Rangan, \emph{{Message-Passing Estimation from
  Quantized Samples}}, {\sf arXiv:1105.6368}, 2011.

\bibitem[KMRU10]{KudekarBMS}
S.~Kudekar, C.~Measson, T.~Richardson, and R.~Urbanke, \emph{{Threshold
  Saturation on BMS Channels via Spatial Coupling}}, Proceedings of the
  International Symposium on Turbo Codes and Iterative Information Processing
  (Brest), 2010.

\bibitem[KMS{\etalchar{+}}11]{KrzakalaEtAl}
F.~Krzakala, M.~M\'ezard, F.~Sausset, Y.~Sun, and L.~Zdeborova,
  \emph{{Statistical physics-based reconstruction in compressed sensing}}, {\sf
  arXiv:1109.4424}, 2011.

\bibitem[KP10]{KudekarPfister}
S.~Kudekar and H.D. Pfister, \emph{The effect of spatial coupling on
  compressive sensing}, 48th Annual Allerton Conference, 2010, pp.~347 --353.

\bibitem[KRU11]{KudekarBEC}
S.~Kudekar, T.~Richardson, and R.~Urbanke, \emph{{Threshold Saturation via
  Spatial Coupling: Why Convolutional LDPC Ensembles Perform So Well over the
  BEC}}, IEEE Trans. on Inform. Theory \textbf{57} (2011), 803--834.

\bibitem[KRU12]{kudekar2012spatially}
S.~Kudekar, T.~Richardson, and R.~Urbanke, \emph{Spatially coupled ensembles
  universally achieve capacity under belief propagation}, Information Theory
  Proceedings (ISIT), 2012 IEEE International Symposium on, IEEE, 2012,
  pp.~453--457.

\bibitem[KT07]{KashinTemlyakov}
B.S. Kashin and V.N. Temlyakov, \emph{A remark on compressed sensing},
  Mathematical Notes \textbf{82} (2007), 748--755.

\bibitem[LF10]{Lentmaier}
M.~Lentmaier and G.~P. Fettweis, \emph{{On the thresholds of generalized LDPC
  convolutional codes based on protographs}}, IEEE Intl. Symp. on Inform.
  Theory (Austin, Texas), August 2010.

\bibitem[MMRU09]{MMRU05}
C.~M{\'e}asson, A.~Montanari, T.~Richardson, and R.~Urbanke, \emph{{The
  Generalized Area Theorem and Some of its Consequences}}, IEEE Trans. on
  Inform. Theory \textbf{55} (2009), no.~11, 4793--4821.

\bibitem[Mon12]{MontanariChapter}
A.~Montanari, \emph{Graphical models concepts in compressed sensing},
  Compressed Sensing (Y.C. Eldar and G.~Kutyniok, eds.), Cambridge University
  Press, 2012.

\bibitem[Ran11]{RanganGAMP}
S.~Rangan, \emph{{Generalized Approximate Message Passing for Estimation with
  Random Linear Mixing}}, IEEE Intl. Symp. on Inform. Theory (St. Perersbourg),
  August 2011, pp.~2168 -- 2172.

\bibitem[R{\'e}n59]{Renyi}
A.~R{\'e}nyi, \emph{{On the dimension and entropy of probability
  distributions}}, Acta Mathematica Hungarica \textbf{10} (1959), 193--215.

\bibitem[RU08]{RiU08}
T.J. Richardson and R.~Urbanke, \emph{{Modern Coding Theory}}, Cambridge
  University Press, Cambridge, 2008.

\bibitem[RWY09]{WainwrightEllP}
G.~Raskutti, M.~J. Wainwright, and B.~Yu, \emph{{Minimax rates of estimation
  for high-dimensional linear regression over $\ell_q$-balls}}, 47th Annual
  Allerton Conference (Monticello, IL), September 2009.

\bibitem[Sch10]{SchniterTurbo}
P.~Schniter, \emph{{Turbo Reconstruction of Structured Sparse Signals}},
  Proceedings of the Conference on Information Sciences and Systems
  (Princeton), 2010.

\bibitem[Sch11]{SchniterOFDM}
\bysame, \emph{{A message-passing receiver for BICM-OFDM over unknown
  clustered-sparse channels}}, {\sf arXiv:1101.4724}, 2011.

\bibitem[SLJZ04]{Costello}
A.~Sridharan, M.~Lentmaier, D.~J.~Costello Jr, and K.~S. Zigangirov,
  \emph{{Convergence analysis of a class of LDPC convolutional codes for the
  erasure channel}}, 43rd Annual Allerton Conference (Monticello, IL),
  September 2004.

\bibitem[SPS10]{SchniterTree}
S.~Som, L.C. Potter, and P.~Schniter, \emph{{Compressive Imaging using
  Approximate Message Passing and a Markov-Tree Prior}}, Proc. Asilomar Conf.
  on Signals, Systems, and Computers, November 2010.

\bibitem[Vil08]{villani2008optimal}
C.~Villani, \emph{Optimal transport: old and new}, vol. 338, Springer, 2008.

\bibitem[VS11]{SchniterEM}
J.~Vila and P.~Schniter, \emph{Expectation-maximization bernoulli-gaussian
  approximate message passing}, Proc. Asilomar Conf. on Signals, Systems, and
  Computers (Pacific Grove, CA), 2011.

\bibitem[Wai09]{Wainwright2009}
M.J. Wainwright, \emph{Information-theoretic limits on sparsity recovery in the
  high-dimensional and noisy setting}, IEEE Trans. on Inform. Theory
  \textbf{55} (2009), 5728--5741.

\bibitem[WV10]{WuVerdu}
Y.~Wu and S.~Verd{\'u}, \emph{{R\'enyi Information Dimension: Fundamental
  Limits of Almost Lossless Analog Compression}}, IEEE Trans. on Inform. Theory
  \textbf{56} (2010), 3721--3748.

\bibitem[WV11a]{WuVerduMMSE}
\bysame, \emph{{MMSE dimension}}, IEEE Trans. on Inform. Theory \textbf{57}
  (2011), no.~8, 4857--4879.

\bibitem[WV11b]{WuVerduNoisy}
\bysame, \emph{{Optimal Phase Transitions in Compressed Sensing}}, {\sf
  arXiv:1111.6822}, 2011.

\end{thebibliography}

\end{document}